\documentclass[runningheads]{llncs}
\usepackage[T1]{fontenc}
\usepackage{hchang}
\usepackage{tda}
\usepackage{wadsheader}
\usepackage{graphicx}

\begin{document}
\maketitle              % typeset the header of the contribution
\begin{abstract}
 %!TEX TS-program = pdflatex
%!TEX encoding = UTF-8 Unicode
%!TeX spellcheck = en-US
%!TeX root = ..\socg23.tex

%----------------------- Abstract-------------------------------------

Let $\gamma$ be a generic closed curve in the plane.
Samuel Blank, in his 1967 Ph.D.\ thesis, determined if $\gamma$ is self-overlapping by
\emph{geometrically} constructing a combinatorial word from $\gamma$.
% by drawing
%``cables'' from each face to a point at infinity and tracing the intersection points between the cables and $\gamma$.
More recently, Zipei Nie, in an unpublished manuscript, computed the minimum homotopy
area of $\gamma$ by 
constructing a combinatorial word \emph{algebraically}. % by interpreting the faces of $\gamma$ as the generators of the fundamental group of a multi-punctured plane.
We provide a unified framework for working with both words and determine the settings under 
which Blank's word and Nie's word are equivalent.
Using this equivalence, we give a new geometric proof for
the correctness of Nie's algorithm.
Unlike previous work, our proof is constructive which allows us to naturally compute the
actual homotopy that realizes the minimum area.  Furthermore, we contribute to the
theory of self-overlapping curves by providing the first polynomial-time
algorithm to compute a self-overlapping decomposition of any closed curve
$\gamma$ with minimum area.

\end{abstract}
%
%

%!TEX TS-program = pdflatex
%!TEX encoding = UTF-8 Unicode
%!TeX spellcheck = en-US
%!TeX root = ..\socg23.tex

%----------------------- Introduction -------------------------------------
\section{Introduction}
\label{sec:intro}

% \begin{TODO}
%     At the moment the introduction was for STACS version that focussed on the
%     representations of curves.
%     Rewrite the introduction and start with minimum area homotopy, then how
%     curve representations and SOD comes into play.
% \end{TODO}

A \EMPH{closed curve} in the plane is a continuous map $\gamma$ from the circle
$\Sp^1$ to the plane~$\R^2$.
In the plane, any closed curve is homotopic to a point.
A homotopy that sweeps out the minimum possible area
is a \EMPH{minimum homotopy}.
Chambers and Wang~\cite{cw2013} introduced the minimum
homotopy area between two simple homotopic curves
with common endpoints as a way to measure the similarity
between the two curves.
They suggest that homotopy area is more robust against noise
than another popular similarity measure on curves called the \emph{Fr\'echet distance}.
However, their algorithm requires that each curve be simple,
which is restrictive.

Fasy, Karako\c{c}, and Wenk~\cite{fkw2017} proved that
the problem of finding the minimum homotopy area is easy on a closed curve that is the boundary of an immersed disk.
Such curves are called \EMPH{self-overlapping}~\cite{evansFasyWenk,so-graphics,mukherjee2014,shor-van-wyk,titus,whitney1937}.
They also established a tight connection between minimum-area homotopy and self-overlapping curves by showing that any generic closed curve can be decomposed at some vertices
into self-overlapping subcurves such that the combined homotopy from the subcurves is minimum.
This structural result gives an exponential-time algorithm for
the minimum homotopy area problem by testing each decomposition in a brute-force manner.

Nie, in an unpublished manuscript~\cite{nie2014}, described a polynomial-time algorithm to
determine the minimum homotopy area of any closed curve in the plane.
Nie's algorithm borrows tools from geometric group theory by representing the curve as a word in the fundamental group $\pi_1(\curve)$, and connects minimum homotopy area to the \emph{cancellation norms}~\cite{cancellation-norm-def,bringmann_truly_2019,folding-1980} of the word, which can be computed using a dynamic program.
However, the algorithm does not naturally compute an associated \emph{minimum-area homotopy}.
%However, such strategy does not immediately imply a polynomial-time algorithm to compute an associated \emph{minimum-area homotopy}, as the geometric objects presented in the dynamic program naturally grows exponentially in complexity.
%\hsien{Does our Dehn twists produce words exponential in length?  Will it show up when we compute the homotopy?}
%\hsien{No, in SO subcurve any positive folding gives homotopy of linear length; we can compute min-area, then min-area SOD, then positive folding per SO subcurve, then combine.}

%\todo{Discuss difficulty with geometric objects and connect to the discussion of curve representations.}

Alternatively, one can interpret the words from the dynamic program \emph{geometrically} as crossing sequences by traversing any subcurve cyclicly and recording the crossings along with their directions with a collection of nicely-drawn \emph{cables} from each face to a point at infinity.
Such geometric representation is known as the Blank words~\cite{blank,poe-eic1-1968}.
In fact, the first application of these combinatorial words given by Blank is an algorithm that determines if a curve is self-overlapping.
Blank words are geometric in nature and thus the associated objects are polynomial in size.
When attempting to interpret Nie's dynamic program from the geometric view, one encounters the question of how to extend Blank's definition of cables to \emph{subcurves}, where the cables inherited from the original curve are no longer positioned well with respect to the subcurves.
To our knowledge, no geometric interpretation of the dynamic program is known.

\subsection{Our Contributions}

%In this work, we provide the first polynomial-time algorithm to compute a minimum homotopy of any closed curve in the plane.
%
%To do so, 
We first show that Blank and Nie's word constructions are, in fact, equivalent under the right assumptions (Section~\ref{sec:curve-to-word}).
Next, we extend the definition of Blank's word to subcurves and arbitrary cable drawings (Section~\ref{SS:norm}), and interpret the dynamic program by Nie geometrically (Section~\ref{SS:correctness}).
Using the self-overlapping decomposition theorem by Fasy, Karako\c{c}, and Wenk~\cite{fkw2017} we provide a correctness proof to the algorithm.
%The strength of our approach is that the geometric interpretation allows us to compute an \emph{actual} minimum-area homotopy in polynomial time.
Finally, we conclude with a new result that a minimum-area self-overlapping decomposition
can be found in polynomial~time.
We emphasize that extending Blank words to allow arbitrary cables is in no way straightforward.
In fact, many assumptions on the cables have to be made in order to connect self-overlapping curves and minimum-area homotopy; handling arbitrary cable systems, as seen in the dynamic program, requires further tools from geometric topology like \emph{Dehn twists}.% (Section~\ref{SS:independence}).

\section{Background}
\label{sec:background}

In this section, we introduce concepts and definitions that are used throughout
the paper.
We assume the readers are familiar with the basic terminology for~curves and
surfaces.

\subsection{Curves and Graphs}

A \EMPH{closed curve} in the plane is a continuous map
$\EMPH{$\gamma$}: \mathbb{S}^1 \to \R^2$, and a \EMPH{path} in the plane is a
continuous map $\EMPH{$\zeta$}: [0,1]\to \R^2$.
A path $\zeta$ is \emph{closed} when~$\zeta(0)=\zeta(1)$.
In this work, we are presented with a \EMPH{generic} curve; that is, one
where there are a finite number of self-intersections, each of
which is transverse and no three strands~cross at the same point.
%In other words, there are exactly two subpaths of the curve at each self-crossing,
%and the corresponding two tangent vectors must be linearly independent. 
See \figref{generic} for an example.
\setlength{\intextsep}{0pt}%
\begin{wrapfigure}{r}{.3\textwidth}
    \begin{center}
    \includegraphics[width=.25\textwidth]{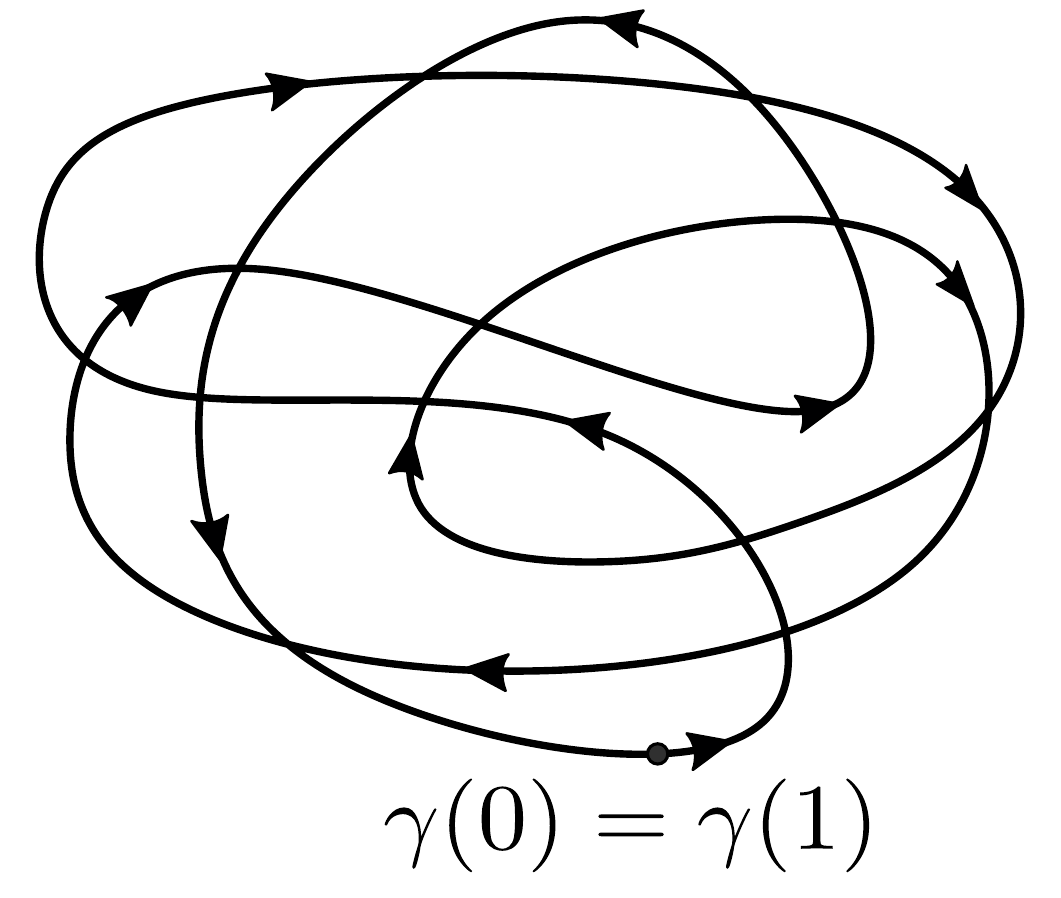}
    \end{center}
    \caption{A generic plane curve induces a four-regular graph.}
    \label{fig:generic}
\end{wrapfigure}
\setlength{\intextsep}{12pt}%
The image of a generic closed curve is naturally associated with a four-regular plane graph.
%we abuse notation and denote this graph as {$\curve$} as well.
The self-intersection points of a curve are \EMPH{vertices},
the paths between vertices are \EMPH{edges}, and the connected
components of the complement of the curve are \EMPH{faces}.
Given a curve, choose an arbitrary starting point~$\gamma(0)=\gamma(1)$ and
orientation for $\curve$.

The \EMPH{dual graph $\curve^*$} is another (multi-)graph, whose
vertices represent the faces of $\gamma$, and two vertices in $\curve^*$ are joined by an edge if there is an edge between the two corresponding faces in $\gamma$.
The dual graph is another plane graph with an inherited embedding~from~$\gamma$.
% by choosing
% an arbitrary point $f^*$ in each face, then for every edge~$e$
% choose a path~$e^*$ between the two points in the faces incident to~$e$
% such that~$e^*$ intersects~$e$ once transversely and~$e^*$ does not intersect
% any other edges of $\curve$.
% The dual graph is another plane graph. \hsien{Change definition to disconnect from drawings.}

%\paragraph*{Tree-Cotree Pairs}
Let $T$ be a spanning tree of $\curve$.
Let $E$ denote the set of edges in $\curve$, the tree~$T$ partitions $E$ into two subsets,
 $T$ and $T^* \coloneqq E
\setminus T$.
The edges in $T^*$ define a spanning tree of $\curve^*$ called the \EMPH{cotree}.
The partition of the edges $(T,T^*)$ is called the \EMPH{tree-cotree pair}.

We call a rooted spanning cotree $T^*$ of $\curve^*$ a \EMPH{breadth-first search tree}
(BFS-tree) if it can be generated from a breadth-first search rooted at the
vertex in $\curve^*$ corresponding to the unbounded face in $\curve$.
Each bounded face $f$ of $\curve$ is a vertex in a breadth-first search tree $T^*$, we associate $f$
with the unique edge incident to $f^*$ in the direction of the root.
Thus, there is a correspondence between edges of $T^*$ and faces of $\curve$.

\subsection{Homotopy and Isotopy}

A \EMPH{homotopy} between two closed curves $\gamma_1$ and $\gamma_2$ that
share a point $p_0$ is a continuous map $H\colon [0,1]\times \Sp^1 \to \mathbb{R}^2$
such that $H(0,\cdot)=\gamma_1$, $H(1,\cdot)=\gamma_2$, and $H(s,0)=p_0=H(s,1)$.
%\footnote{This is known as \EMPH{free homotopy} in standard topology textbook.}
%
We define a homotopy between two paths similarly, where the two endpoints
are fixed throughout the continuous morph.
Notice that homotopy between two closed curves as \emph{closed curves} and the
homotopy between them as \emph{closed paths} with an identical starting points are different.
%We restrict ourselves to $\mathbb{R}^2$ where any two closed curves are homotopic.
%
A homotopy between two injective paths $\zeta_1$ and $\zeta_2$ is an \EMPH{isotopy} if every intermediate 
path~$H(s,\cdot)$ is injective for all $s$.
%in $\R^2$ is
% a homotopy $H\colon [0,1]\times [0,1] \to \R^2$ with $H(0,\cdot)=\zeta_1,
%  H(1,\cdot)=\zeta_2$ such that $H(t,\cdot)$ is injective for all $t\in[0,1]$.
The notion of isotopy naturally extends to a collection of paths.

We can think of $\curve$ as a topological space and consider the \EMPH{fundamental
group~$\pi_1(\curve)$}.
Elements of the fundamental group are called \EMPH{words}, %\hsien{might be overloading}
whose letters correspond to equivalence classes of homotopic closed paths in $\curve$.
%By proposition 1A.2 in \cite{hatcher},   % I don't think this needs citation
The fundamental group of $\curve$ is a free group with basis consisting of the classes
corresponding to the cotree edges of any tree-cotree pair of $\curve$.

Let $H$ be a homotopy between curves $\gamma_1$ and $\gamma_2$.
Let $\EMPH{$\#H^{-1}(x)$} \colon \R^2 \to \Z$ be the function that assigns to each $x\in\R^2$
the number of times the intermediate curves $H$ sweep over $x$.
The homotopy area of $H$ is
\[
    \EMPH{$\Area(H)$} \coloneqq \int_{\R^2} \#H^{-1}(x) \, dx.
\]

The minimum area homotopy between $\gamma_1$ and $\gamma_2$ is the infimum of
the homotopy area over all homotopies between between $\gamma_1$ and $\gamma_2$.
We denote this by
\(
    \EMPH{$\Area_H(\gamma_1,\gamma_2)$} \coloneqq \inf_{H} \, \Area(H).
\)
When $\gamma_2$ is the constant curve at a specific point $p_0$ on $\gamma_1$,
define $\EMPH{$\Area_H(\gamma)$} \coloneqq \Area_H(\gamma,p_0)$.
See \figref{eight} for an example of~a homotopy.

\begin{figure}[htb]
    \captionsetup[subfigure]{justification=centering}
    \centering
    \begin{subfigure}[b]{0.15\textwidth}
        \includegraphics[width=\textwidth]{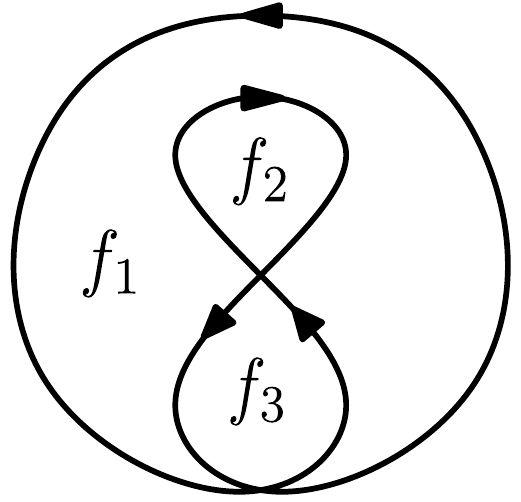}
       \subcaption{}\label{fig:initial-curve}
    \end{subfigure}
        \hspace{1cm}
        \begin{subfigure}[b]{0.15\textwidth}
        \includegraphics[width=\textwidth]{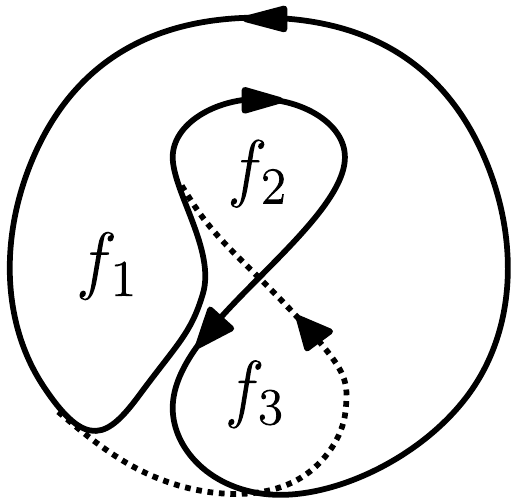}
        \subcaption{}\label{fig:sweepf3-1}
        \end{subfigure}
        \hspace{1cm}
        \begin{subfigure}[b]{0.15\textwidth}
    \includegraphics[width=\textwidth]{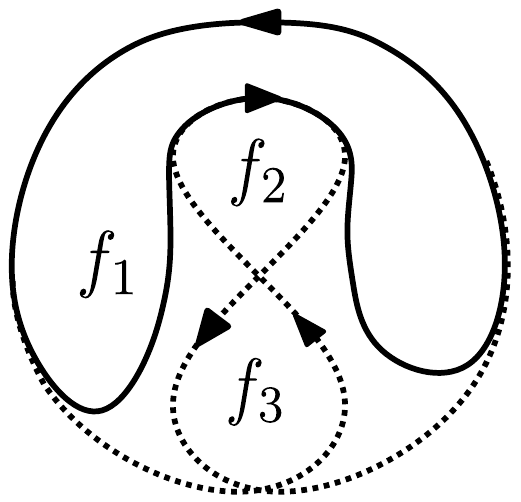}
      \subcaption{}\label{fig:sweepf3-2}
        \end{subfigure}
        \hspace{1cm}
        \begin{subfigure}[b]{0.15\textwidth}
    \includegraphics[width=\textwidth]{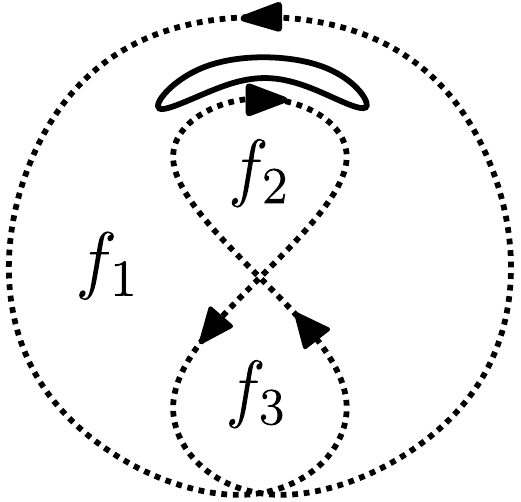}
	\subcaption{}\label{fig:avoidf2}
    \end{subfigure}
    \caption{(\subref{fig:initial-curve}) A generic closed curve in the plane.
        (\subref{fig:sweepf3-1}) We see a homotopy that sweeps over the face $f_3$.
        (\subref{fig:sweepf3-2}) The homotopy sweeps $f_3$ again.
        (\subref{fig:avoidf2}) The homotopy avoids sweeping over the face $f_2$.
        This is a minimum area homotopy for the curve, the area is $\Area(f_1) + 2\cdot\Area(f_3)$.
    }
    \label{fig:eight}
\end{figure}

For each $x\in \RR \setminus \curve$, the \EMPH{winding number} of $\gamma$ at $x$,
denoted as \EMPH{$\wind(x,\gamma)$}, is the number of times $\gamma$ ``wraps around'' $x$,
with a \emph{positive} sign if it is counterclockwise, and \emph{negative} sign otherwise.
%Put it differently, the winding number is the number of times $\gamma$ passes through an
%arbitrary ray from $x$ to infinity from right to left, minus the number of times $\gamma$
%passes the ray from left to right.
%A pleasant description is given in~\cite{chinn-steenrod}.
The winding number is a constant on each face.
The \EMPH{winding area} of $\gamma$ is defined to be the integral
\[
\EMPH{$\Area_W(\gamma)$} \coloneqq \int_{\R^2} \abs{\wind(x,\gamma)} \, dx
= \sum_{\text{face $f$}} \abs{\wind(f,\gamma)} \cdot \Area(f).
\]

% Next we describe a decomposition of a curve into simple cycles
% based on the depth that is due to Chang-Erickson~\cite{changErickson17}.
% Let $d$ denote the depth of $\gamma$, for all $1<j<d$, let $R_j$ denote
% the set of points with $\depth(x,\gamma)\geq d+1-j$.
% Consider a open neighborhood of the closure of $R_j\cup \tilde{R}_{j-1},$
%call this $\tilde{R}_j$. Each $\tilde{R}_j$ is the disjoint union of closed disks.
%Then $\tilde{R}_0$ is the empty set,
% $\tilde{R}_1$ is disk containing all points of depth $d,
% \tilde{R}_2$ is disk containing all points of depth $d-1$ and so on.
%Finally, $\tilde{R}_d$ is a disk containing the entire curve.
%We call such a decomposition the \EMPH{depth cycles}
%of a curve $\gamma$.

The \EMPH{depth} of a face $f$ is the minimal number of edges crossed
by a path from $f$ to the exterior face.
The depth is a constant on each face.
We say the depth of a curve is equal the maximum depth over all faces.
We define the \EMPH{depth area}~to~be
\[
\EMPH{$\Area_D(\gamma)$} \coloneqq \int_{\R^2} \depth(x,\gamma)\, dx
= \sum_{\text{face $f$}} \depth(f) \cdot \Area(f).
\]

Chambers and Wang~\cite{cw2013} showed that the winding area gives a lower
bound for the minimum homotopy area.
On the other hand, there is always a homotopy with area $\Area_D(\gamma)$;
one such homotopy can be constructed by smoothing the curve
at each vertex into simple depth cycles \cite{changErickson17},
 then contracting each simple cycle.
Therefore we have
\begin{equation}
\label{eq:sandwich}
    \Area_W(\gamma) \leq \Area_H(\gamma) \leq \Area_D(\gamma).
\end{equation}

\subsection{Self-Overlapping Curves}
\label{SS:so-curves}

% \begin{TODO}
%     Reassess if we need to move this subsection to Section~4 after the later
%     parts are done.
% \end{TODO}

A generic curve $\gamma$ is \EMPH{self-overlapping} if there exists an immersion of the two disk $F:\mathbb{D}^2\rightarrow \RR^2$ such that $\curve= F|_{\partial\!D^2}$.
We say a map $F$ \EMPH{extends} $\gamma$.
The image $F(\mathbb{D}^2)$ is the \EMPH{interior} of~$\gamma$.
%Intuitively, the image the disk $\mathbb{D}^2$ has two sides, one colored with blue and the other with red.
%A curve $\gamma$ is self-overlapping if it is the boundary of a disk $\mathbb{D}^2$
%with only one color facing up~\cite{mukherjee2014}.
There are several equivalent ways to define self-overlapping curves
\cite{evansFasyWenk,titus,shor-van-wyk,so-graphics,mukherjee2014}.
Properties of self-overlapping curves are well-studied~\cite{evans18}; in particular,
any self-overlapping curve has rotation number 1, where the \EMPH{rotation number} of a curve $\gamma$ is the winding number of the derivative $\gamma'$ about the origin \cite{whitney1937}.
Also, the minimum homotopy area of any self-overlapping curve is equal to its winding area:
$\Area_W(\gamma) = \Area_H(\gamma)$~\cite{fkw2017}.

The study of self-overlapping curves traces back to Whitney~\cite{whitney1937} and Titus~\cite{titus}.
Polynomial-time algorithms for determining if a curve is self-overlapping
have been given \cite{blank,shor-van-wyk}, as well as NP-hardness result for extensions to surfaces and higher-dimensional spaces~\cite{eppsteinMumford}.
%See the thesis by Evans~\cite{evans18} for a more detailed history.
%These works are not only interested in determining if a curve is self-overlapping, they also determine the number of equivalent extensions.
%See \figref{extensions} for an example of a curve with two extensions
% that are not diffeomorphic.
% In a related work, Eppstein and Mumford show it is NP-complete to determine if a
% curve is the boundary of an immersed surface with boundary in $\R^3$
% \cite{eppsteinMumford}. \hsien{Other results we want to cite?}

For any curve, the \EMPH{intersection sequence}%
\footnote{also known as the unsigned Gauss code~\cite{changErickson17,gauss}}
$[\gamma]_V$ is a cyclic sequence of vertices~$[v_0, v_1, \ldots , v_{n-1}]$
with $v_n = v_0$, where each $v_i$ is an intersection point of $\gamma$.
Each vertex appears exactly twice in $\gamma_V$.
Two vertices $x$ and $y$ are \EMPH{linked} if the two appearances of $x$ and
$y$ in $\gamma_V$ alternate in cyclic order:  $\dots x \dots y \dots x \dots y \dots \,$.

A pair of symbols of the same vertex $x$ induces two
natural subcurves generated by \EMPH{smoothing} the vertex $x$;
see \figref{vertex-cut} for an example.
(In this work, every smoothing is done in the way that respects the orientation
and splits the curve into two subcurves.)
A \EMPH{vertex pairing} is a collection of pairwise unlinked vertex pairs in $[\gamma]_V$.

A \EMPH{self-overlapping decomposition} $\Gamma$ of $\gamma$ is a vertex pairing
such that the induced subcurves are self-overlapping; see
\figref{total-decomp} and \figref{intersections-p1} for examples.
The subcurves that result from a vertex pairing are not necessary
self-overlapping; see \figref{intersections-p0}.
For a self-overlapping decomposition $\Gamma$ of $\gamma$, denote the set of
induced subcurves by $\{\gamma_i\}_{i=1}^\ell$.
Since each $\gamma_i$ is self-overlapping, the minimum homotopy area is equal to
its winding area.
We define the \EMPH{area of self-overlapping decomposition} to be
\[
\EMPH{$\Area_\Gamma(\gamma)$} \coloneqq
\sum_{i=1}^\ell \Area_W(\gamma_i)
= \sum_{i=1}^\ell \Area_H(\gamma_i).
\]

\begin{figure}[t]
    \captionsetup[subfigure]{justification=centering}
    \centering
    \begin{subfigure}[b]{0.22\textwidth}
        \includegraphics[width=\textwidth]{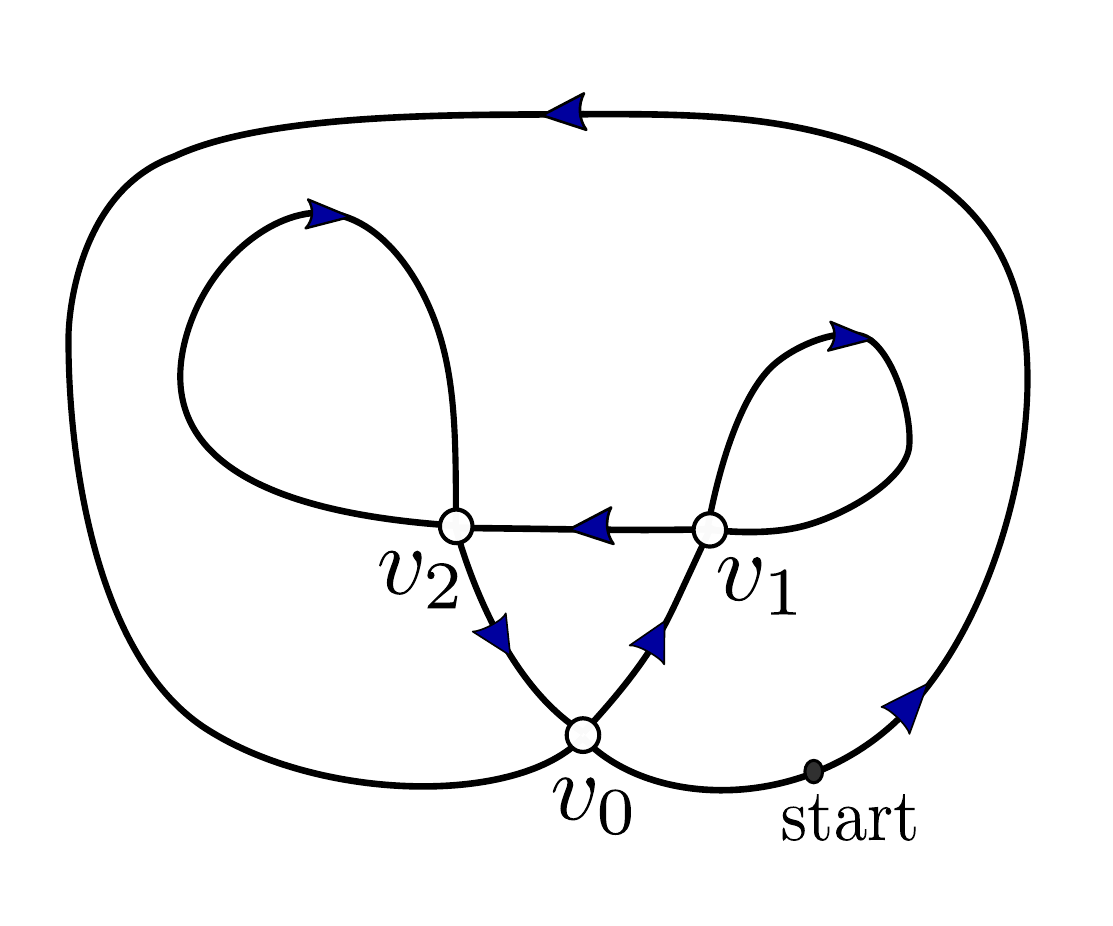}
        \subcaption{}\label{fig:intersections}
    \end{subfigure}
    \begin{subfigure}[b]{0.22\textwidth}
        \includegraphics[width=\textwidth]{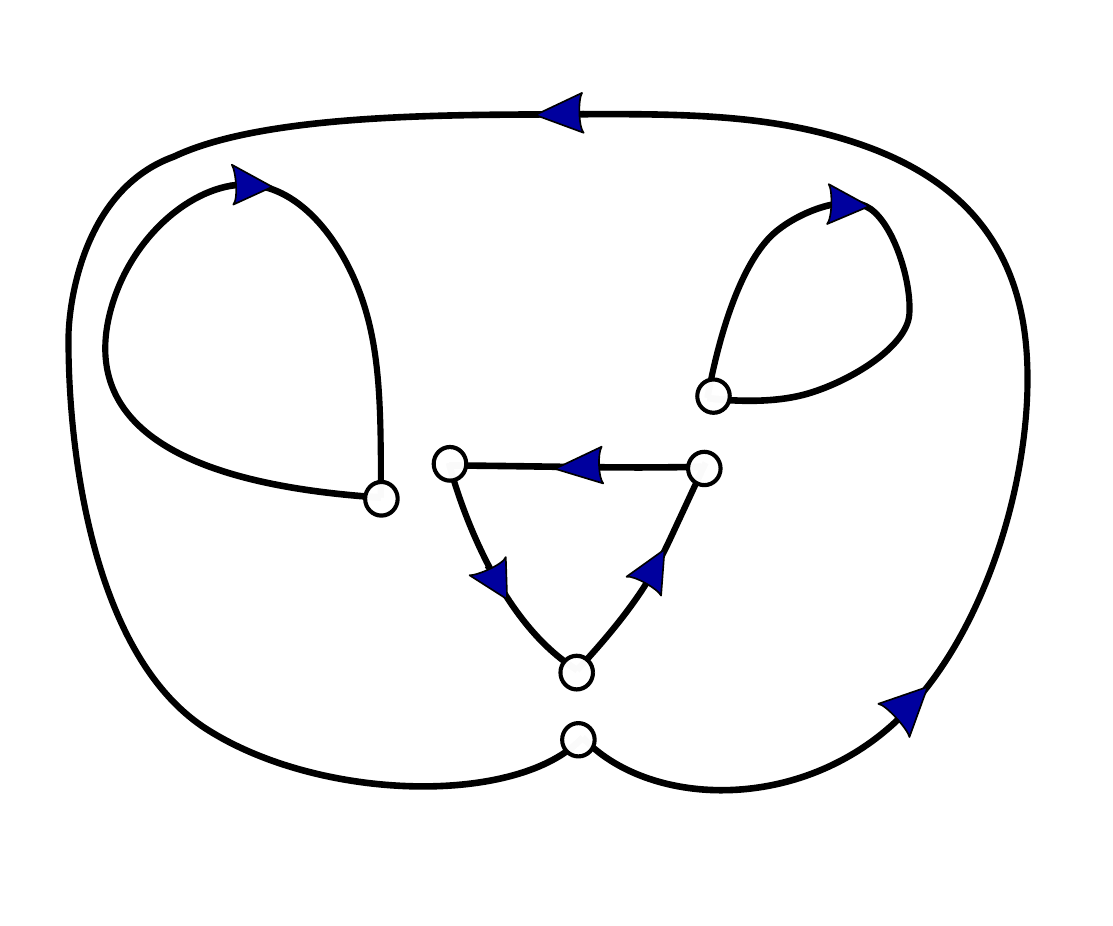}
        \subcaption{}\label{fig:total-decomp}
    \end{subfigure}
    \begin{subfigure}[b]{0.22\textwidth}
        \includegraphics[width=\textwidth]{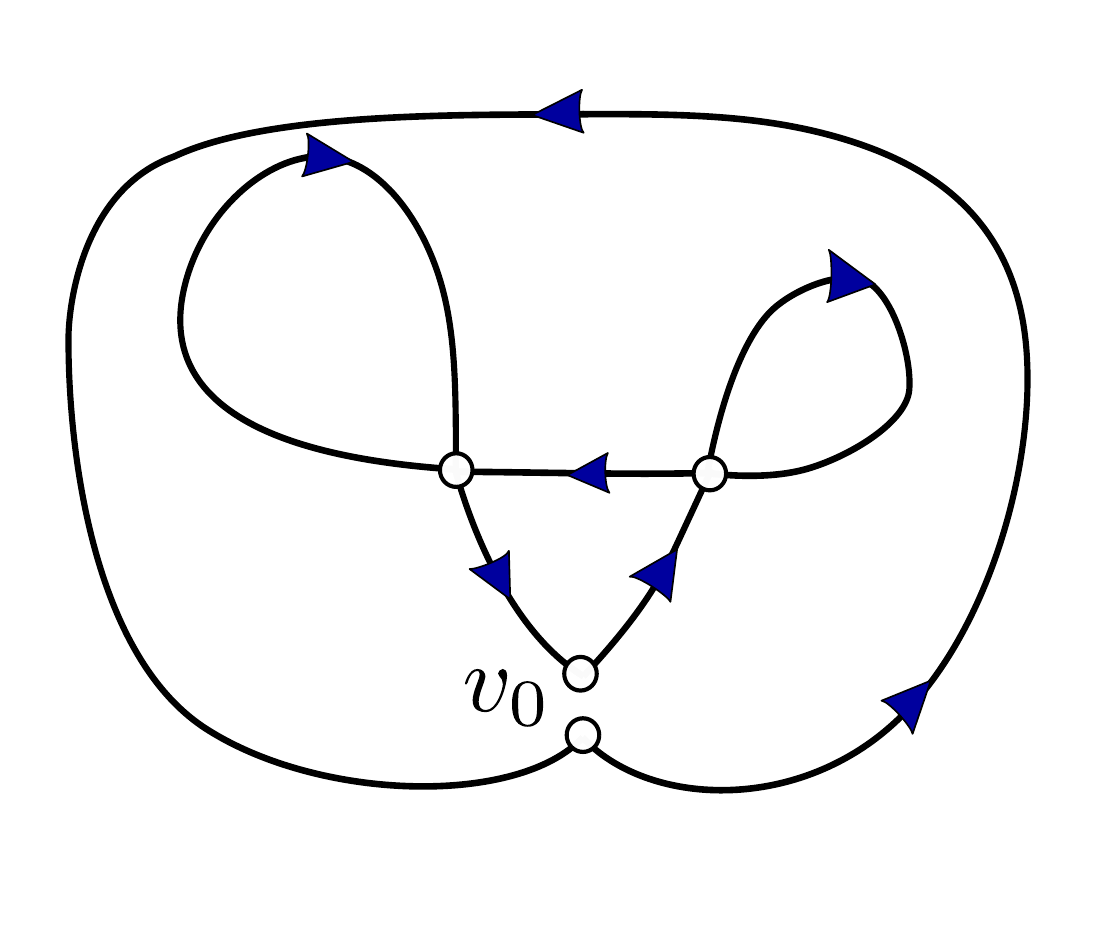}
        \subcaption{}\label{fig:intersections-p0}
    \end{subfigure}
    \begin{subfigure}[b]{0.22\textwidth}
        \includegraphics[width=\textwidth]{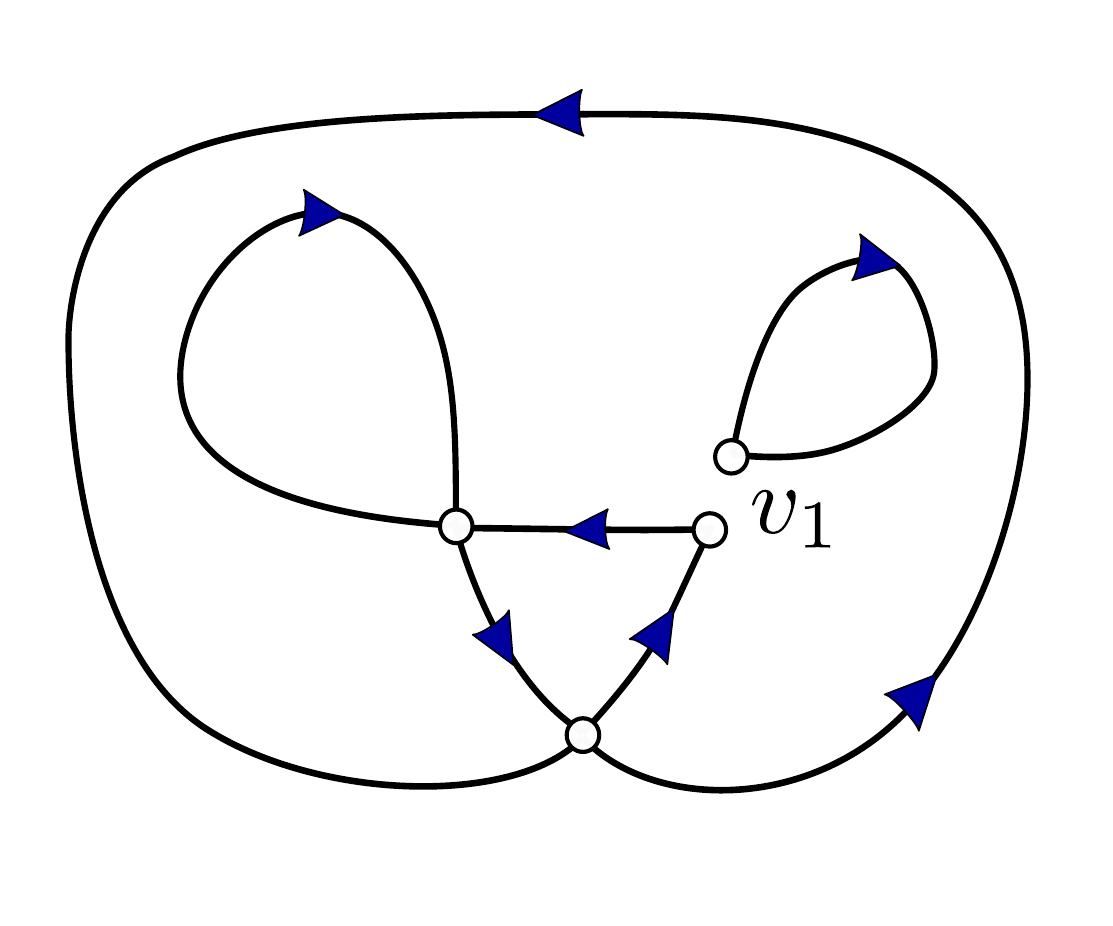}
        \subcaption{}\label{fig:intersections-p1}
    \end{subfigure}

    \caption{
        (\subref{fig:intersections}) Curve $\gamma$ with intersection sequence
        $\gamma_V=[v_0,v_1,v_1,v_2,v_2,v_0]$.
        (\subref{fig:total-decomp}) All vertices are paired. %in the intersection sequence.
        (\subref{fig:intersections-p0})  One of the subcurves is not self-overlapping.
        (\subref{fig:intersections-p1}) Both subcurves are self-overlapping.}
    \label{fig:vertex-cut}
\end{figure}

Fasy, Karakoç, and Wenk~\cite{fkw2017,karakoc17} proved the following structural theorem.
\begin{theorem}[Self-Overlapping Decomposition~{\cite[Theorem 20]{fkw2017}}]
    \label{thm:structural}
    Any curve $\curve$ has
    a self-overlapping decomposition whose area is minimum over all
    null-homotopies of $\curve$.
\end{theorem}
%

% Given the above theorem one might think we should look for a self-overlapping decomposition with minimum area directly, which
% \todo{Remove or complete the above sentence.}
%
% The above theorem implies if one finds the minimum area over all self-overlapping decompositions, then one has found the minimum homotopy area.
% \hsien{Talk about how no poly-time algorithm is known, and why ours requires the
% existance before the algorithmic; may in the introduction.}

%!TEX TS-program = pdflatex
%!TEX encoding = UTF-8 Unicode
%!TeX spellcheck = en-US
%!TeX root = ..\socg23.tex

%----------------------- Curves_to_Words-------------------------------------
\section{From Curves to Words}
\label{sec:curve-to-word}

In order to work with plane curves, one must choose a \emph{representation}.
%There are several natural ways to represent plane curves in the literature.
%Examples include
%
%\emph{polygonal curves}, consisting of edge segments and vertices at the joints
%of every pair of consecutive segments;
%
%\emph{closed walks} on some underlying combinatorial surface; and
%
%\emph{angle-and-length sequences}, often used in the complex plane.
An important class of representations for plane curves are the various
\emph{combinatorial words}.
One example is the \emph{Gauss code}~\cite{gauss}.
%To construct the Gauss code from a curve, assign each crossing a label;
%then, traverse the curve and record the sequence of labels encountered,
%possibly also with the direction of how the second strand passes through
%the first at each intersection.
Determining whether a Gauss code corresponds to an actual plane curve is
one of the earliest computational topology questions
\cite{erickson-note}.

A plane curve (and its homotopic equivalents) can also be viewed as
a word in the fundamental group $\pi_1(\gamma)$ of $\curve$~\cite{blank,poe-eic1-1968,nie2014}.
If we put a point $p_i$ in each bounded face $f_i$, the curve~$\curve$ %as a space
is generated by the unique generators of each~$\RR^2 - \set{p_i}$.
Nie~\cite{nie2014} represents curves as words in the fundamental group to find the minimum area swept out by contracting a curve to a point.
If the curve lies in a plane with punctures, one can define the
\emph{crossing sequence} of the curve with respect to a \emph{system of arcs},
cutting the plane open into a simply-connected region.
%An example that uses the crossing sequences is given by Blank in
Blank~\cite{blank} represents curves using a crossing sequences to determine if a curve is self-overlapping.
%---that is, the boundary of an immersed disk---
%by checking whether the combinatorial word,
%constructed by drawing arcs from each face to infinity then traversing the curve and recording
%the signed intersection sequence of the curve and the arcs, satisfies certain groupability condition.  (See Section~\ref{SS:norm} and Theorem~\ref{thm:blank} for details.)
%In other words, he created the system of arcs based on the input curve $\gamma$,
%viewed as a space by adding a puncture to each face of $\gamma$.
While Blank constructed the words \emph{geometrically} by drawing arcs and
Nie defined the words \emph{algebraically}, the dual view between the system
of arcs and fundamental group suggests that the resemblance between Blank
and Nie's constructions is not a coincidence.

In this section,
%we show that Blank and Nie's word constructions are, in fact, equivalent under the right .
we describe the construction by Blank;
then, we interpret Blank's construction as a way of choosing the basis for the fundamental group under further restriction~\cite{poe-eic1-1968}.
We prove that the Blank word is indeed unique when the restriction is enforced,
providing clarification to Blank's original definition.
We give a complete description of Nie's word
construction and prove that Nie's word and Blank's word are equivalent.
%We demonstrate the one-to-one correspondence between Blank's word and Nie's word,
%thus establishing their equivalence as combinatorial words.

% we describe the assumptions required in order for the words to be well-defined.
% We then prove the equivalence between Blank's word and Nie's word under such assumptions.
% The equivalence has been explored by Po\'enaru~\cite{poe-eic1-1968};
% however our combinatorial one-to-one
% correspondence provides a geometric way of interpreting Nie's work.

\subsection{Blank's Word Construction}
\label{sec:blank}

We now describe Blank's word construction~\cite[page~5]{blank}.
Let $\gamma$ be a generic closed curve in the plane,
pick a point in the unbounded face of $\curve$, call it the \EMPH{basepoint~$p_0$}.
From each bounded face $f_i$, pick a \EMPH{representative point $p_i$}.
%\todo{BTF/CW: perhaps we should have picked them before, not here}
%
Now connect each $p_i$ to $p_0$ by a simple path in such a way that no two paths
intersect each other.
We call the collection of such simple paths a \EMPH{cable system},
denoted as \EMPH{$\Pi$}, and
each individual path \EMPH{$\pi_i$} from $p_i$ to $p_0$ as
a~\EMPH{cable}.

Orient each $\pi_i$ from $p_i$ to $p_0$.
Now traverse $\curve$ from an arbitrary \EMPH{starting point} of $\curve$ and
construct a cyclic word by writing down the indices of $\curve$ crossing the
cables $\pi_i$ in the order they appear on $\curve$; each index $i$ has a
\emph{positive} sign if we cross $\pi_i$ from right to left and a
\emph{negative} sign if from left to right.
We denote negative crossing
with an overline $\str{\overline{i}}$.  We call the resulting combinatorial word
over the faces a \EMPH{Blank word} of $\curve$ with respect to $\Pi$, denoted as
\EMPH{$[\curve]_B(\Pi)$}.  \figref{blank-word} provides an example of Blank's
construction.

\begin{figure}[h!]
    \captionsetup[subfigure]{justification=centering}
    \centering
    \begin{subfigure}[t]{0.3\textwidth}
        \includegraphics[width=\textwidth]{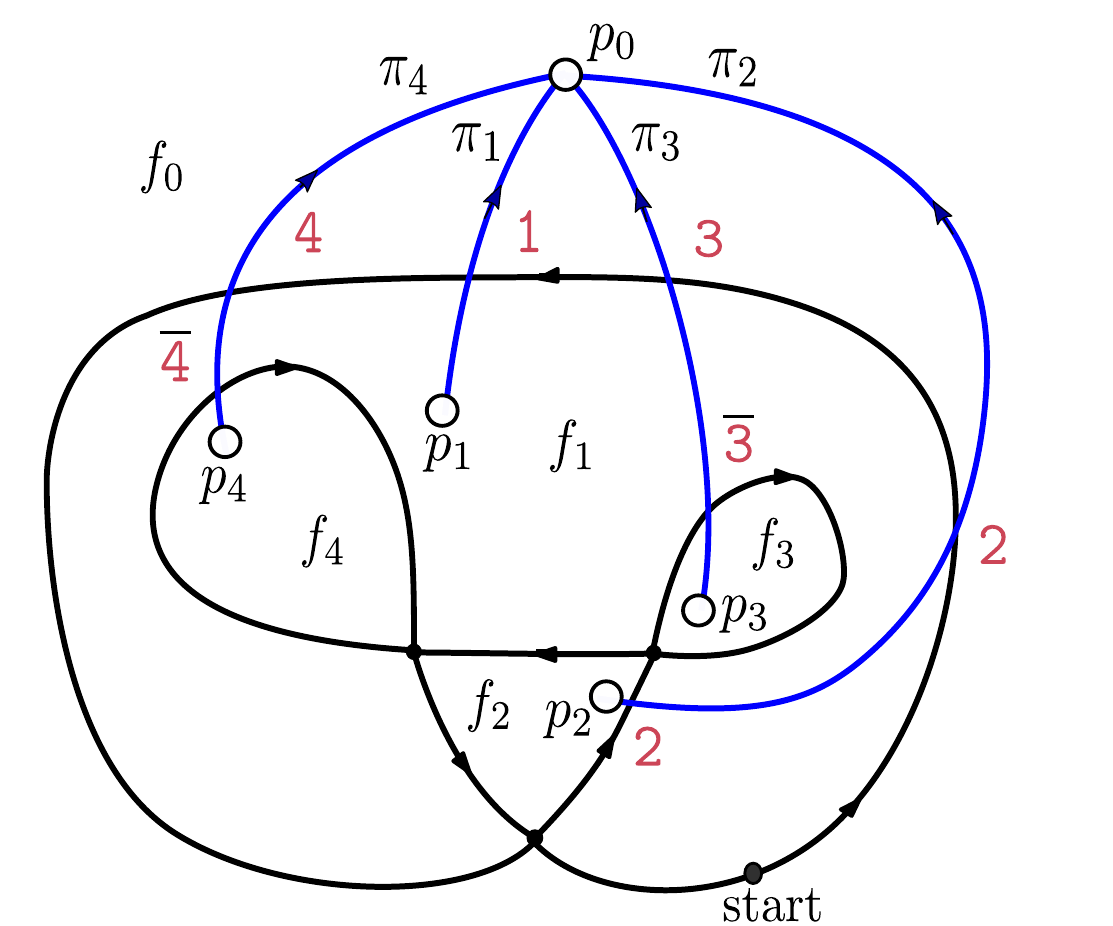}
         \subcaption{}\label{fig:blank1}
    \end{subfigure}
    \hspace{1cm}
    \begin{subfigure}[t]{0.3\textwidth}
        \includegraphics[width=\textwidth]{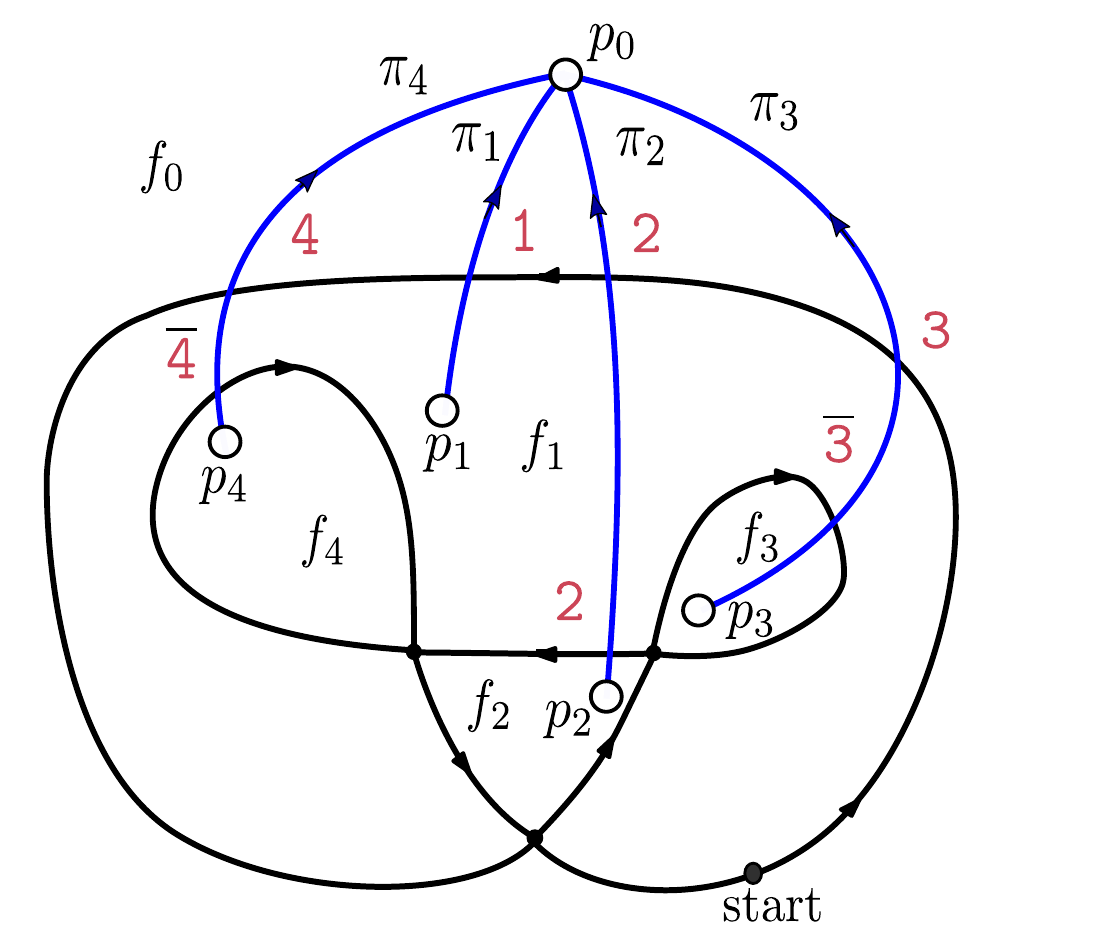}
        \subcaption{}\label{fig:blank2}
    \end{subfigure}

    \caption{
        (\subref{fig:blank1}) A curve $\gamma$ with labeled faces and edges, $\Pi_a$ is drawn in blue.
        The Blank word of~$\gamma$ corresponding to $\Pi_a$ is $[\gamma]_B(\Pi_a) = 			[\str{23142\overline{3}\overline{4}}]$.
        (\subref{fig:blank2}) The same curve with a different choice of cables $\Pi_b$.
        The corresponding Blank word is~$[\gamma]_B(\Pi_b) =~[\str{3214\overline{3}2\overline{4}}]$.
    }
    \label{fig:blank-word}
\end{figure}

% \begin{TODO}
%     HC: I just strip out the shortest path assumption from the cable systems.
%     Check later sections whenever we talk about cable system do we need
%     shortest-paths, and add the adjective when necessary.
% \end{TODO}

A word $w$ is \EMPH{reduced} if there are no two consecutive symbols in $w$ that are identical
and with opposite signs.
We can enforce every Blank word to be reduced by imposing the following  \EMPH{shortest path assumption}:
each cable has a minimum number of intersections with~$\gamma$ among all paths from $p_i$ to $p_0$.
%
% \begin{itemize}\label{I:shortest}
%     \item each cable has a minimum number of intersections with $\gamma$ among all paths from
%         $p_i$ to $p_0$.\todo{BTF/CW: this is a bulleted list with only one item.
%         needs to be structured a bit better. also, what `shortest' is needs to
%         be defined.} \hsien{``Shortest'' means minimum number of intersections.}
% \end{itemize}
%
A simple proof~\cite{blank,frisch10} shows that if $\Pi$ satisfies the shortest path assumption,
the corresponding Blank word with respect to $\Pi$ is reduced.
However, the choice of the cable system, and how it affects the constructed Blank word, was never explicitly discussed in the original work (presumably
because for the purpose of detecting self-overlapping curves, any cable system
satisfying the shortest path assumption works).
In general, reduced Blank words constructed from different cable systems for the same curve are not identical, see \figref{blank1} and \figref{blank2} for an example.
In this paper, we show that if the two cable systems have the same \emph{cable ordering}--the (cyclic) order of cables around point $p_0$ in the unbounded face--then their corresponding
(reduced) Blank words are the same, under proper assumptions on the cable~system.
%
%  \carola{Where does the isotopy come from? How is it defined? It's an isotopy between what -- two cable systems? Where do we need this? I'm confused.}
%  \hsien{Isotopy is between two cable systems.  They are given, later by the algorithm.  The reason is the make sure that the reduced Blank word is always well defined one you specified its isotopy class, so that actually location of the cables does not matter.}

Our first observation is that the Blank words are invariant under
cable isotopy; therefore the cable system can be specified up to isotopy.  %The proof is in \appendref{words}.

\begin{restatable}[Isotopy Invariance]{lemma}{blankisotopy}
%\begin{lemma}[Isotopy Invariance]
\label{lem:isotopy}

    	The reduced Blank word is invariant under cable~isotopy.

\end{restatable}
\begin{proof}
    Let $\gamma$ be a curve.
    Discretize the isotopy of the cables and consider all the possible
    \emph{homotopy moves}~\cite{changErickson17} performed on $\curve$ and the
    cables involving up to two strands from $\curve$ and a cable, because
    isotopy disallows the crossing of two cables.
    No $\biarc10$ move---the move
    that creates/destroys a self-loop---is possible as cables do not
    self-intersect.
    Any $\biarc20$ move which creates/destroys a bigon is in
    between a cable and a strand from $\gamma$, which means the two
    intersections must have opposite signs, and therefore the reduced Blank word
    does not change.
    Any~$\arc33$ move which moves a strand across another
    intersection does not change the signs of the intersections, so while the
    order of strands crossing the cable changes, the order of cables crossed by
    $\gamma$ remains the same.  
    Thus the reduced Blank word stays the same.
\end{proof}
We remark that we can perform an isotopy so that the Blank words are reduced even
when the cables are not necessary shortest paths.
In the rest of the paper, we sometimes assume Blank words to be reduced based on the context.
%\hsien{Paragraph on we assume words to be reduced even without shortest path assumption, and justification by isotopy.}

\paragraph*{Manage the Cable Systems}
Next, we show that Blank words are well-defined once we fix the choice of basepoint $p_0$ and the cyclic cable ordering around $p_0$, as long as the cables are drawn in a reasonable way.
%
% For the purpose of demonstrating equivalence and proving uniqueness\todo{BTF/CW: of
% what. we just had a lemma that proved uniqueness ...}, we further
% restrict the cables to satisfy one extra condition.
%
Fix a tree-cotree pair $(T,T^*)$ of $\curve$, where the root of the
cotree is on $p_0$.
We say that a cable system $\Pi$ is \EMPH{managed} with respect to the cotree
$T^*$ if each path $\pi_i$ has to be a path on~$T^*$ from the root $p_0$ to the leaf $p_i$.
%
% \begin{itemize}\label{I:managed}
%     \item Each path $\pi_i$ has to be a path on $T^*$ from the root $p_0$ to the leaf $p_i$.
% \end{itemize}
%
Given such a collection of cotree paths, one can slightly perturb them to ensure that all
paths are simple and disjoint.\footnote{In other words, the cables are
\EMPH{weakly-simple}~\cite{tou-stsps-1989}.}
%\hsien{Do we need the shortest-path assumption in the definition?}
%\note{Better define using weak-simplicity, to avoid issue of changing the
%homotopy area.}
See \figref{cable-mgmt} for examples.
Not every cable system can be managed with respect to $T^*$, see \figref{no-mgmt} for an example.

\begin{figure}[htb]
    \captionsetup[subfigure]{justification=centering}
    \centering

    \hspace{.2cm}
    \begin{subfigure}[b]{0.2\textwidth}
        \includegraphics[width=\textwidth]{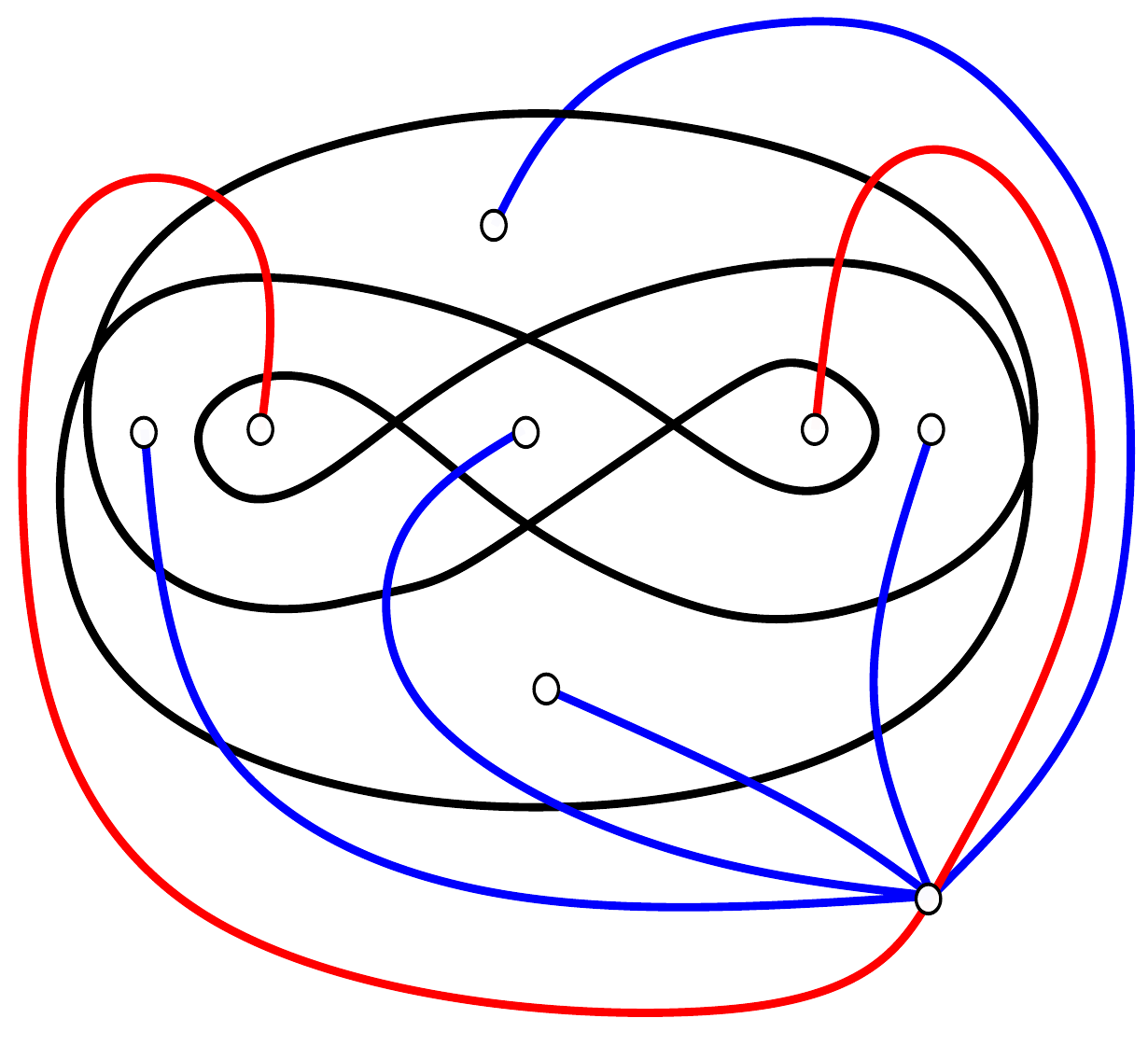}
        \subcaption{}\label{fig:mgmt-2}
    \end{subfigure}
    \hspace{.2cm}
    \begin{subfigure}[b]{0.2\textwidth}
        \includegraphics[width=\textwidth]{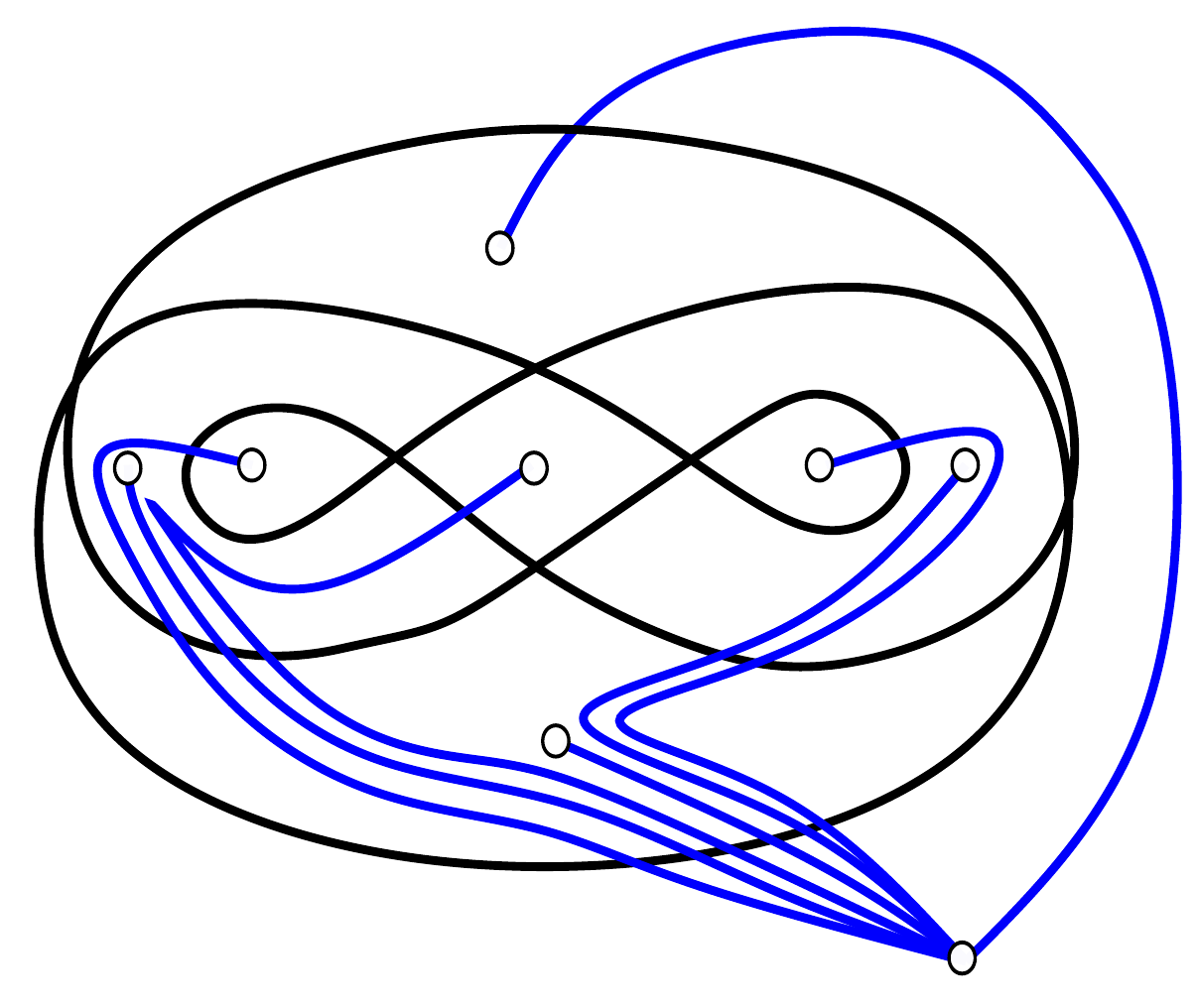}
        \subcaption{}\label{fig:mgmt-3}
    \end{subfigure}
    \hspace{.2cm}
    \begin{subfigure}[b]{0.2\textwidth}
        \includegraphics[width=\textwidth]{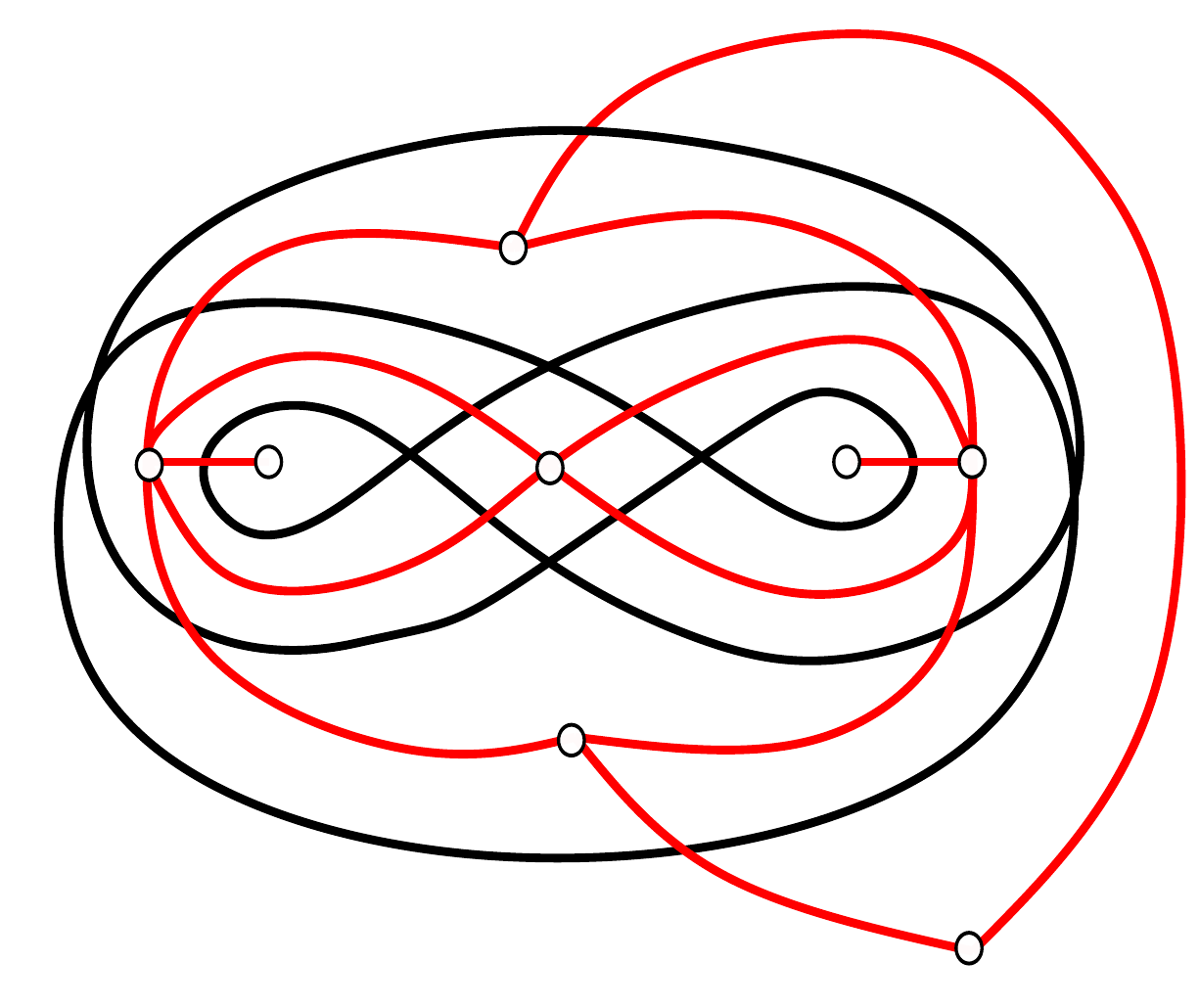}
        \subcaption{}\label{fig:mgmt-4}
    \end{subfigure}
    \hspace{.2cm}
    \begin{subfigure}[b]{0.2\textwidth}
        \includegraphics[width=\textwidth]{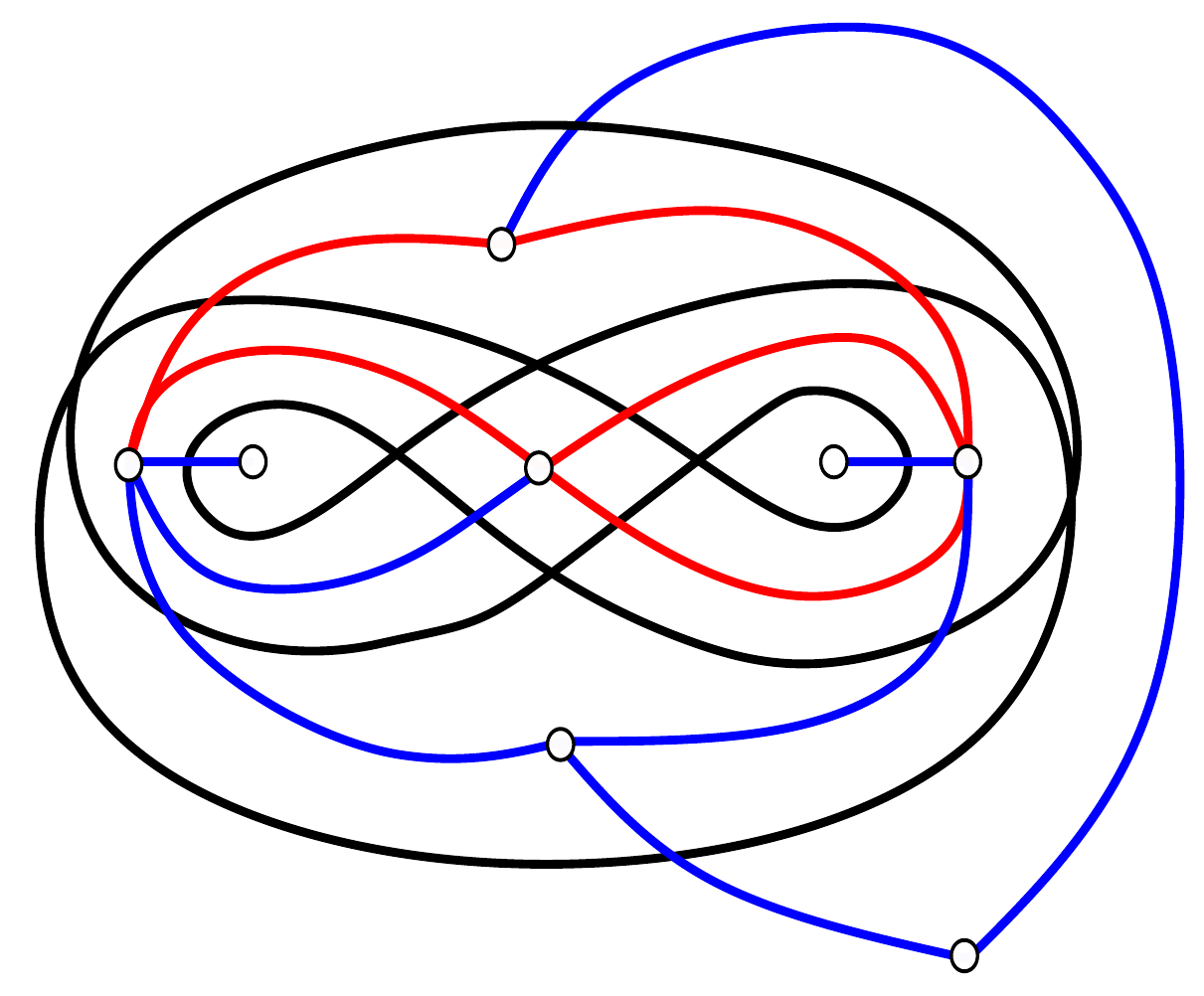}
        \subcaption{}\label{fig:mgmt-5}
    \end{subfigure}

    \caption{
        %(a) A curve $\gamma$.
        (\subref{fig:mgmt-2}) A cable system $\Pi_1$ on $\curve$ that is not managed.
        The red cables do not follow existing paths to the exterior face.
        (\subref{fig:mgmt-3}) A managed cable system $\Pi_2$ on $\curve$.
        (\subref{fig:mgmt-4}) The dual $\curve^*$ in red.
        (\subref{fig:mgmt-5}) The spanning tree $T^*$ in~$\curve^*$ generated by the managed cable system $\Pi_2$.
        }
\label{fig:cable-mgmt}
\end{figure}

%In \appendref{words},
We now show that if two managed cable systems satisfying shortest path assumption 
with identical cable ordering around $p_0$, their corresponding Blank words are the same.
Note that managed cable systems require a fixed tree-cotree pair.
We emphasize that the shortest path assumption is necessary; one can construct
two (not necessarily shortest) cable systems having the same cable ordering but
different corresponding reduced Blank words (see \figref{blank3} and
\figref{blank4}).

% \paragraph*{Contour Representation}
% To prove that the reduced Blank word is unique after fixing the cable ordering, we need to introduce
% an alternative view to the input curve $\gamma$.
% %
% Let $f$ be a face of $\gamma$.
% The \EMPH{depth} of face $f$ is the minimum number of times a path, from the interior of $f$ to the unbounded face, has to cross~$\gamma$.
% (If the path passes through a vertex of $\gamma$, it counts as two crossings.)
% %
% For each $j$, define region~$R_{\ge j}$ to be the closure of the union of faces whose depth is at least $j$.
% Each region might consist of several connected components; in addition, each component is formed by \emph{blocks} each bounded by a simple cycle, where two blocks can share only vertices (but not edges).
% We refer to the boundary cycles of the blocks as \EMPH{contours}.
% Denote the region bounded by a contour~$\kappa$ as \EMPH{$D_\kappa$}.
% The collection of contours for~$\gamma$ forms a nested family of simple cycles; contour
%$\kappa$ is the ancestor of contour $\kappa'$ if region $D_\kappa$ properly contains $D_{\kappa'}$.
% We define the depth of a contour $\kappa$ to be the depth of the face within $\kappa$ with smallest depth.

\begin{restatable}[Blank Word is Unique]{lemma}{blankunique}
    \label{lem:blank-unique}
    Given a curve $\curve$, if the basepoint $p_0$ and the cable ordering of
    a managed cable system $\Pi$ satisfying the shortest path assumption is
    fixed, then the Blank word of $\curve$ is unique.
\end{restatable}

\begin{proof}

    We will argue that once the basepoint $p_0$ and the order of cables in $\Pi$ around $p_0$ is
    fixed, all the drawings of $\Pi$ respecting the cable ordering lead to the same Blank word.
    %
    %Consider the (planar) dual graph of $\gamma$; the cables can be viewed as a collection
    %of simple paths in the dual.
    Because the cable system is managed, the tree-cotree pair of $\gamma$ are fixed and
    we can safely contract the primal tree and treat the
    graph as a collection of nested loops.
    If the path passes through a vertex of $\gamma$, it counts as two crossings.
    We prove that all the cables to the loops at certain depth have a fixed ordering by
    induction on the depth.
    This is sufficient as any two cable systems with the same cables and ordering on every
    loop of the same depth must be isotopic, thus by \lemref{isotopy} their Blank words are identical.
    Because of the shortest path assumption, there is only one way to draw the cables to the depth-1 contours.
    % Now if we denote the cable ordering of $\Pi$ as a cyclic permutation $\pi$, the cables
    %connecting to depth-1 faces partitions the rest of the symbols in $\pi$ (excluding those associated
    %with depth-1 faces) into a sequence of subwords $\pi_1, \dots, \pi_{k}$, where $k$ is the
    %number of depth-$1$ faces.
    % (Notice that some $\pi_i$ might be empty, and each $\pi_i$ is no longer cyclic.)
    % Imagine we extend the cables to the boundary of the depth-2 contours in a pairwise disjoint fashion; the cable orders on the depth-2 contours are uniquely determined.

    For any $\ell$, imagine all cables of depth at least $\ell$ are currently drawn from $p_0$
    to the depth-$\ell$ loops, where on each loop the collection of the cables are precisely
    those faced contained within the loop, and the cable ordering on the loops is fixed.
    Due to the shortest path assumption, every cable of depth $\ell$ has to terminate at their
    corresponding face.
    There is at most one unique way to extend each cable of depth $\ell$ to its representative
    point in the face while keeping all depth-$\ell$ cables disjoint and simple, up to isotopy of the cables.
    (If no such drawing exist, this particular cable ordering is not realizable as a cable system.) %\note{requires a proof}
    By \lemref{isotopy}, isotopy does not change the order the curve $\gamma$
    passing through these depth-$\ell$ cables.
    Now we partition the cables to faces of depth greater than $\ell$ based on
    children loops that contains the corresponding faces.
    There is one unique way to extend each cable of depth greater than $\ell$ to the
    depth-$(\ell+1)$ loops up to isotopy.
    Again by \lemref{isotopy}, isotopy does not change the order the curve
    $\gamma$ passing through the cables of depth more than $\ell$.
    By induction, the Blank word of $\gamma$ is unique.
\end{proof}

\begin{figure}[htb]
    \centering
    \includegraphics[width=.3\textwidth]{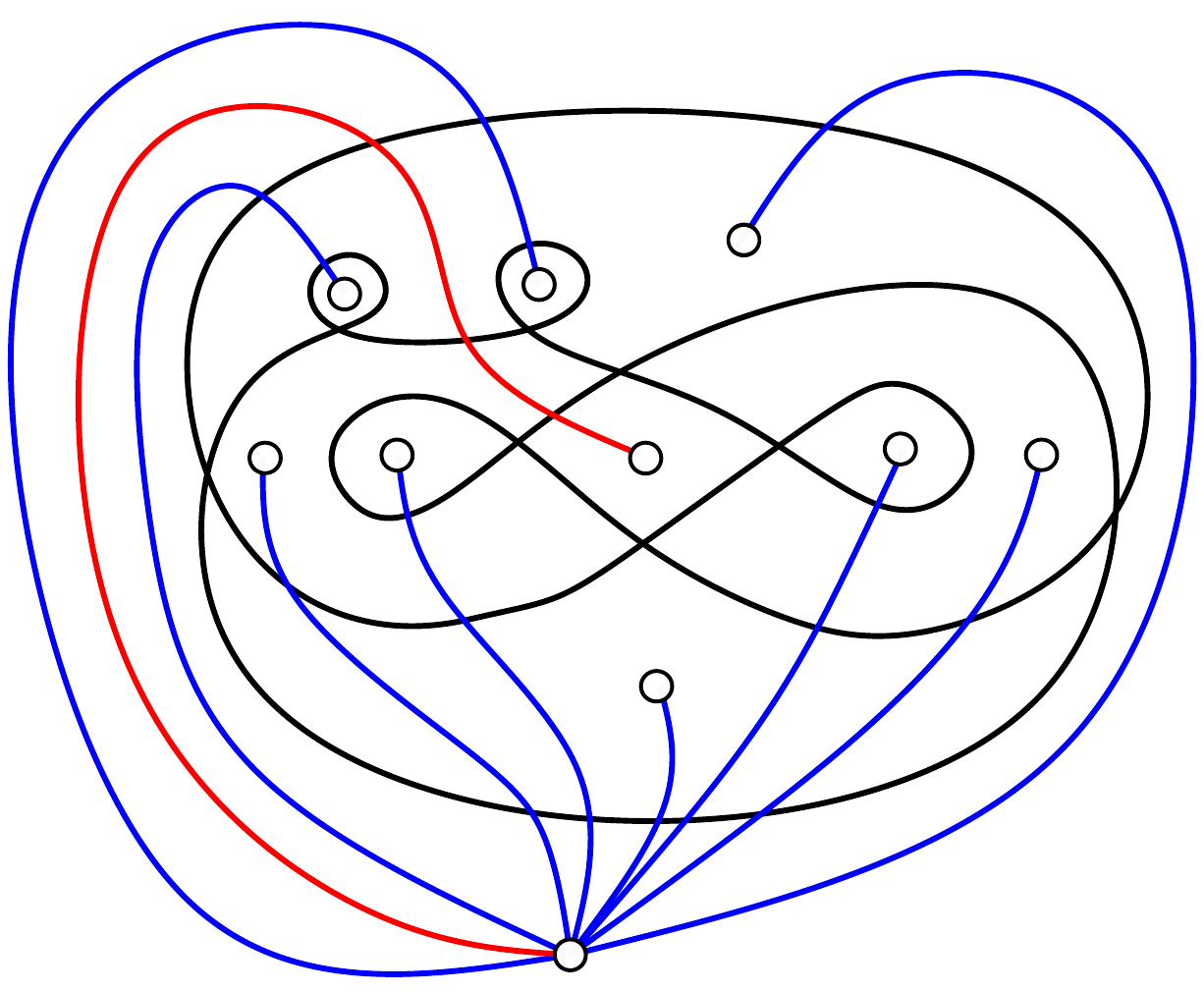}
    \caption{An example of a cable system with a cable, in red, that cannot be managed.}
    \label{fig:no-mgmt}
\end{figure}

\begin{figure}[h!]
    \captionsetup[subfigure]{justification=centering}
    \centering
    \begin{subfigure}[t]{0.35\textwidth}
        \includegraphics[width=\textwidth]{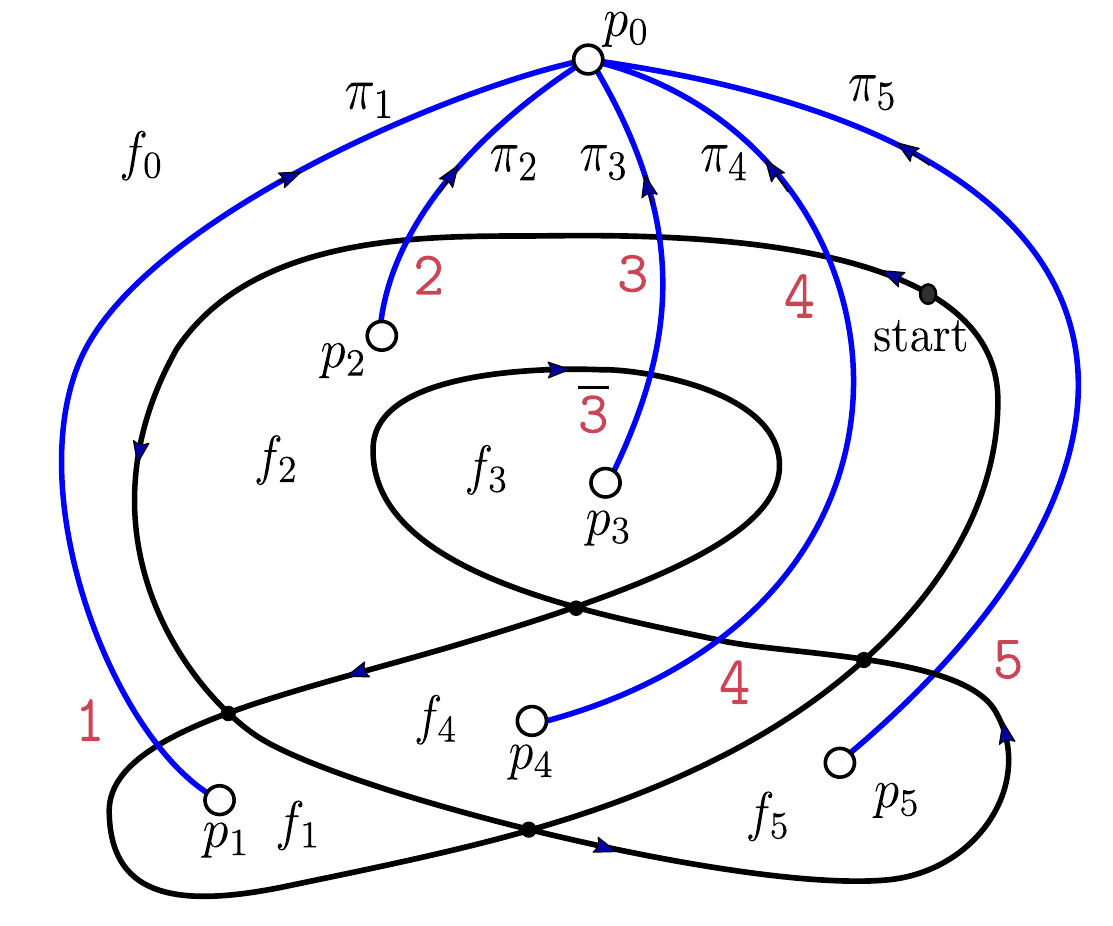}
        %\caption{A shortest path cable system.}
        \subcaption{}\label{fig:blank3}
    \end{subfigure}
    \hspace{1cm}
    \begin{subfigure}[t]{0.35\textwidth}
        \includegraphics[width=\textwidth]{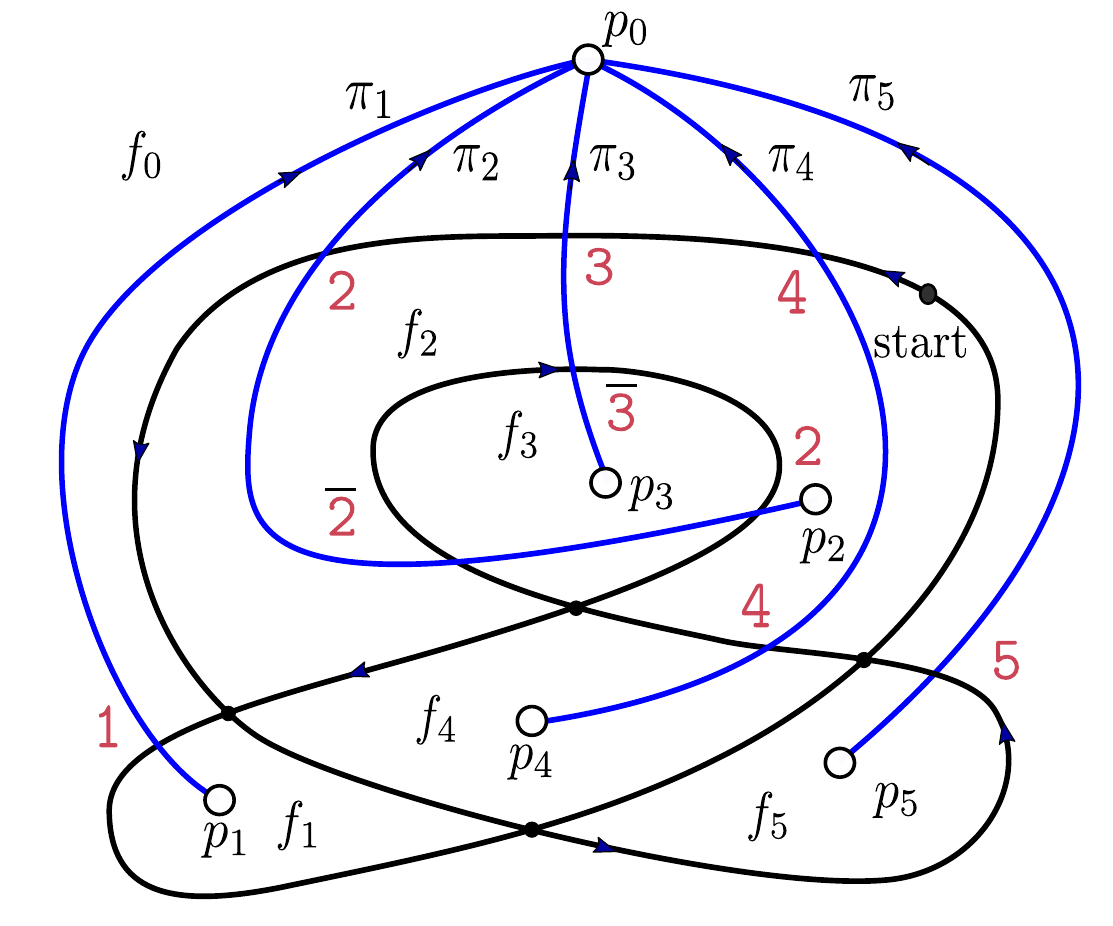}
        %\caption{A non shortest path cable system.}
        \subcaption{}\label{fig:blank4}
    \end{subfigure}
    \caption{
        (\subref{fig:blank3}) A curve with labeled faces and edges. A shortest path cable system $\Pi_a$ is drawn in blue.
        From the indicated start position, the Blank Word is $[\gamma]_B(\Pi_a) = [\str{43254\overline{3}1}]$.
        (\subref{fig:blank4}) The same curve with a cable system $\Pi_b$ that does not fulfill the shortest path assumption. The Blank word is $[\gamma]_B(\Pi_b) = [\str{43254\overline{2}\overline{3}21}]$.
        }
        \label{fig:blank-word-bfs}
\end{figure}

Therefore, given any plane curve $\gamma$, the Blank word is
well-defined (if exists), independent of the cable system after specifying a cyclic permutation
of all the bounded faces of $\gamma$.

\subsection{The Nie Word Construction}
\label{SS:Nie}

In an unpublished manuscript~\cite{nie2014,nie-emails}, Nie described how to
compute the minimum homotopy area between any two planar closed curves using the
language of geometric group theory.
Nie constructed a combinatorial word representing the planar closed curve,
followed by performing dynamic programming on the word based on a structure
called ``foldings'' (see Section~\ref{SS:norm}).
But first, let us describe the word construction.

Choose a point $p_i$ for each bounded face $f_i$ of $\curve$; denote the collection of points
%$\bigcup_i\set{p_i}$ 
as~\EMPH{$P$}.
Consider the punctured plane $X \coloneqq \RR^2 \setminus P$ and its fundamental group $\pi_1(X)$.
% Any closed curve $\gamma$ can be viewed as an element in $\pi_1(X)$, where all the curves homotopic to $\gamma$ is identified with $\gamma$ \note{do we want this?}.
Choose a set of generators $\Sigma$ for $\pi_1(X)$, where each $x_i$ in $\Sigma$ represents the
generator of $\pi_1(\RR^2\setminus \set{p_i}) \cong \ZZ$.

Now the fundamental group $\pi_1(X)$ is a free group over such generators, and
the curve $\gamma$ can be represented as a word over generators of $\pi_1(X)$.
However, there is more than one way to map each generator of
$\pi_1(\RR^2\setminus \set{p_i})$ into $\pi_1(X)$, due to the fact that in order
for $\pi_1(X)$ to be a group, one has to choose an endpoint $x_0$ and turn each
closed curve in $\pi_1(X)$ into a closed path connecting to $x_0$.
Nie never specified the choice of the connecting path because his algebraic formulation always gives the same answer
under any mapping of the generators.
%\hsien{Is this cable independence?  Check Nie's again.}

Nie's construction can also be interpreted combinatorially~\cite{nie-emails}.
Again consider the curve~$\curve$ as a four-regular plane graph.
Pick a tree-cotree pair $(T,T^*)$ of $\curve$ such that $T^*$ is a BFS-tree; naturally
 the tree $T$ is contractible.
For the sake of illustration, contract~$T$ into a single point~$t$; now each cotree edges is a
single closed path at $t$, enclosing at least one point in~$P$.
For our purpose of proving word equivalence, there are two natural sets of
generators for~$X \coloneqq \RR^2 \setminus P$:
\begin{itemize}\itemsep=0pt
    \item set of all cotree edges, and
    \item set of all face boundaries; i.e.\ sequences of cotree edges around each face containing $p_i$.
\end{itemize}
We now describe the change-of-basis between the two sets of generators in graph-theoretic terms.
Traverse $\curve$ from some arbitrary starting point and orient each edge of~$\curve$ accordingly.
Now, for each face $f_i$, define the \EMPH{boundary operator $\bdry$} by mapping face $f_i$ to the
signed cyclic sequence of edges around face $f_i$, where each edge is signed positively if it is
oriented counter-clockwise and negatively otherwise.

Now, write the curve $\gamma$ as a cyclic word over the cotree edges $T^*$ by traversing $\gamma$,
ignoring all tree edges in $T$.
We perform the following procedure inductively on the cotree $T^*$ to construct another cyclic word,
 this time as an element in the free group over the faces of~$\gamma$.
Starting from the leaves $f$ of $T^*$, rewrite each edge~$e$ bounding the face $f$ (that is, the
 dual of the unique edge connecting $f$ to its parent in~$T^*$) as a singleton word based on the
 index of $f$, with
 positive sign if edge $e$ is oriented counter-clockwise, or with negative sign otherwise.
Next, for any internal node $f$ of $T^*$, the boundary $\bdry f$ consists of a sequence of
(1) tree edges, (2) cotree edges to children of $f$ in $T^*$ denoted as $e_{1}, e_{2}, \ldots, e_{r}$,
and (3) (a unique) cotree edge to parent of $f$ denoted as $e_f$:
\[
    \bdry f = [e_f e_{1} e_{2} \ldots e_{r}].  %= [e_{j+1} \ldots e_{r} e_f e_{1} \ldots e_{j}].
\]
We can now inductively rewrite each child cotree edge $e_i$ as a free word $w_i$ over the faces
(and ignore all tree edges).
We emphasize that each word for the child cotree edge constructed inductively is a free word,
not a cyclic word.
Choose a particular but arbitrary way to break the cyclic sequence of faces and rewrite the equation:
\[
    e_f =  \bar{w}_{r} \cdot \cdots \cdot \bar{w}_{j+1} \cdot (\bar{w}_j)' \cdot \bdry f \cdot (\bar{w}_j)'' \cdot \bar{w}_{j-1}
     \cdot \cdots \cdot \bar{w}_1,
\]
where $\bar{w}_j = (\bar{w}_j)' (\bar{w}_j)''$ is a particular way of breaking the face word $\bar{w}_j$ into two.
%\note{Well, in fact $\bdry f$ might show up in the middle of a bundle of some children $e_i$. Rewrite.}
%
This gives us a free word over the faces for edge $e_f$, and thus by induction we have rewritten
$\gamma$ as a free word over the faces.
%\note{Blank word is cyclic?}
Finally, we can turn the free word back into a cyclic word, by observing that the cyclic permutation
of the constructed free word over the faces does not affect the element we are getting in
$\pi_1(X)$ (but as a side
 effect of choosing the basepoint $p_0$ of $\gamma$).

We call the resulting signed sequence of faces the \EMPH{Nie word} and denoted
as~$[\gamma]_N(\Sigma)$, where~$\Sigma$ is the choices we made when breaking up the
cyclic word at each cotree edge, referred to as a \EMPH{cycle flattening}.
Notice that the definition of $[\gamma]_N$ depends on how we choose to break the cyclic
edge sequences, and thus is not well-defined without specifying the choices.

\subsection{Word Equivalence}

Now we are ready to prove that the two words, one defined geometrically and the
other algebraically, are in fact equivalent.

\begin{theorem}[Word Equivalence]
    \label{lem:nie-cables}
    Let $\gamma$ be any plane curve.
    For a Nie word $[\gamma]_N(\Sigma)$ with a fixed cycle flattening~$\Sigma$, there is a
    managed cable system $\Pi$ such that the Blank word $[\gamma]_B(\Pi)$ is equal to $[\gamma]_N(\Sigma)$.
    Conversely, any managed cable system~$\Pi$ induces a cycle flattening $\Sigma$ such
    that $[\gamma]_B(\Pi)$ and $[\gamma]_N(\Sigma)$ are equal.
\end{theorem}

\begin{proof}
    First, fix a tree-cotree pair $(T,T^*)$ for $\gamma$ such that $T^*$ is a BFS-tree.
    Orient the edges of the cotree $T^*$ so that it is rooted at some fixed basepoint $p_0$.
    We prove the following statement by induction on the nodes of $T^*$ from leaves to the root,
    which implies the theorem:
    \begin{quote}
        The Blank subword corresponding to any cotree edge $e$ is the same as the Nie subword
        corresponding to $e$.
    \end{quote}

    To prove the statement, we will construct the cables in $\Pi$ gradually from
    each face to~$p_0$,
    at each step stopping at the cotree edge $e$ in $T^*$.
    Let $f$ be an arbitrary non-root node in~$T^*$, and edge $e$ be the unique edge from $f$ to its
    parent in~$T^*$.
    If $f$ is a leaf, $e$ is the only edge in $\bdry f$ that is not in tree $T$.
    This means, when we write $\bdry f$ using edges not in~$T$, we have~$\bdry f = \pm e$, with
    positive sign if $e$ is oriented counter-clockwise and negative sign otherwise.
    We draw the cable from the representative point in face $f$ to $e$; there is only one possible
    way to draw the cable up to isotopy.

    If $f$ is not a leaf, let $e_1, \dots, e_r$ be other non-tree edges on $\bdry f$ besides $e$ in
    counter-clockwise order around $\bdry f$, \emph{flipping their orientation defined by traversing $\gamma$ if necessary}.
    By induction hypothesis, the Blank subword of $e_i$ is the same as its Nie subword;
    denote the Blank (or Nie) subword of $e_i$ as $w_i$.
    This suggests that as we traverse $e_i$, the cables in $\Pi$ seem is exactly equal to $w_i$.
    Now we need to draw the cable $\pi_f$ from the representative point of $f$ to edge $e$.
    By construction of the Nie subword and the given cycle flattening $\Sigma$, the Blank subword
    on~$e$ must be of the form
    \[
        \bar{w}_{r} \cdot \cdots \cdot \bar{w}_{j+1} \cdot (\bar{w}_j)' \cdot f \cdot
        (\bar{w}_j)'' \cdot \bar{w}_{j-1} \cdot \cdots \cdot \bar{w}_1,
    \]
    where $\bar{w}_j = (\bar{w}_j)' (\bar{w}_j)''$ is a particular way of breaking the face word
    $\bar{w}_j$ into two. (See Figure~\ref{fig:induction}.)
    Because of the shortest path assumption, the collection of symbols inside
    each~$w_i$
    corresponds to exactly the faces contained within the region formed by
    cotree edge~$e_i$
    and the primal tree $T$.
    Thus, we extend all the cables intersecting the edges $e_1, \dots, e_r$
    to $e$, and create the representative point of face $f$ within the subregion of $f$ bounded
    by the last cable in $\bar{w}_{r} \cdot \cdots \cdot \bar{w}_{j+1} \cdot (\bar{w}_j)'$ and the first
    cable in $(\bar{w}_j)'' \cdot \bar{w}_{j-1} \cdot \cdots \cdot \bar{w}_1$.
    This way we can draw the cable $\pi_f$ so that the corresponding Blank subword is equal to
    the Nie subword.
    By induction, we have $[\gamma]_B(\Pi) = [\gamma]_N(\Sigma)$ for the constructed cable
    system~$\Pi$, which is managed and satisfies the shortest path assumption by the choice
    of $T^*$ being a BFS-tree.

    From the above construction we can recover a cycle flattening from a given managed
    cable system satisfying the shortest path property, thus the converse holds as well.
    %\note{The converse.  Need to show that any position of the representative point in $f$ gives a cycle flattening.}
\end{proof}

\begin{figure}[h!]
    \centering
    \includegraphics[width=.35\textwidth]{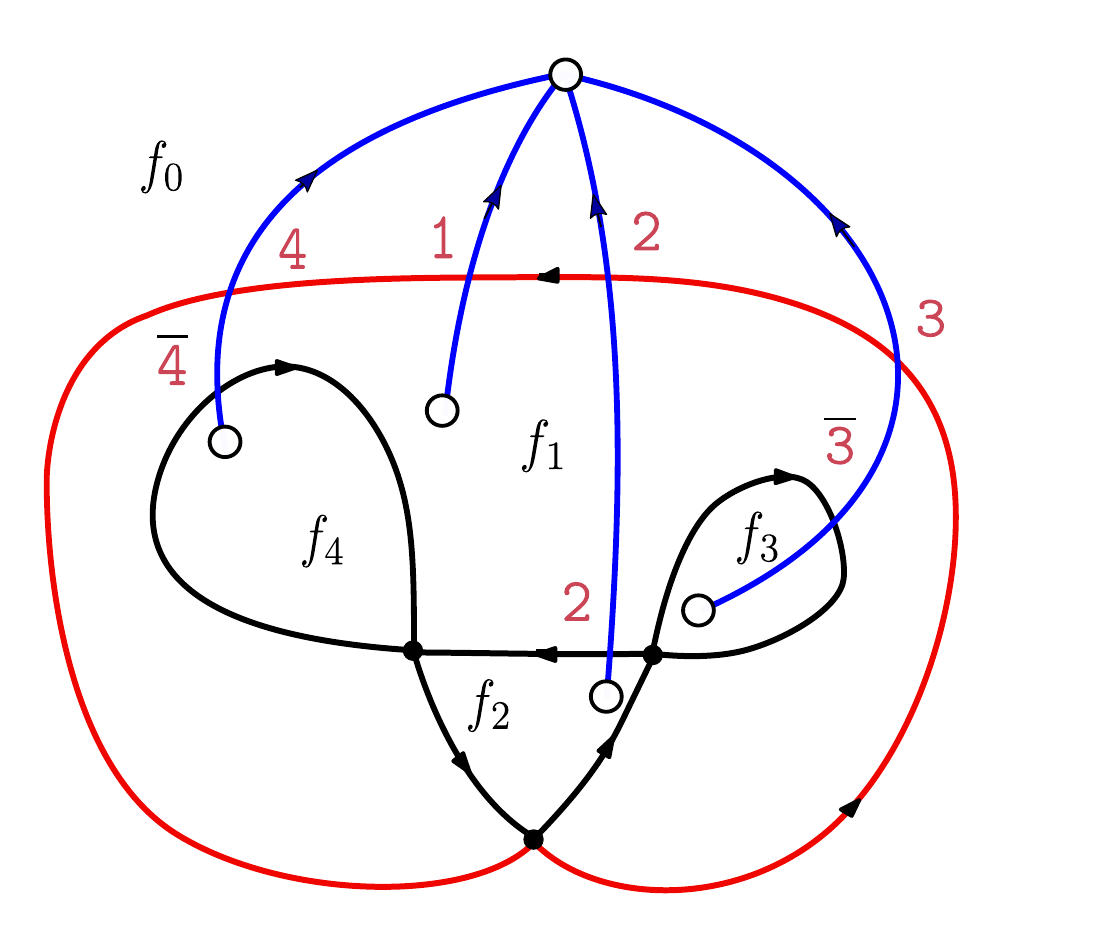}
    \caption{
        As we traverse the red edge $\gamma_r$ intersects $\str{3},\str{2}$, then $\str{1}$---the cable
        corresponding to $f_1$---then $\str{4}$.
        As we traverse the boundary of $f_1$ we traverse the red edge, followed by $\str{\overline{4}}$ and
        $\str{\overline{23}}$.
        This choice of cable system $\Pi$ corresponds to cycle flattening at $f_1$ as
        $\str{\overline{23}3214\overline{4}}$ by writing $f_1 = \overline{e_2} e_3 e_1 e_4$,
        or equivalently,  $\str{\overline{23}3214\overline{4}}=\overline{e_2} e_3 e_1 e_4$.
        The Blank subword on $\gamma_r$, with respect to the cycle flattening, is
        $e_1=\str{\overline{\overline{3}}}\str{2}\cdot\str{\overline{23}3214\overline{4}}\cdot\str{\overline{\overline{4}}}
        = \str{3214}$, as expected.
        }
        \label{fig:induction}
\end{figure}

% \begin{corollary}[Word Correspondence]\label{C:equivalence}
%     For any planar curve $\gamma$, there exists a managed cable system $\Pi$ and cycle flattening
%     $\Sigma$ such that there is a one-to-one correspondence between
%     the Blank word $[\gamma]_B(\Pi)$ and the Nie word $[\gamma]_N(\Sigma)$.
%     \todo{BTF/CW: isn't this the same as the theorem above?}
% \end{corollary}

\figref{order} gives an example demonstrating the one-to-one correspondence, for four different cable systems and cycle flattenings, on the same curve and
tree-cotree pair.
One consequence coming from the equivalence between two words and Lemma~\ref{lem:blank-unique}
%and the uniqueness of Blank word after fixing the cable ordering,
is that Nie word is uniquely determined after knowing
the subwords corresponding to cotree edges incident to the unbounded face.
This is not obviously from the definition of Nie word itself.

With this equivalence in hand, for the remainder of
the paper we refer to a Nie word or a Blank word of a curve $\gamma$ as the \EMPH{word}, denoted as $\word$ by dropping the subscripts.
Keep in mind, however, that the formal equivalence holds only when
the cable system is~managed.

\begin{figure}[h!]
    \captionsetup[subfigure]{justification=centering}
    \centering

    \begin{subfigure}[b]{0.24\textwidth}
        \includegraphics[width=\textwidth]{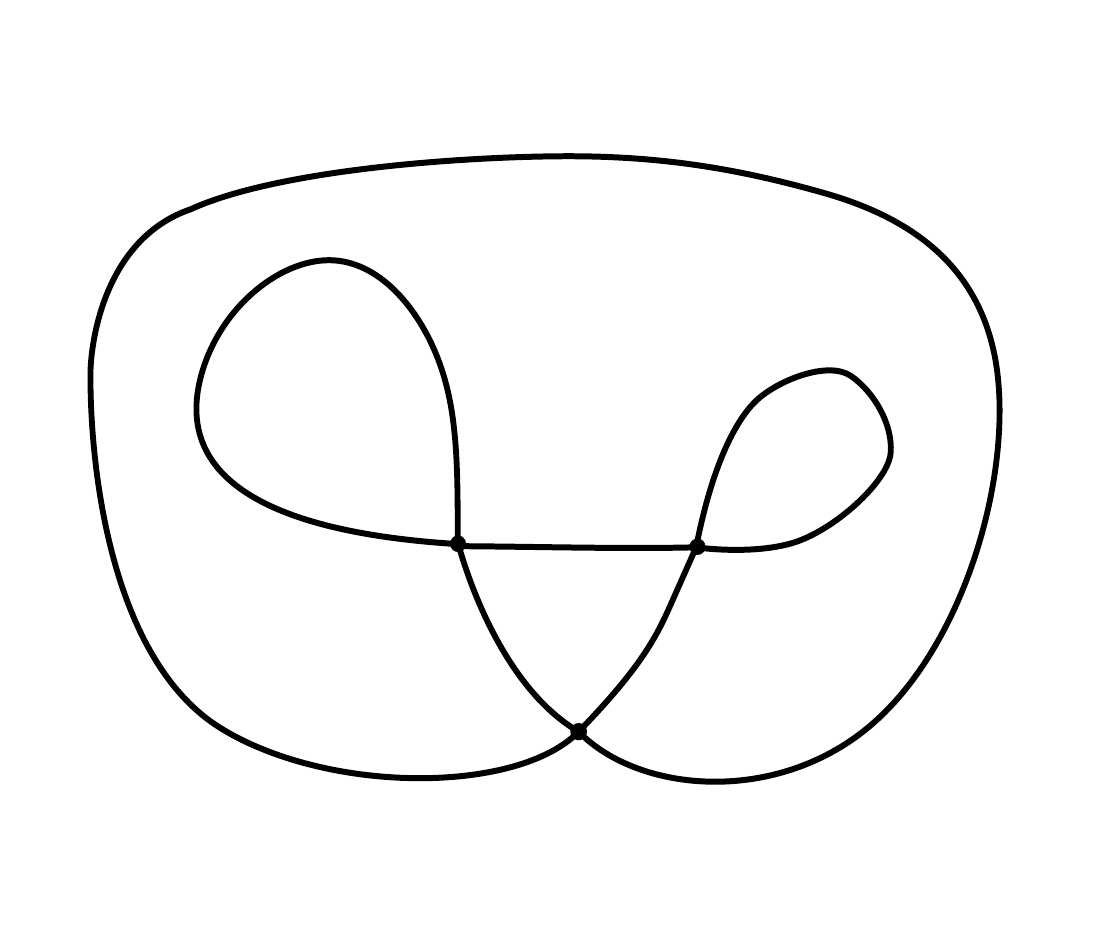}
       \subcaption{}\label{fig:order-base}
    \end{subfigure}
    \hspace{.5cm}
    \begin{subfigure}[b]{0.24\textwidth}
        \includegraphics[width=\textwidth]{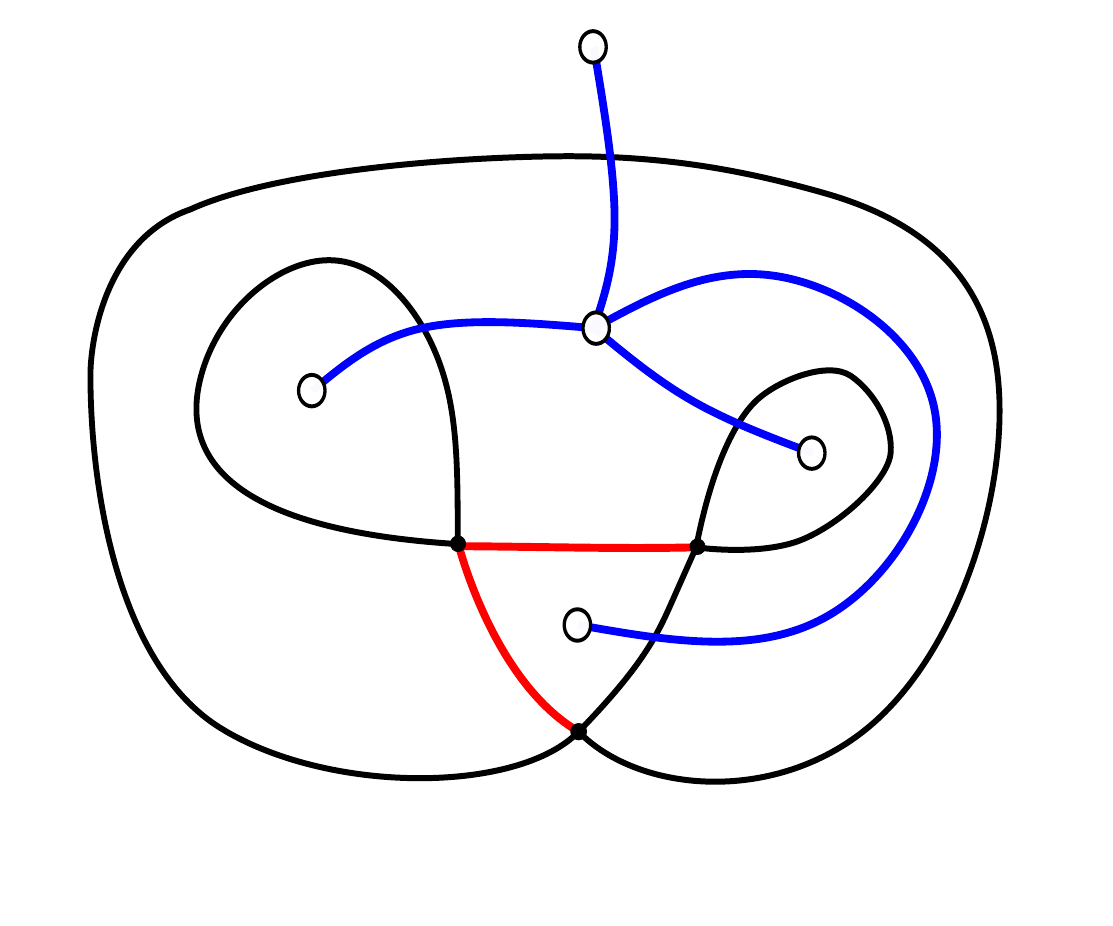}
        \subcaption{}\label{fig:order-trees}
    \end{subfigure}
    \hspace{.5cm}
    \begin{subfigure}[b]{0.24\textwidth}
        \includegraphics[width=\textwidth]{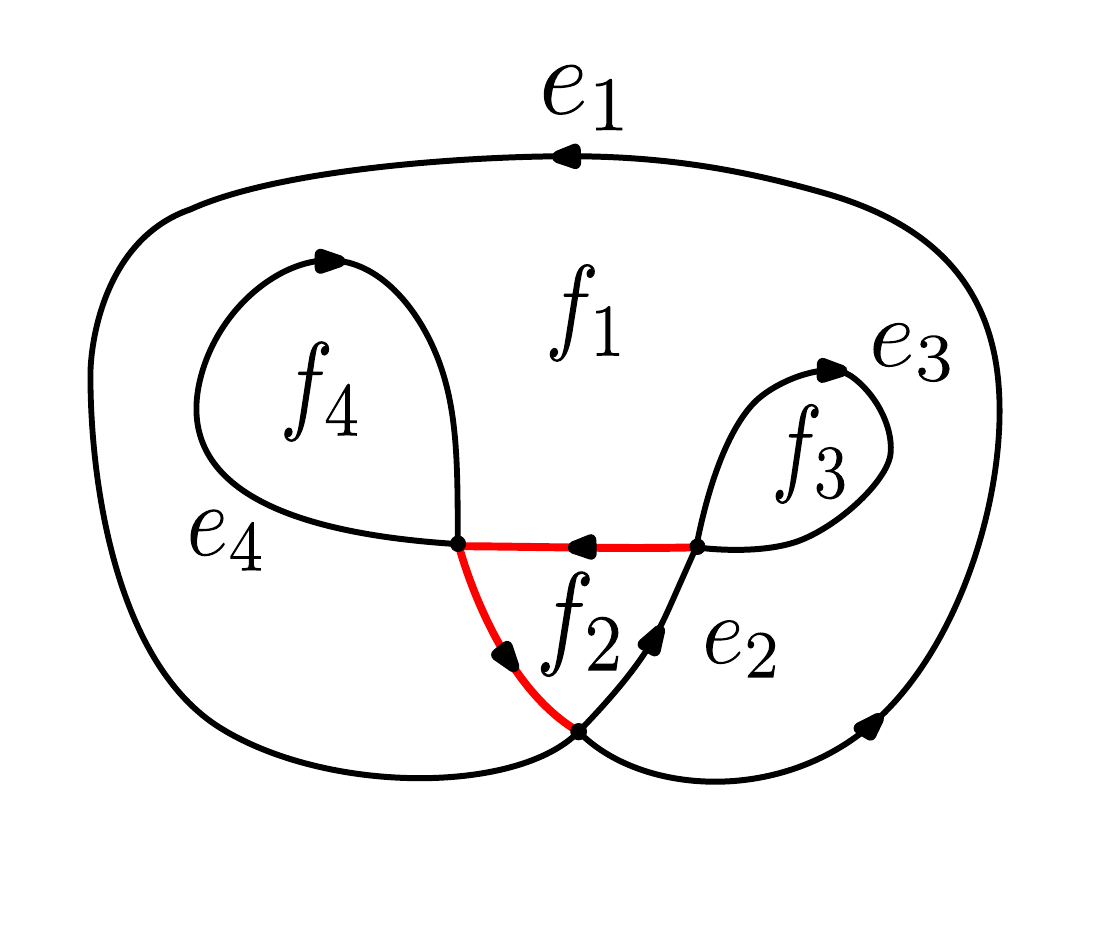}
        \subcaption{}\label{fig:order-labels}
    \end{subfigure}\\

    \begin{subfigure}[b]{0.24\textwidth}
        \includegraphics[width=\textwidth]{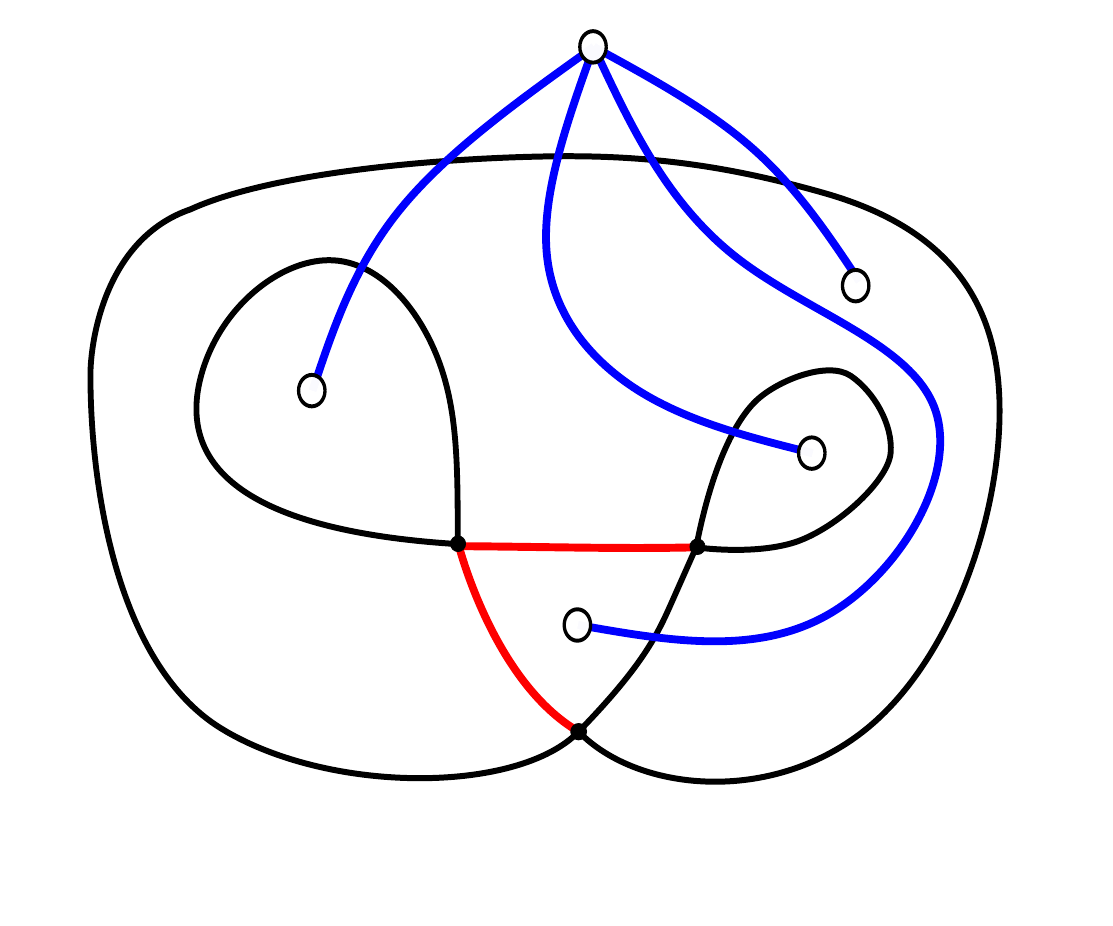}
        \subcaption{}\label{fig:order-blank-1}
    \end{subfigure}
    \begin{subfigure}[b]{0.24\textwidth}
        \includegraphics[width=\textwidth]{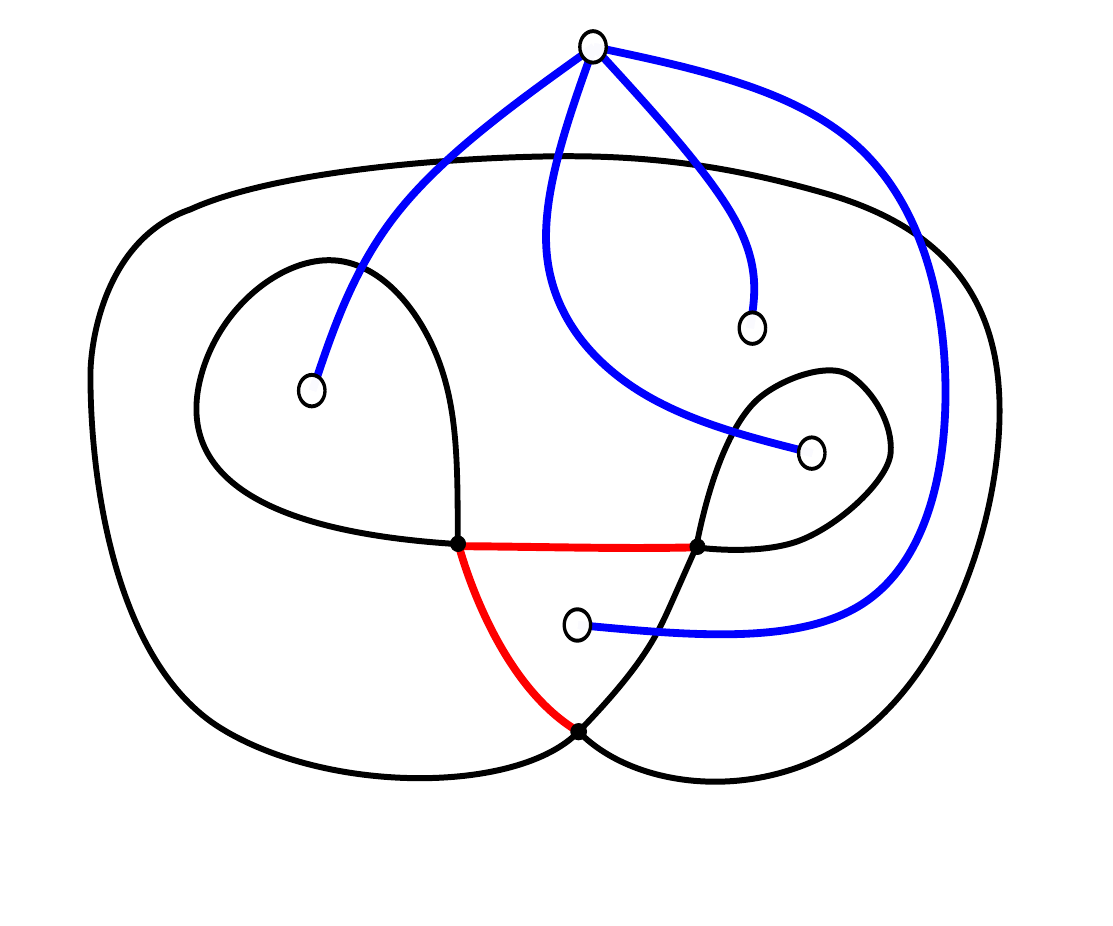}
        \subcaption{}\label{fig:order-blank-2}
    \end{subfigure}
    \begin{subfigure}[b]{0.24\textwidth}
        \includegraphics[width=\textwidth]{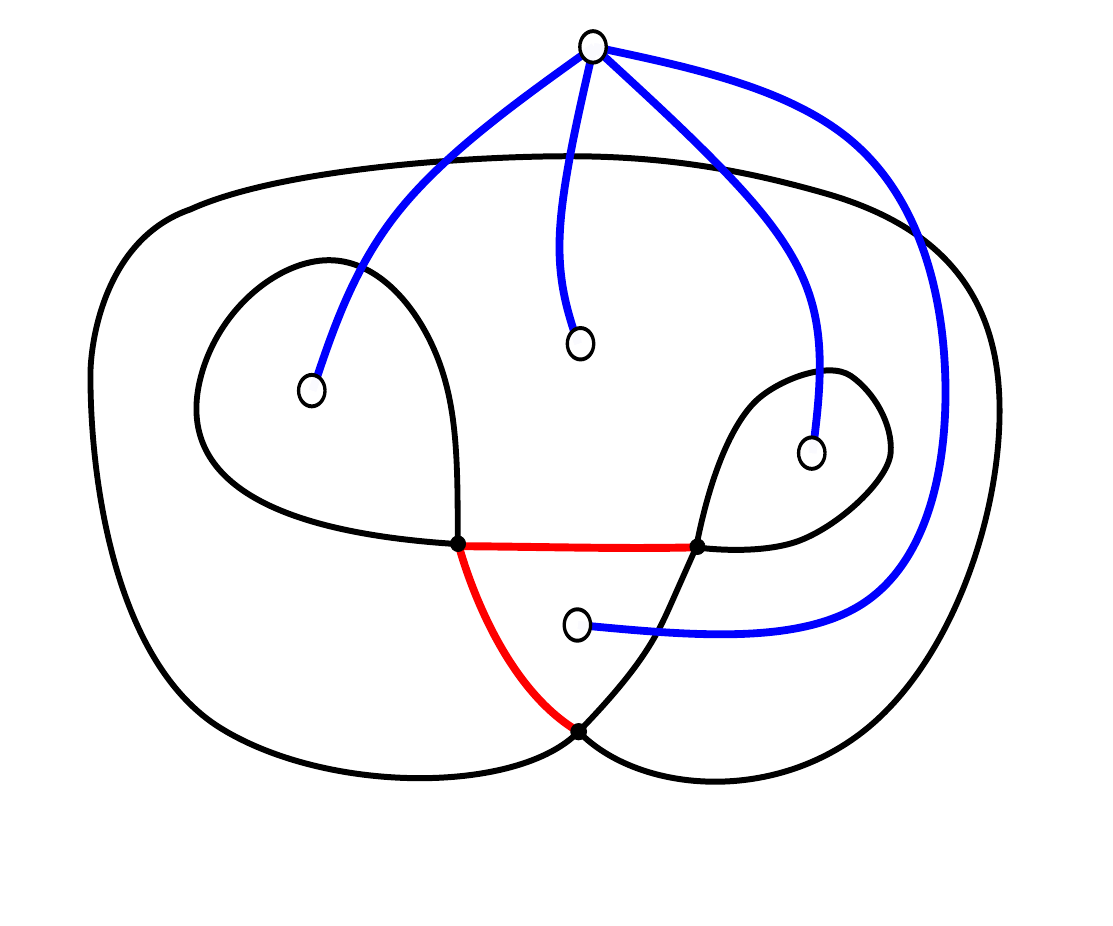}
        \subcaption{}\label{fig:order-blank-3}
    \end{subfigure}
    \begin{subfigure}[b]{0.24\textwidth}
        \includegraphics[width=\textwidth]{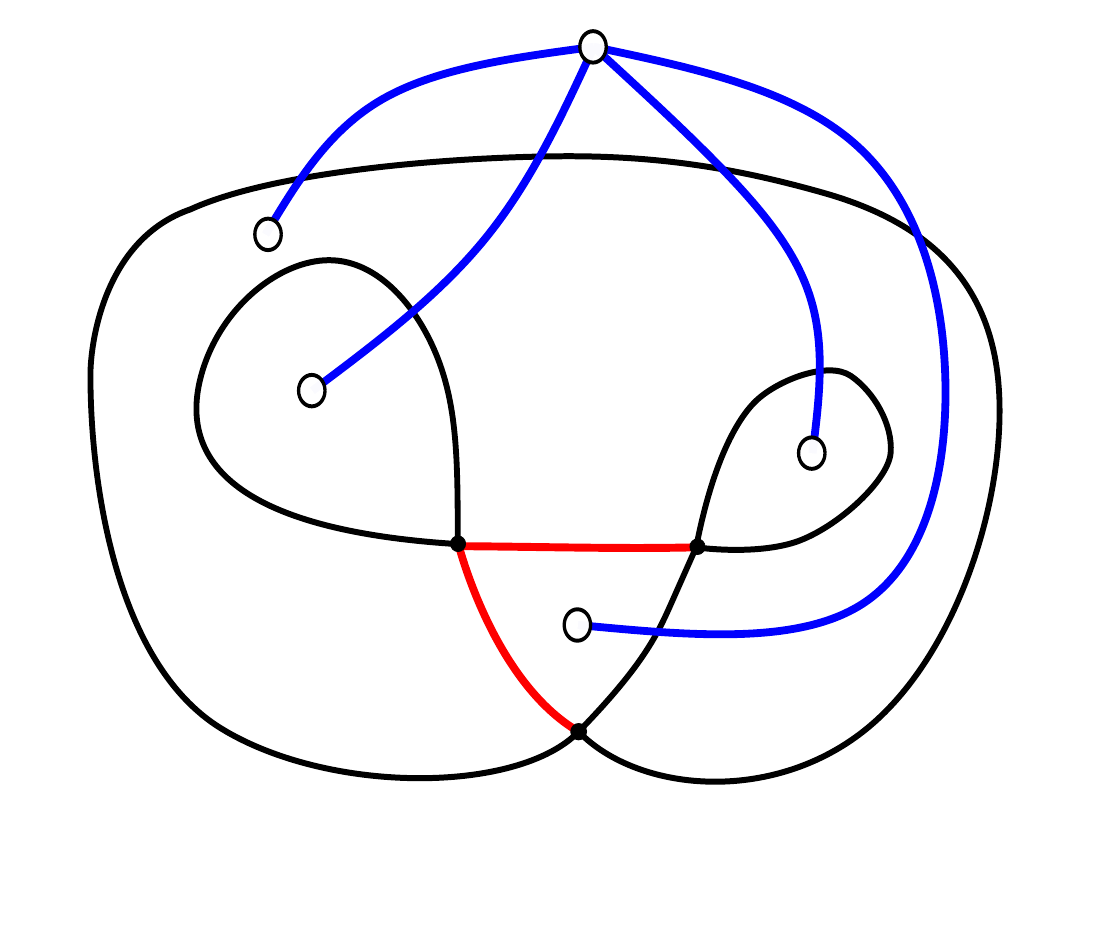}
        \subcaption{}\label{fig:order-blank-4}
    \end{subfigure}

    \caption{
        (\subref{fig:order-base}) A curve $\gamma$.
        (\subref{fig:order-trees}) A spanning tree in red and cotree in blue.
        (\subref{fig:order-labels}) A labeling of the coedges and faces.
        We have four ways to break the cyclic face sequence for $e_1$, represented using the cable system $\Pi$:
        (\subref{fig:order-blank-1})~$e_1=\partial{f_1}e_2\overline{e_3}\overline{e_4}$;
        (\subref{fig:order-blank-2})~$e_1=e_2\partial{f_1}\overline{e_3}\overline{e_4}$;
        (\subref{fig:order-blank-3})~$e_1=e_2\overline{e_3}\partial{f_1}\overline{e_4}$; and
        (\subref{fig:order-blank-4})~$e_1=e_2\overline{e_3}\overline{e_4}\partial{f_1}$.
        }
        \label{fig:order}
\end{figure}

%!TEX TS-program = pdflatex
%!TEX encoding = UTF-8 Unicode
%!TeX spellcheck = en-US
%!TeX root = ..\socg23.tex

%----------------------- Geometric Interpretation-------------------------------------
\section{Foldings and Self-Overlapping Decompositions}
\label{sec:interpretation}

In this section,
we give a geometric proof of the correctness to Nie's dynamic program.
To do so, we show that the minimum homotopy area of a curve can be computed from its Blank word using an algebraic quantity of the word called the \emph{cancellation norm}, which is independent of the drawing of the cables.
We then~show a minimum-area self-overlapping decomposition can be found in polynomial~time.

\subsection{The Cancellation Norm and Blank Cuts}
\label{SS:norm}

Given a (cyclic) word $w$, a \EMPH{pairing} is a letter and its inverse $(\str{f,\overline{f}})$ in $w$.
Two letter pairings, $(\str{f_1,\overline{f}_1})$ and $(\str{f_2,\overline{f}_2})$,
are \EMPH{linked} in a word if
the letter pairs occur in alternating order in the word,
$[\str{\cdots f_1 \cdots f_2 \cdots \overline{f}_1 \cdots \overline{f}_2 \cdots}]$.
A \EMPH{folding}
%\footnote{The word `folding' is due to applications to RNA sequencing \cite{folding-1980}.}
of a word is a set of letter pairings
such that no two pairings in the set are linked.
For example, in the word $[\str{\overline{2}31546\overline{5}4\overline{6}2\overline{3}}]$
the set $\{(\str{5,\overline{5}}),(\str{\overline{3},3})\}$ is a folding
while~$\{(\str{5,\overline{5}}),(\str{6,\overline{6}})\}$~is~not.

The cancellation norm is defined in terms of pairings.
The norm also applies in the more general setting where every
letter has an associated nonnegative weight.
% $wt: \set{\text{letters}}\to \R_{\geq 0}$.
A letter is \EMPH{unpaired} in a folding if it does not participate in any pairing of the folding.
For a word of length $m$, computing the cancellation norm 
takes~$\mathcal{O}(m^3)$ time and $\mathcal{O}(m^2)$
space~\cite{cancellation-norm-def,folding-1980}.
Recently, a more efficient algorithm for computing the cancellation norm
appears in Bringmann \etal~\cite{bringmann_truly_2019};
this algorithm uses fast matrix multiplications and runs in $\mathcal{O}(m^{2.8603})$ time.
% \hsien{Really...? State in terms of matrix multiplication constant.}
% \brad{It is quite nasty to write in terms of matrix mult. (bottom of
% sec. 3.4 in \cite{bringmann_truly_2019}}

The \EMPH{weighted cancellation norm} of a word $w$ is defined to be the minimum sum of weights of all the unpaired letters in $w$ across all foldings of $w$~\cite{cancellation-norm-def,folding-1980}.
% The \EMPH{weighted cancellation norm} of a word $w = [\ell_1\ell_2,...,\ell_m]$ is defined to be
% \[
% \norm{w} \coloneqq \min_\mathcal{F} \sum_i \textrm{wt}(\ell_i)
% \]
% where $\mathcal{F}$ is the set of all foldings of $w$ and $i$ ranges over all
% unpaired positions in $w$ \cite{cancellation-norm-def,folding-1980}.
If $w$ is a word where each letter~$\str{f_i}$ corresponds to a face $f_i$ of a curve,
we define the weight of $\str{f_i}$ to be $\textrm{Area}(f_i)$.
The \EMPH{area of a folding} is the sum of weights of all the unpaired symbols in a folding.
The weighted cancellation norm becomes~
\(
\EMPH{$\norm{w}$} \coloneqq \min_\mathcal{F} \sum_i \textrm{Area}(f_i)
\)
where $\mathcal{F}$ is the set of all foldings of $w$ and $i$ ranges over all unpaired letter in $w$.
%
%The algorithm computes the minimum homotopy area between any two curves, for our purposes, one of curves is a point.

A dynamic program, similar to the one for matrix chain multiplication, is applied
on the word.
Let $w=f_1f_2\dotsm f_\ell$ where $\ell\geq 2.$ Assume we have
computed the cancellation norm of all subwords with length less than
$\ell$. Let $w'=f_1f_2\dotsm f_{\ell-1}.$  If $f_\ell$ is not the inverse
of $f_i$ for $1\leq i\leq \ell-1,$ then $f_\ell$ is unpaired and $||w||=||w'||+Area(f_\ell).$
Otherwise, $f_\ell$ participates in a folding and there exits at least one $k$
 where $1\leq k\leq \ell-1$ and $f_k=f_\ell^{-1}.$ Let~$w_1=f_1\dotsm f_{k-1}$
  and~$w_2=f_{k+1}\dotsm f_{\ell-1}.$ Then, we
find the $k$ that minimizes $||w_1||+||w_2||.$ We~have
$$||w||=\min\{||w'||+Area(f_\ell),\underset{k}{\min}\{||w_1||+||w_2||\}\}$$
Nie shows that the weighted cancellation norm whose weights
correspond to face areas is equal to the minimum homotopy area using the
triangle inequality and geometric group theory. 
Our proof that follows is more geometric and leads to a natural homotopy that achieves
the minimum area.

\medskip
%\paragraph*{Blank cuts}

We now show how to interpret the cancellation norm geometrically.
Let $(\str{f},\str{\bar{f}})$ be a face pairing in a folding of the word $[\gamma]_B(\Pi)$ for some cable system $\Pi$.
Denote the cable in~$\Pi$ ending at face $f$ as $\pi_f$.
Cable $\pi_f$ intersects $\gamma$ at two points corresponding to the pairing $(\str{f},\str{\bar{f}})$, which we denote as $p$ and $q$ respectively.
Let~$\pi'_f$ be the simple subpath of $\pi_f$ so that
$\pi'_f(0)=q$ and $\pi'_f(1)=p$.
We call~$\pi'_f$ a \EMPH{Blank cut}~\cite{blank,evansFasyWenk,marx74} (see \figref{blank-cut}).
Any face pairing defines a Blank cut, and the result of a Blank cut produces two curves each with fewer faces than the original curve: namely,
$\gamma_1$ which is the restriction of $\gamma$ from $q$
to $p$ following by the reverse of path $\pi'_f$,
and $\gamma_2$ which is the restriction of $\gamma$ from $p$ to $q$ followed~by~path $\pi'_f$.

\begin{figure}[t]
    \captionsetup[subfigure]{justification=centering}
    \centering
    \begin{subfigure}[t]{0.2\textwidth}
    	\includegraphics[width=\textwidth]{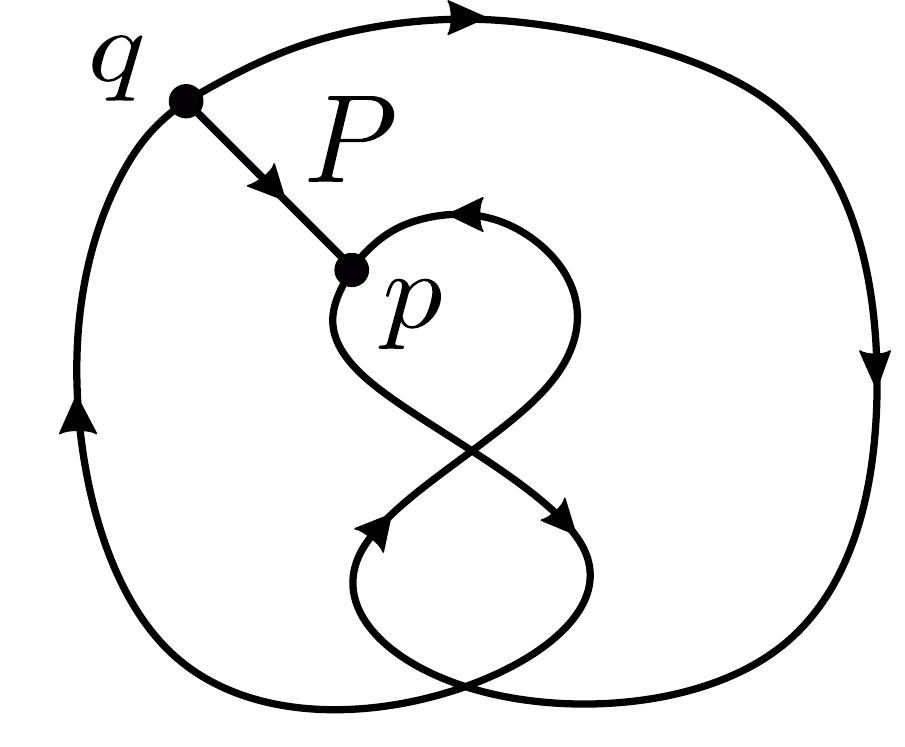}
	    \subcaption{}
    \label{fig:pre-cut}
    \end{subfigure}
    \hspace{1.5cm}
    \begin{subfigure}[t]{0.2\textwidth}
  	    \includegraphics[width=\textwidth]{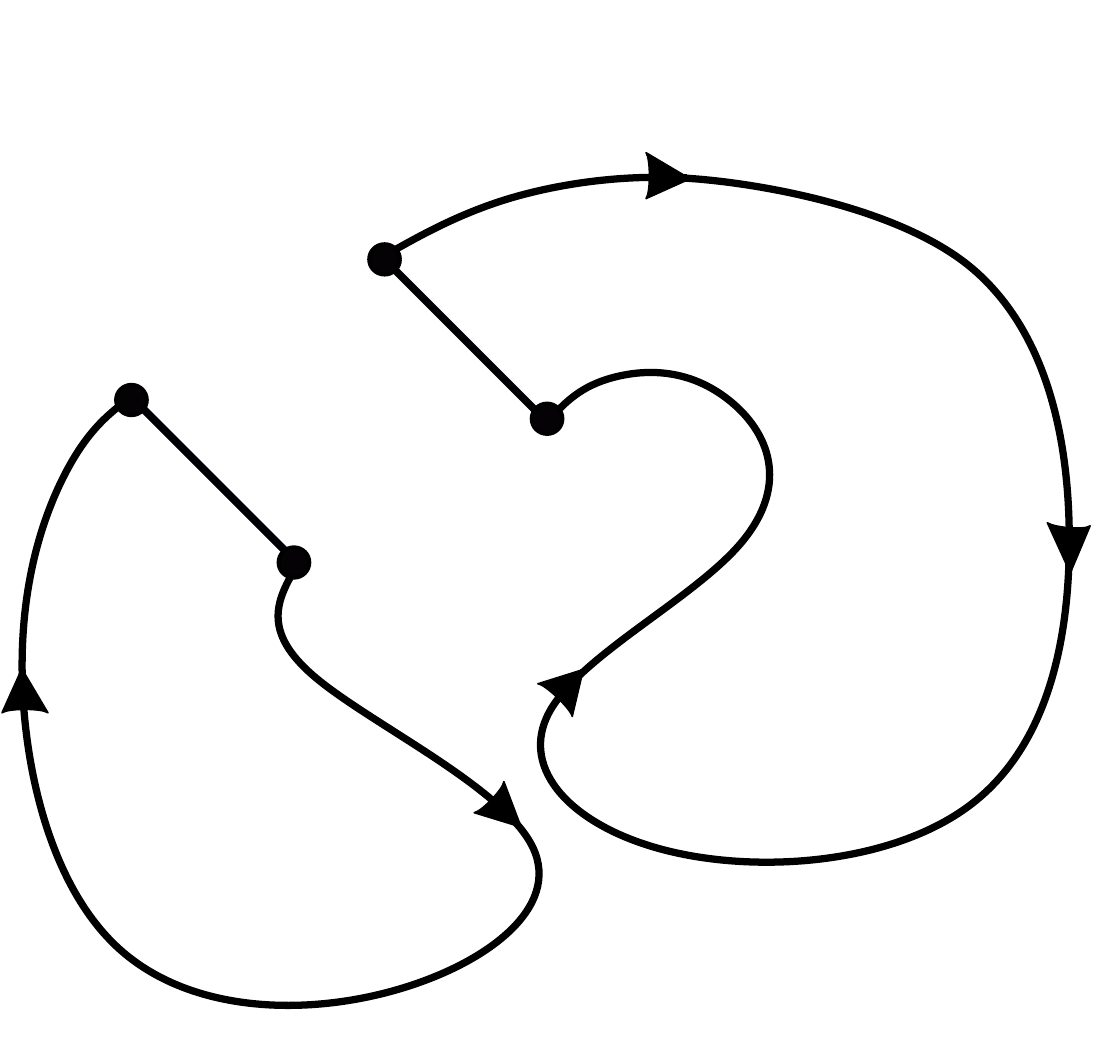}
	    \subcaption{}
    \label{fig:post-cut}
    \end{subfigure}

	\caption{(\subref{fig:pre-cut}) A curve with labeled path $P$.
	(\subref{fig:post-cut}) The two induced subcurves from cutting along $P.$}
\label{fig:blank-cut}
\end{figure}

In order to not partially cut any face, we require all Blank cuts to occur along
the boundary of the face being cut.
When cutting face $f_i$ along path~$\pi_j$,
we reroute all cables crossing the interior of $f_i$, including $\pi_j$ but excluding $\pi_i$, along the boundary of $f_i$ through an \emph{isotopy},
so that no cables %(besides $\pi_i$)
intersect $\pi_i$.
Lemma~\ref{lem:isotopy} ensures that the reduced Blank word remains unchanged.
See \figref{re-route} for an example.
Notice that different cables crossing $f_i$ might be routed around different sides of $f_i$ in order to avoid intersecting cable $\pi_i$ and puncture $p_i$.
% If necessary, we also reroute any cables that lie between
% the initial and final position of $\pi_j$.
% This can always be done because $\pi_j$ is rerouted
% to avoid the puncture in $f_i$, any other affected cables
% will have $\pi_j$ between them and the puncture in $f_i$,
% see \figref{re-route}\subref{fig:reroute2}.
This way, we ensure the face areas of the subcurves are in one-to-one correspondence with the symbols in the subwords induced by a folding.

\begin{figure}[t]
    \captionsetup[subfigure]{justification=centering}
    \centering
    \begin{subfigure}[b]{0.2\textwidth}
        \includegraphics[width=\textwidth]{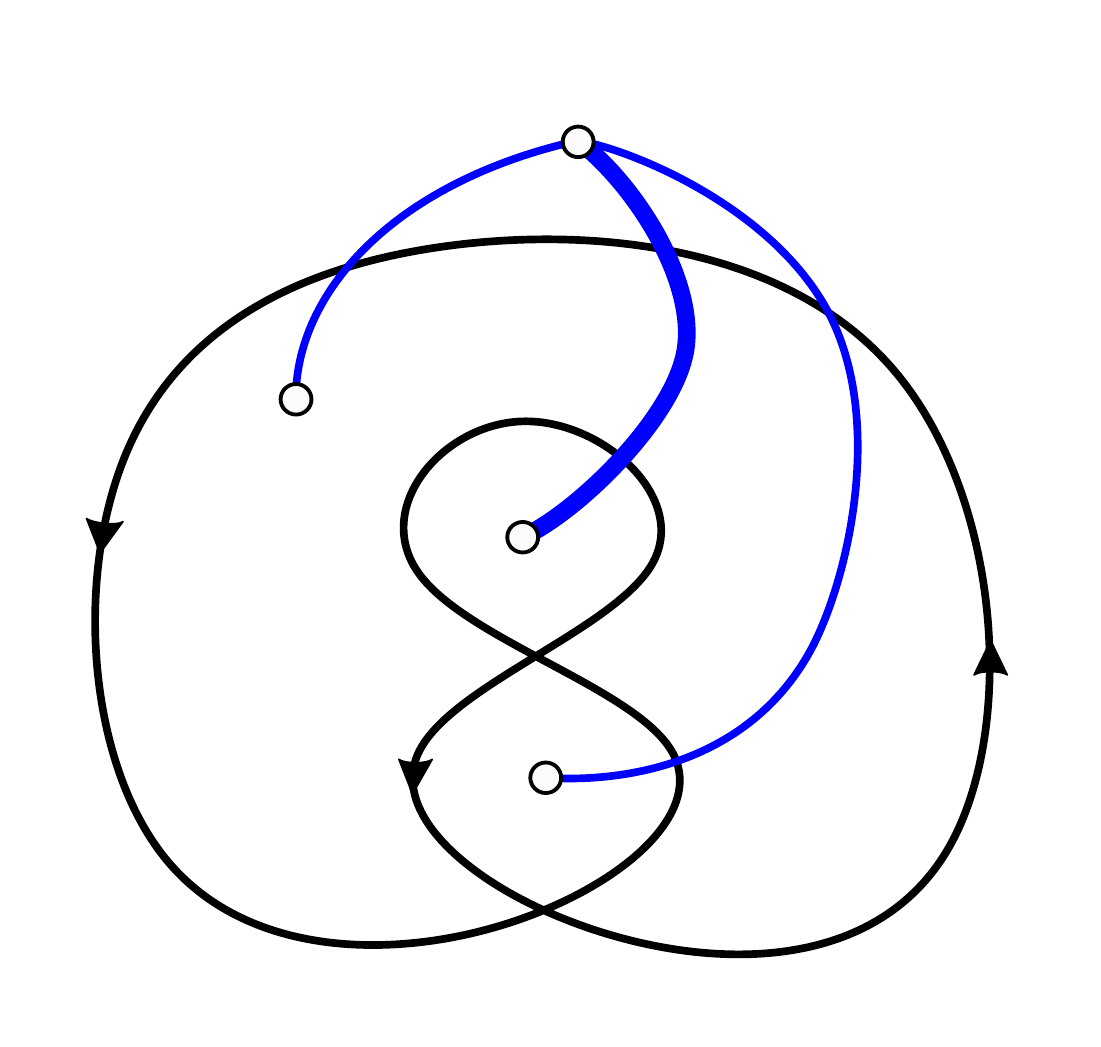}
        \subcaption{}
        \label{fig:reroute1}
    \end{subfigure}
    \begin{subfigure}[b]{0.2\textwidth}
        \includegraphics[width=\textwidth]{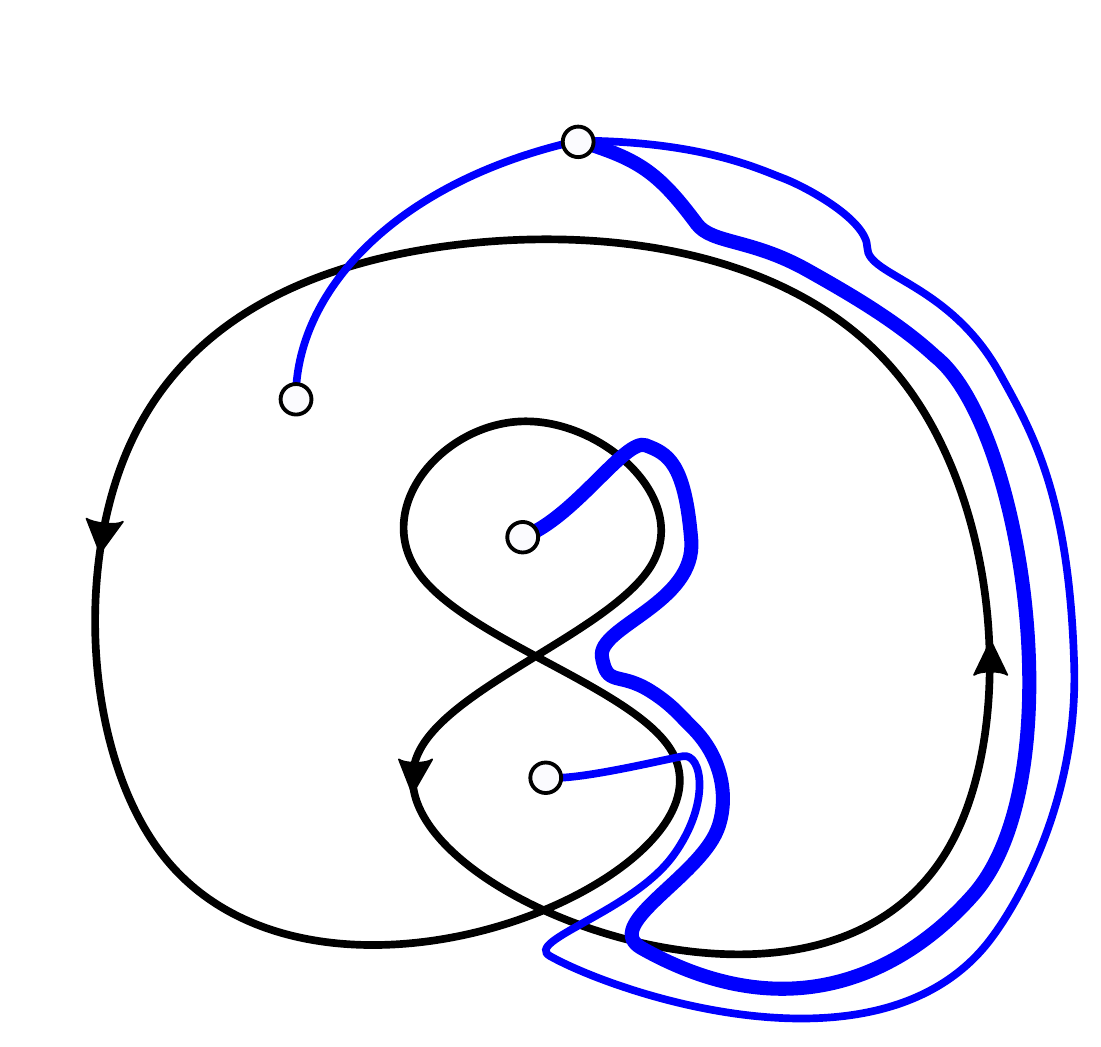}
        \subcaption{}
        \label{fig:reroute2}
    \end{subfigure}
    \begin{subfigure}[b]{0.2\textwidth}
        \includegraphics[width=\textwidth]{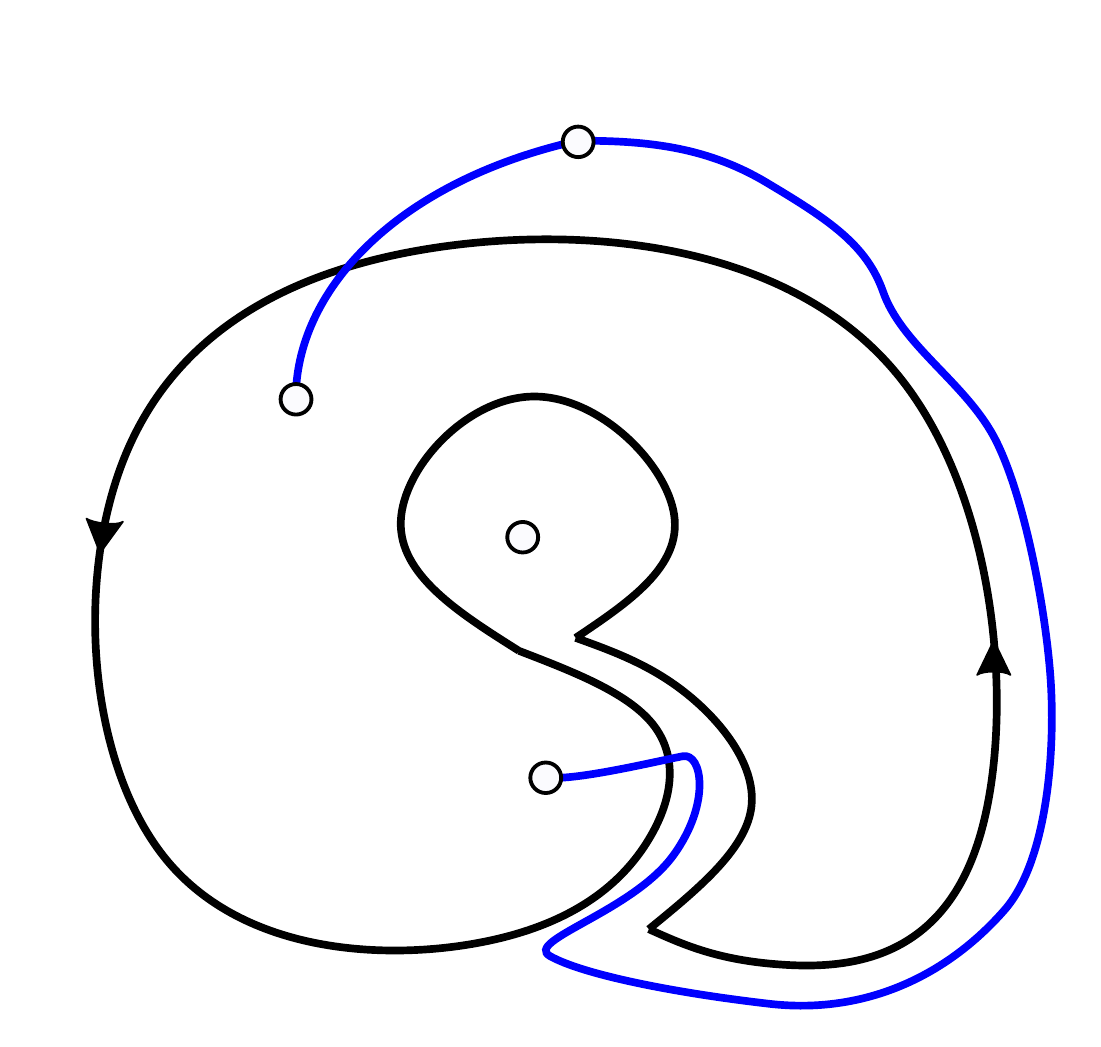}
        \subcaption{}
        \label{fig:reroute3}
    \end{subfigure}
    \begin{subfigure}[b]{0.2\textwidth}
        \includegraphics[width=\textwidth]{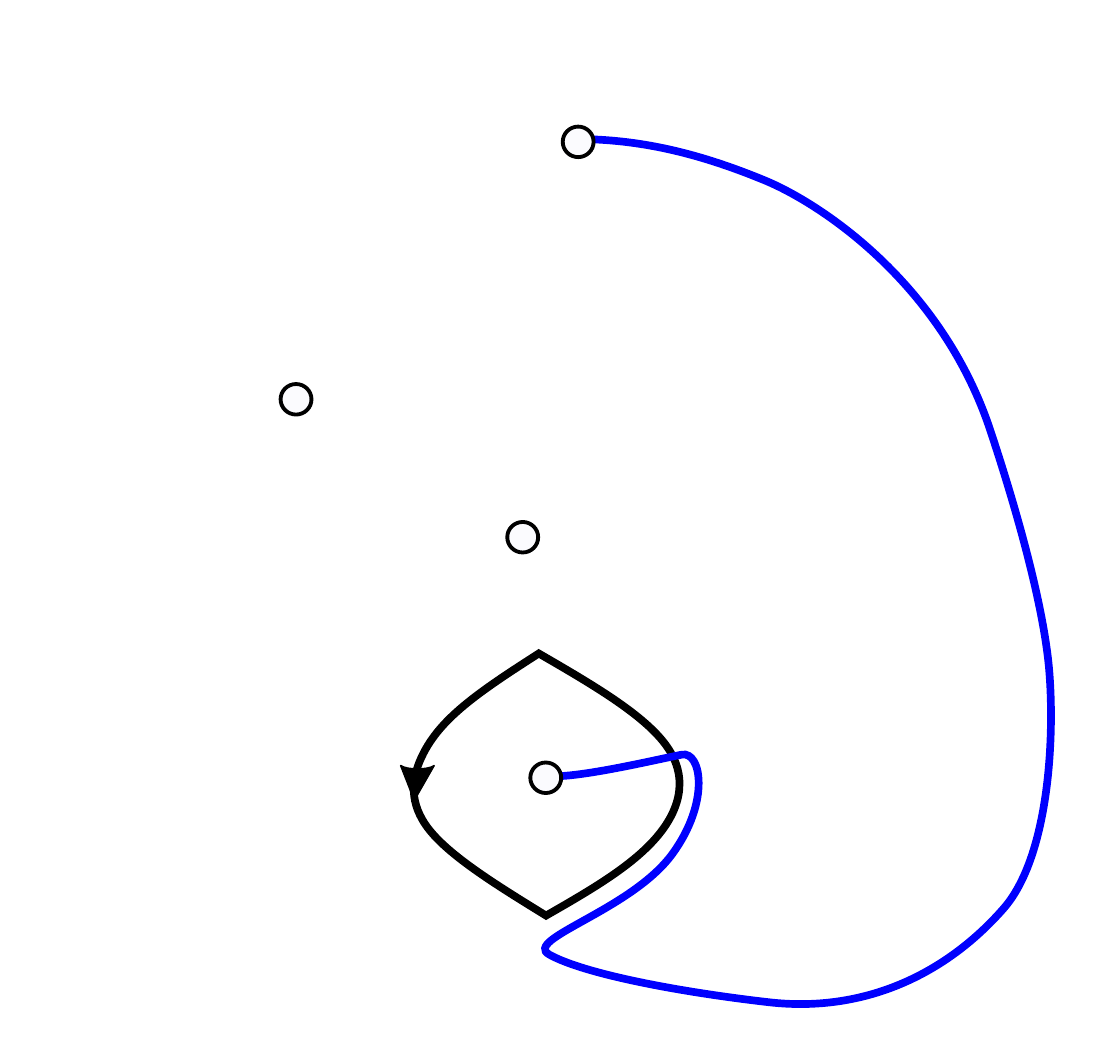}
        \subcaption{}
        \label{fig:reroute4}
    \end{subfigure}
    \caption{(\subref{fig:reroute1}) A curve with cables.
    (\subref{fig:reroute2}) Isotopy the cables to not partially cut any faces.
    (\subref{fig:reroute3}) One subcurve resulting from cutting along the middle cable.
    The curve is weakly simple and there are two cables in this face.
    %The area of this curve is the sum of the areas of the faces that combined
    %to make the curve.
    (\subref{fig:reroute4}) The other subcurve.
\label{fig:re-route}}
\end{figure}

Using the concept of Blank cut we can determine if a curve is self-overlapping.
A subword~$\sigma$ of $w$ is
\EMPH{positive} if $\sigma = \str{f_1f_2\ldots f_k}$, where each letter
$\str{f_i}$ is positive.
A pairing $(\str{f},\str{\bar{f}})$ is \EMPH{positive} if one of the two
subwords of the (cyclic) word $w$ in between the two symbols
$\str{f},\str{\bar{f}}$ is positive; in other words, $w =
[\str{f}p\str{\bar{f}}w']$ for some positive word $p$ and some word $w'$.
%\end{itemize}
%
A folding of $w$ is called a \EMPH{positive folding}\footnote{Blank called these
pairings \emph{groupings}} if all pairings in $w$ are positive, and the word
constructed by replacing each positive pairing (including the positive word in-between) $\str{f} p \str{\bar{f}}$ in the
folding with the empty string is still positive.
%Replacing a positive pairing $\str{f} p \str{\bar{f}}$ with the identity is called a \EMPH{positive folding}\footnote{Blank called these pairings \emph{groupings}}.
Words that have positive foldings are called \emph{positively foldable}.
Blank established the characterization of self-overlapping curves through Blank cuts.

\begin{theorem}[Self-Overlapping Detection~{\cite{blank}}]\label{thm:blank}
    Curve $\curve$ is self-overlapping if and only if $\curve$ has rotation
    number 1 and $[\curve]_B(\Pi)$ is positively foldable for any shortest $\Pi$.
\end{theorem}

However, we face a difficulty when interpreting Nie's dynamic program geometrically.
In our proof we have to work with \emph{subcurves} (and their extensions) of the original curve and the induced cable system.
For example, after a Blank cut or a vertex decomposition, there might be multiple cables connecting to the same face creating multiple punctures per face,
and cables might not be managed or follow shortest paths to the unbounded face
(see \figref{reroute3} and \figref{not-short}).
In other words, the subword corresponding to a subcurve with respect to the induced cable system might not be a regular Blank word (remember that Blank word is only well-defined when the cable system is managed, all cables are shortest paths, and the cable ordering is fixed; see Section~\ref{lem:blank-unique}).
To remedy this, we tame the cable system first by rerouting them into another cable system that is managed and satisfies the shortest path assumption, then merging all the cables ending at each face.
We show that while such operations change the Blank word of the curve, the cancellation norm of the curve
%the rotation number of the curve
and the positive foldability does not change.
We summarize the property needed below.
%for the rest of the proof;
%see Section \ref{SS:independence} for specific details.

\begin{lemma}[Cable Independence]
\label{lem:independence}
Let $\curve$ be any curve with two cable systems~$\Pi$ and $\Pi'$ such that the weights of the cables in $\Pi$ ending at any fixed face sum up to the ones of $\Pi'$.
Then any folding $F$ of $\word(\Pi)$ can be turned into another folding $F'$ of $\word(\Pi')$, such that the area of the two foldings are identical.
As a corollary, the minimum area of foldings (the cancellation norm) of $[\curve](\Pi)$ and the existence of a positive folding of $[\curve](\Pi)$ are independent of the choice~of~$\Pi$.
\end{lemma}
%In \appendref{foldings}, 
Next, we prove that for each folding there is a homotopy with~equal~area.

\begin{restatable}[Folding to Homotopy]{lemma}{foldingarea}
\label{lem:folding-area}
	Let $\curve$ be a curve and $\Pi$ be a managed cable system satisfying the shortest path assumption,
	and let $F$ be a folding of~$[\curve](\Pi)$.
	There exists a null-homotopy of $\curve$ with area
	equal to the area of $F$.
\end{restatable}

\begin{proof}
	We induct on the number of pairings in $F$.
  	  Consider a folding $F$ with~$k$ pairings, and let $(\str{f},\str{\bar{f}})$ 
    	be a pairing in $F$ that does not contain any other pairing in between.
	A Blank cut along this pairing decomposes
	the curve into two subcurves, at least one of which does not
	contain any pairings. 
  	Call this subcurve~$\curve_1$.
	We will contract $\curve_1$ to the cut using a homotopy
	with area equal to the sum of area of all the unpaired letters in the subword
	of $\curve_1$; such homotopy exist by the base case of induction.
	However, after we contract $\curve_1$,
	the cable system on the remaining curve might no longer be managed or shortest, and there might be multiple cables in some faces.
	Therefore we strengthen the inductive hypothesis by assuming that for any subword of $[\gamma](\Pi)$ and its corresponding subcurve $\curve'$ of~$\gamma$ with the induced (unmanaged multi-)cable system, and a folding $F'$ with less $k$ pairings on the subword,
	there exists a null-homotopy of $\curve'$ with area equal to the area of the folding $F'$, and the destination of the null-homotopy can be anywhere on $\curve'$. 
	
    Let $w'$ be the subword of $[\gamma](\Pi)$ by removing the subword between the pairing~$(\str{f},\str{\bar{f}})$, folding $F'$ on subword $w'$ be constructed from $F$ by removing the pairing $(\str{f},\str{\bar{f}})$.
    Let $\curve'$ denote the curve that results from contracting
	$\curve_1$ to the cut and let $\Pi'$ denote the cable system
	induced on $\curve'$ by $\Pi$.
	By the (strengthened) induction hypothesis $\curve'$ has a null-homotopy with area equal to the one of $F'$.
    Combining with the null-homotopy of $\gamma_1$ and we are done.

    For the base case when $F$ is the empty pairing, 
    if the cables are not managed, construct a managed cable system satisfying the shortest path assumption $\Pi^*$ for $\curve$ and 
    apply Lemma~\ref{lem:independence};
    the sum of weights of all letters in the new Blank word remains unchanged.
    % to ensure that the null-homotopy to be constructed can be done without loss of generality to be with respect to the managed cable system $\Pi^*$.
    %
    To construct the null-homotopy,
    decompose the curve into depth cycles by performing a smoothing at each vertex~\cite{changErickson17},
	to obtain a null-homotopy with area equal to the depth area.
	Since each cable in $\Pi^*$ follows a shortest path to the unbounded face, the number of times an unsigned letter
	appears in the word is equal to the depth of the face,
	and thus the homotopy area is equal to the area of the folding.
    In fact, we can choose any point on the curve to be the destination of the null-homotopy.

\end{proof}

What we are left with is to prove Lemma~\ref{lem:independence}:
Given a curve $\gamma$ and a cable system $\Pi$, if we have a folding on the face word $[\gamma](\Pi)$, there is an equivalent folding of the new word $[\gamma](\Pi')$ that has the same area if we choose a different cable system.
As a result, the cancellation norm and the positive foldability of a word
are independent to the cable system chosen.
Notice that it is sufficient to assume $\Pi'$ to be a managed cable system with shortest path assumptions and single cable to each face.

Any two cable systems can always be connected by a sequence of isotopy and  order switching between two adjacent cables, followed by a redrawing of the cables induced by a homeomorphism of the plane fixing the punctures $\set{p_i}$, which fixes the cable ordering and the disjointness between cables but the isotopy classes of the cables change.
We emphasize that either operation will change the Blank word, but we can always find a folding of the new word that preserves the area.

%\paragraph*{Switching cable orders.}

Next, we show that switching two adjacent cables preserves the area and positivity of the folding

\begin{restatable}[Switch Invariance]{lemma}{normswitch}
\label{lem:norm-switch}
    Given curve $\curve$, a cable system $\Pi$, and a folding $F$ on $\word(\Pi)$.
    Switching the order of two adjacent cables produces another folding on the new word with equivalent area.
    Furthermore, the new folding is positive if and only if $F$ is.
\end{restatable}

\begin{proof}
Consider two cables $\pi_f$ and $\pi_g$ in $\Pi$ that are adjacent in the rotation system.
Assume $\pi_g$ is in the clockwise direction of $\pi_f$
in the rotation system and we are trying to move $\pi_f$ across $\pi_g$.
Let $w \coloneqq [\curve]_B(\Pi)$ denote the Blank word before the order of the cables are switched and $w'$ the word after the cables are switched; by Lemma~\ref{lem:isotopy} the words are well-defined up to isotopy classes of $\pi_f$ and~$\pi_g$.
First, we argue that, without loss of generality, we can assume the drawing of $\pi_g$ follows $\pi_f$ in an $\e$-neighborhood before continuing towards puncture $p_g$, by drawing $\pi_g$ still counter-clockwise to and $\e$-close to $\pi_f$ from $p_0$ to $p_f$ and back to $p_0$, followed by the original drawing of $\pi_g$.  
The new drawing is disjoint from the rest of the cables and is isotopic to the original drawing.
This way we can ensure any instance of symbol $\str{f}$ in $w$ is followed immediately by a $\str{g}$, and
any instance of $\str{\bar{f}}$ has a $\str{\bar{g}}$ proceeding in the word.
(Notice we cannot necessarily say the same about $\str{g}$ because cable $\pi_g$ continues after reaching $p_f$.)

% Perform an isotopy so that $\pi_g$ follows beside $\pi_f$ from the puncture $p_f$ to the exterior face.
% Any instance of $\str{f}$ in $W$ has $\str{g}$ following in the word, and
% any instance of $\str{\bar{f}}$ has $\str{\bar{g}}$ proceeding in the word.

Let $F$ be a folding of $w$, we construct a folding $F'$ of $w'$
with equal area.  We separate into two cases.
If no instance of $(\str{f},\str{\bar{f}})$ is in $F$, then all pairings in $F$ can be paired in $W'$ and we can set $F' \coloneqq F$.
Now we assume there is an instance of~$(\str{f},\str{\bar{f}})$ in $F$.
We further split into subcases.
Look at the two $\str{g}$/$\str{\bar{g}}$ symbols adjacent to the $(\str{f},\str{\bar{f}})$ pair in $F$.
If the two $\str{g}$ symbols is a pairing in $F$, then we can again set $F' \coloneqq F$.
If exactly one of the $\str{g}$/$\str{\bar{g}}$ adjacent to the $(\str{f},\str{\bar{f}})$ pair is the paired with another $\str{\bar{g}}$ or $\str{g}$ that is not adjacent to the $(\str{f},\str{\bar{f}})$ pair in~$w$,
then switching the cables na\"ively results in a linked pair,
    \[
    w = [\ldots \underbracket{\str{f} \str{g} \ldots \overbracket{\str{g} \ldots \str{\bar{g}}}\str{\bar{f}}}\ldots] 
    \longrightarrow  
    w' = [\ldots\str{g} \rlap{$\underbracket{\phantom{\str{f}\ldots \str{g} \ldots \str{\bar{f}}}}$}\str{f}\ldots
    \overbracket{\str{g} \ldots \str{\bar{f}}\str{\bar{g}}}\ldots].
    \]
Instead, after the cables are switched, we pair the $\str{g}$ and $\str{\bar{g}}$ that appear next to the $(\str{f},\str{\bar{f}})$ pair together in $F'$:
\[
    w' = [\ldots  \underbracket{\str{g} \overbracket{\str{f} \ldots \str{g} \ldots \str{\bar{f}}} \str{\bar{g}}}\ldots].
\]
The new $F'$ is a folding without linked pairs and does not change the area because the same number and types of symbols are paired.
Finally, if both $\str{g}$ and~$\str{\bar{g}}$ adjacent to the $(\str{f},\str{\bar{f}})$ pair are paired with letters not adjacent to $\str{\bar{f}}$ 
and~$\str{f}$ in $F$,
we pair $\str{g}$ next to the $\str{f}$ with the $\str{\bar{g}}$ next to $\str{\bar{f}}$ and the $\str{g}$ not adjacent to the~$\str{f}$ with the $\str{\bar{g}}$
not adjacent to $\str{\bar{f}}$ in $F'$:
   \[
   w = [\ldots \underbracket{\str{f} \overbracket{\str{g} \ldots \str{\bar{g}}}\ldots \overbracket{\str{g} \ldots \str{\bar{g}}}\str{\bar{f}}}\ldots] 
   \longrightarrow
   w' = [\ldots \underbracket{\str{g} \overbracket{\str{f} \ldots \underbracket{\str{\bar{g}}\ldots \str{g}} \ldots \str{\bar{f}}}\str{\bar{g}}}\ldots].
   \]
Again, $F'$ remains a folding without linked pairs and has the same area as $F$. 
\end{proof}

% As a corollary we have the following.

%\paragraph*{Converting between cable systems with the same ordering.}

Next, we show that the cancellation norm does not depend on the
cables having the shortest path property.
While the face word is unique up to isotopy of the cables based on \lemref{isotopy}, two cable systems with identical ordering around~$p_0$ may not be isotopic to each other.
The theory of mapping class groups~\cite{primer,martelli2016}
provides tools to convert between all possible isotopy classes of cable systems with same cable ordering.

% \brittany{it seems we've transitioned into background on mapping class groups.
% Should this be in the background?}
% \brad{The mapping class group ideas are not used until here, I
% worry the reader will have forgotten about them by the time they get here}
Let \EMPH{$S_{g,n}$} denote a surface with genus $g$
and $n$ punctures.
Let \EMPH{$\text{Diffeo}^+(S_{g,n})$} denote the group of all orientation-preserving diffeomorphisms from $S_{g,n}\to S_{g,n}$.
Define an equivalence relation~$\sim$ on $\text{Diffeo}^+(S_{g,n})$, where
$\phi\sim \psi$ if there exists an isotopy between $\phi$ and $\psi$.
The \EMPH{mapping class group} of $S_{g,n}$ is the group
\(
\textnormal{MCG}(S_{g,n}) = \textnormal{Diffeo}^+(S_{g,n})\slash \sim.
\)
A simple closed curve or a puncture-to-puncture arc $\lambda$ on a surface is \EMPH{nonseparating} if the surface remains connected after cutting along $\lambda$.
A simple closed curve (or arc) in a surface is \EMPH{essential}
if it is not homotopic to a point and not homotopic to a puncture.

Let $\lambda$ be a simple closed curve in a surface.
We now describe a particular mapping class in $\textnormal{MCG}(S_{g,n})$, a \EMPH{Dehn twist} --- about $\lambda$.
Let$A_\lambda$ be an arbitrarily thin annulus homeomorphic to the product of $\Sp^1$ and the unit interval $I$; i.e. $A_\lambda \cong \Sp^1 \times I$
with $\lambda\cong \Sp^1\times \{0\}.$
The twist is the map $T_\lambda: A_\lambda \to A_\lambda$ where~$T_\lambda(\theta,t)=(\theta+ 2\pi t,t)$~\cite{primer}.
Intuitively, we can think of $T_\lambda$ as acting on any path that crosses $A_\lambda$ by fixing one boundary component of the annulus
and rotating the other boundary component one full revolution, dragging all masses within $A_\lambda$ along.
See \figref{dehn-twist} for an example.
Dehn twists are well-defined as element in $\textnormal{MCG}(S_{g,n})$ and depend only on the isotopy class of $\lambda$ \cite{martelli2016}.

\begin{figure}[h!]
    \centering
    \begin{subfigure}[t]{0.2\textwidth}
        \includegraphics[width=\textwidth]{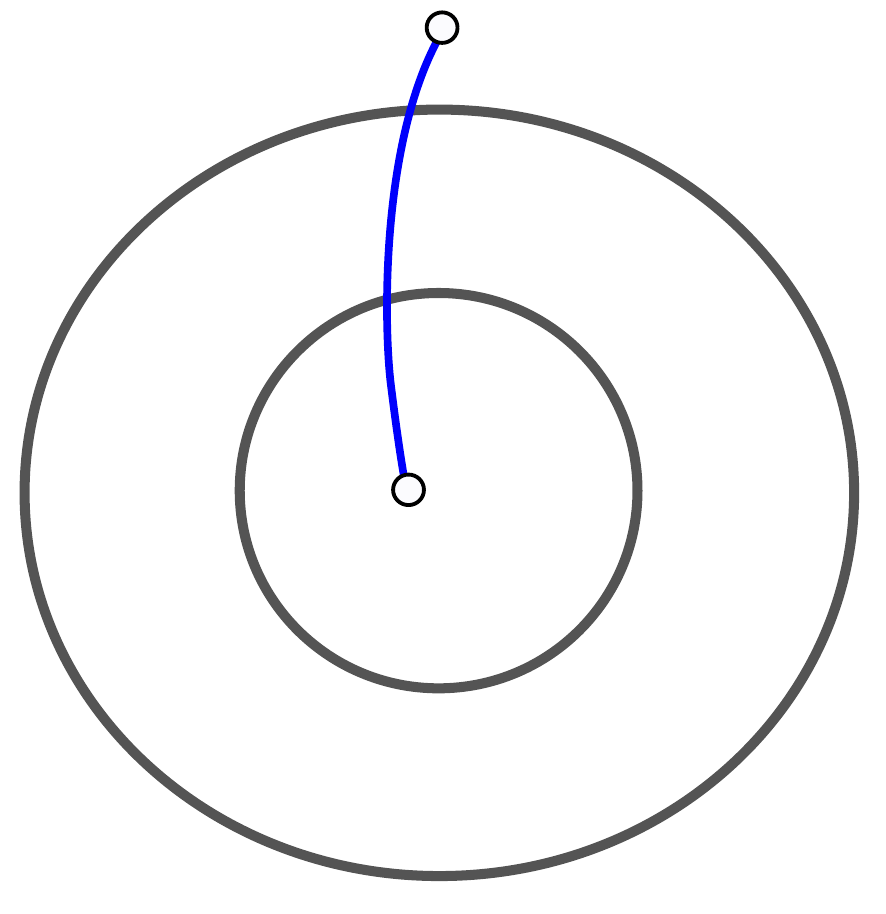}
        % \caption{Before the twist.}
        \label{fig:pre-twist}
    \end{subfigure}
    \hspace{1cm}
    \begin{subfigure}[t]{0.2\textwidth}
        \includegraphics[width=\textwidth]{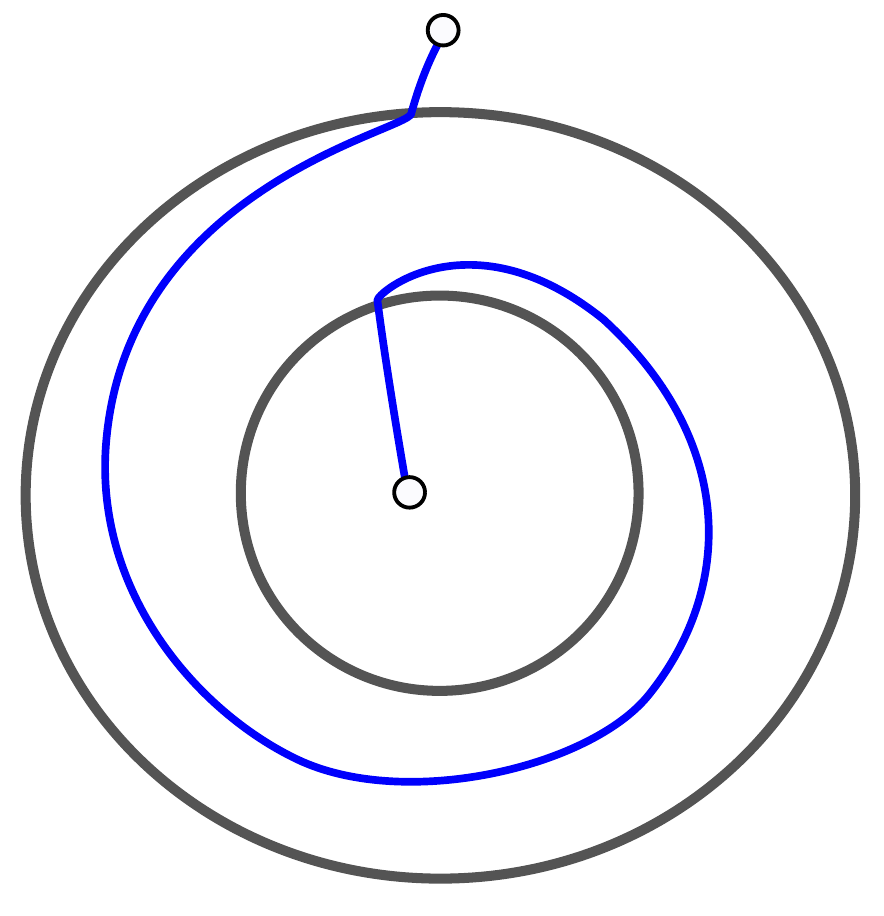}
        % \caption{After the twist.}
        \label{fig:post-twist}
    \end{subfigure}

    \caption{
        An illustration of a Dehn twist.
}
    \label{fig:dehn-twist}
\end{figure}

Let $D_n$ be the closed disk $\mathbb{D}^2$ with $n$ punctures.
The \EMPH{pure mapping class group} of $D_n$, $\textnormal{PMCG}(D_n)$,
is the subgroup of $\textnormal{MCG}(D_n)$ that fixes all punctures.
The group $\textnormal{PMCG}(D_n)$ is generated by a finite number
of closed curves.

\begin{theorem}[Generators~{\cite[\S9.3]{primer}}] \label{thm:generators}
$\textnormal{PMCG}(D_n)$ is generated by Dehn twists
about the set of simple closed curves that surround exactly two punctures.
%\hsien{Careful. Does actual drawing of the closed curves matter?}
\end{theorem}

%Each cable is an essential arc between punctures.
Thus, given any cable system we can perform Dehn twists about these generators to obtain another cable system with the same cable ordering such that all cables are shortest paths.
We choose a disk in the plane that contains the curve $\gamma$, the punctures, and the cables.
(One subtlety is that even when the two cable systems have the same \emph{cyclic} ordering, the choice of the disk may produce different \emph{linear} cable ordering when we decide the location to take the containing disk.  This can be easily resolved by redrawing some of the cables using isotopy across the infinity.)
We now show that performing Dehn twists about suitable simple closed curves does not change the cancellation norm of the word.
%\hsien{Subtlety: even with same cyclic cable ordering, the choice of the disk may give different linear cable ordering.  Urr.  Secondary.}

\begin{lemma}[Twist Invariance]\label{lem:norm-twist}
    Let $\gamma$ be a curve and $\Pi$ be a cable system, and let $F$ be a folding on the corresponding face word $[\gamma](\Pi)$.
    Dehn twists about the cables in $\Pi$ produce another folding on the new word with equal area.
\end{lemma}

\begin{proof}
Let $\Pi$ be a cable system on a closed curve $\gamma$.
Given two punctures $p_i$ and~$p_j$, let $c_{i,j}$
denote the closed curve that contains only
$p_i$ and $p_j$ and follows the cables $\pi_i$ and $\pi_j$
to the exterior face, then connect.
See \figref{pre-twist} for an illustration.
There may be cables
between $\pi_i$ and $\pi_j$ in the rotation order about $p_0$.
We bundle all such cables together and denote the corresponding face word \EMPH{$B$}.
By \thmref{generators}, the closed curves $c_{i,j}$s generate the pure mapping
class group.
To show that twisting about $c_{i,j}$ does not change the norm,
let $w$ be the word generated by the cables before twisting about
$c_{i,j}$ and let $w'$ be the word generated by the cables after the twist.
We will show that $||w'||\leq ||w||$ and~$||w||\leq ||w'||$.

Instances of $\str{f_i}$ and $\str{\bar{f}_i}$ in $w$ become
$(\str{\bar{f}_i}\bar{B}\str{\bar{f}_j}{B}) \str{f_i} (\bar{B}\str{f_j}{B}\str{f_i})$ and~$(\str{\bar{f}_i}\bar{B}\str{\bar{f}_j}{B}) \str{\bar{f}_i} (\bar{B}\str{f_j}{B}\str{f_i})$ in $w'$.
Instances of $\str{f_j}$ and $\str{\bar{f}_j}$ in $w$ become
$(B\str{\bar{f}_i}\bar{B}\str{\bar{f}_j}) \str{f_j} (\str{f_j}B\str{f_i}\bar{B})$ 
and $(B\str{\bar{f}_i}\bar{B}\str{\bar{f}_j}) \str{\bar{f}_j} (\str{f_j}B\str{f_i}\bar{B})$ in $w'$.
See \figref{do-the-twist} for an example.

\begin{figure}[h!]
    \centering
    \begin{subfigure}[t]{0.45\textwidth}
        \includegraphics[width=\textwidth]{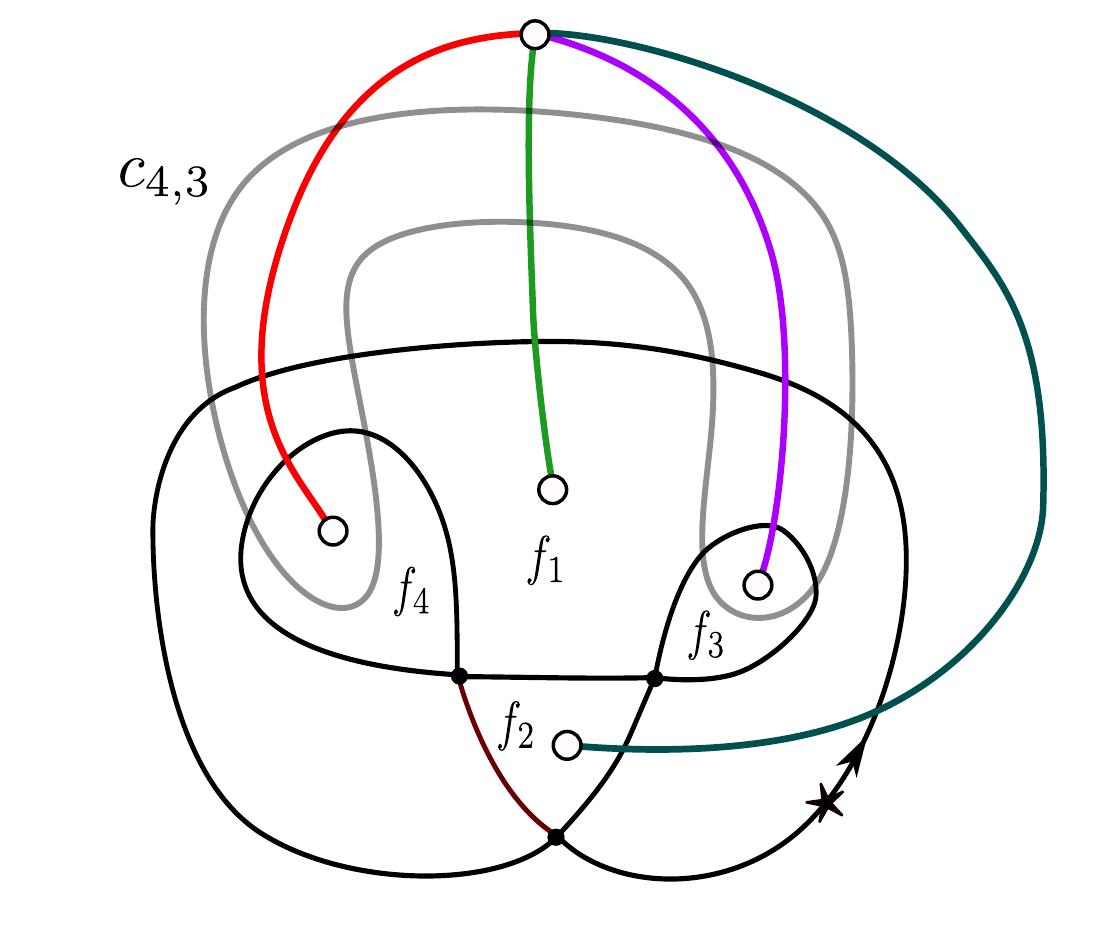}
            \subcaption{Before the twist.}\label{fig:pre-twist}
    \end{subfigure}
    \hspace{1cm}
    \begin{subfigure}[t]{0.45\textwidth}
        \includegraphics[width=\textwidth]{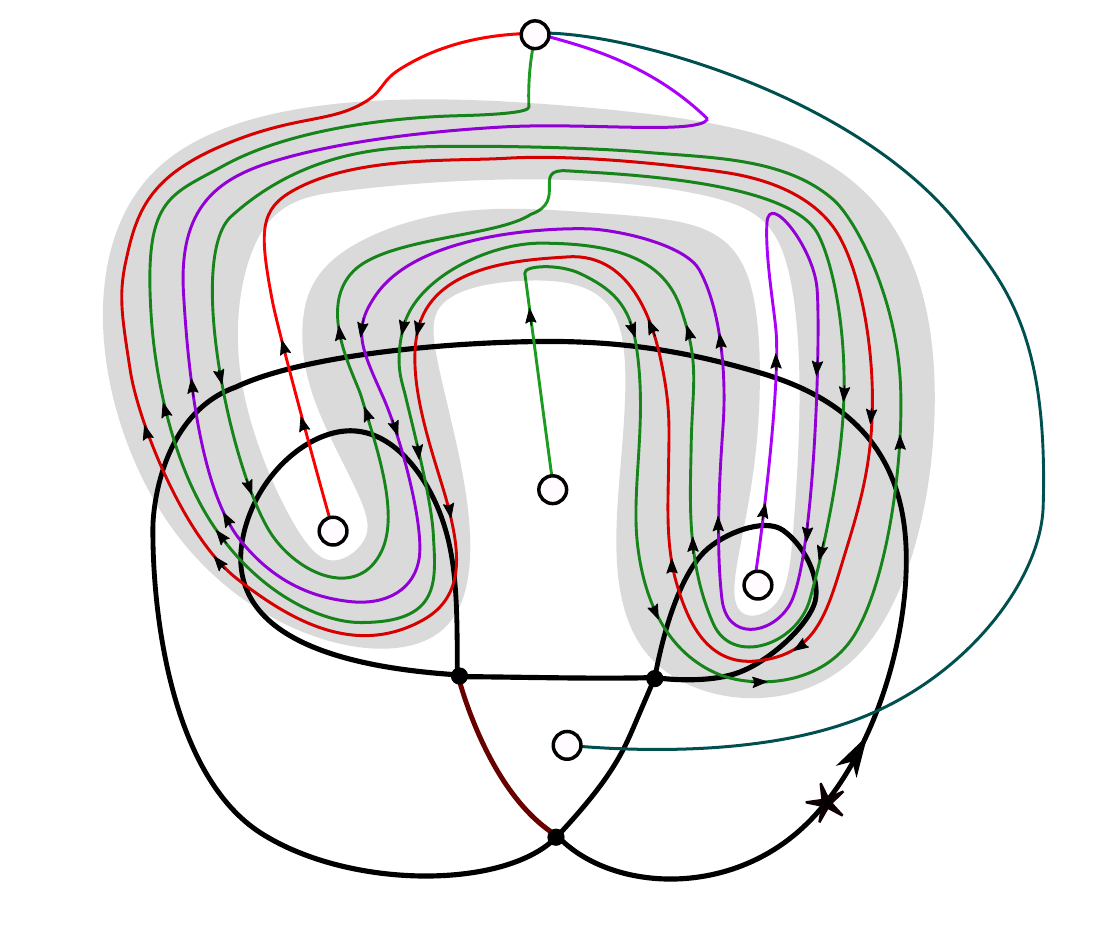}
        \subcaption{After the twist.}\label{fig:post-twisted}
    \end{subfigure}

    \caption{
    (\subref{fig:pre-twist})
    An example of the Dehn twist.
    The curve is shown in black.
    Starting from the star and following the indicated orientation, the Blank word is $w = [\str{23142\overline{34}}]$.
    In the instance we have $i=4$ and $j=3$;
    the curve $c_{4,3}$ is shown in gray.
    The word $B$ is $\str{1}$.
    (\subref{fig:post-twisted}) After twisting about an annulus formed by fattening $c_{4,3}$ the new word is
    $w' = [\str{2}\strp{1\overline{413}}\str{3}\strp{314\overline{1}}\str{1}\strb{\overline{413}1}\str{4}\strb{\overline{1}314}\str{2}
    \strp{1\overline{413}}\str{\overline{3}}\strp{314\overline{1}}\strb{\overline{413}1}\str{\overline{4}}\strb{\overline{1}314}]$.
    }
    \label{fig:do-the-twist}
\end{figure}

We now construct a folding $F'$ of $w'$ with equal area.
All pairs in $F$ are in~$F'$. All letters added by the twist are then paired.
If $\str{f_i}$ is in $w$ and unpaired in $F$ then we pair the conjugate letters
on each side of $\str{f_i}$ in $w'$.
If $\str{f_i}$ is in $w$ and paired in $F$ then
we pair the added letters to the right of $\str{f_i}$ with the added letters to the left of the paired $\str{\bar{f}_i}$ in $w'$;
similarly, we pair the added letters to the left of $\str{f_i}$ with the added letters to the right of $\str{\bar{f}_i}$.
This can be done because the added words are the same mirror pair around symbol $\str{f_i}$ and $\str{\bar{f}_i}$.
\[
w' = [\ldots (\str{\bar{f}_i}\bar{B}\str{\bar{f}_j}{B}) \overbracket{\str{f_i} (\bar{B}\str{f_j}{B}\str{f_i})
\ldots (\str{\bar{f}_i}\bar{B}\str{\bar{f}_j}{B}) \str{\bar{f}_i}} (\bar{B}\str{f_j}{B}\str{f_i}) \ldots]
\]
where the overbracket indicates the pairing $(\str{f_i},\str{\bar{f}_i})$.
Thus, a folding of $w'$ of equal area exist.  Because the cancellation norm is computed by taking minimum over all foldings, we have $||w'||\leq ||w||$.

On the other hand, given a folding $F'$ of $w'$ we construct a folding $F$ of $w$ with area at most the area of $F'$.
If an element of $w$ is unpaired in $F'$ it is also unpaired in $F$.
There are three types of pairings in $F'$.
If both letters in the pair are in $w$: in which case we add this pair to $F$.
If neither letter in the pair is in~$w$, in which case we do nothing to $F$.
The third type of pairing contains a letter of $w$ paired with a letter
in that only presents in $w'$ but not in $w$;
call this pair~$(\str{f}^0,\str{\bar{f}}^1)$ where $\str{f}^0$ is in $w$ and
$\str{\bar{f}}^1$ is in $w'\backslash w$.
We will show that there is a unique unpaired letter $\str{f}^k \in w'\backslash w$
that is left unpaired in $F$,
%contributes the area of face corresponding to $\str{f}^0$,
or there is another~$\str{\bar{f}}^k\in w$
that we can pair $\str{f}^0$ with in $F$.

We construct a sequence of letters in $w'$ that are all $\str{f}$ and $\str{\bar{f}}$s and terminates with either an unpaired letter in $w'\backslash w$
or a paired letter in $w$.
Let $P$
%$P:\set{\text{letters}}\to \set{\text{letters}}$
denote the function
that maps a letter to its paired letter in $F'$, or the identity otherwise.
%Letters in $w'\backslash w$ come in inverse pairs as conjugates.
The map $C$
%$C:\set{\text{letters}}\to \set{\text{letters}}$
maps a letter of $w'\backslash w$ to its conjugate inverse in the mirror pair added by the Dehn twist if the conjugate inverse has not yet been visited, or the identity otherwise.
Note that both $P$ and $C$ are injective, and they only map a symbol $\str{f}$ to its inverse $\str{\bar{f}}$ and vice versa.
Beginning with the pairing $(\str{f}^0,\str{\bar{f}}^1)$, we alternate in applying the functions $P$ and $C$: for any integer $i$, let
$\str{f}^{2i} \coloneqq C(\str{f}^{2i-1})$ and
$\str{f}^{2i+1} \coloneqq P(\str{f}^{2i})$.
% \[
% \begin{aligned}
% \str{f}^{2i} \coloneqq C(\str{f}^{2i-1}), \\
% \str{f}^{2i+1} \coloneqq P(\str{f}^{2i}).
% \end{aligned}
% \]
Since the word is finite, this unique sequence terminates with
either an unpaired letter $\str{f}^k$ of $w'\backslash w$
%with weight equal to $\str{f}^0$
or an element $\str{\bar{f}}^k$ in $w$.

\begin{itemize}\itemsep=0pt
\item
If the sequence of letters ends with an unpaired letter in $w'\backslash w$,
this letter uniquely corresponds to $\str{f}^0$ by following pairings
and corresponding inverses in $F'$.
For example,
let $F'$ be the folding indicated by the overbrackets,
with the underbrackets denoting corresponding conjugate inverse~letters,
\[
[\str{2}\strp{1\overline{413}}
\rlap{$\overbracket{\phantom{\str{3}\strp{314\overline{1}}\str{1}\strb{\overline{41}}\strb{\overline{3}}}}$}
\str{3}\strp{314\overline{1}}\str{1}\strb{\overline{41}}
\rlap{$\underbracket{\phantom{\strb{\overline{3}} \strb{1}\str{4}\strb{\overline{1}} \strb{3}}}$}
\strb{\overline{3}} \strb{1}\str{4}\strb{\overline{1}}
\rlap{$\overbracket{\phantom{\strb{314}\str{2}\strp{1\overline{41}} \strp{\overline{3}}}}$}
\strb{314}\str{2}\strp{1\overline{41}}
\rlap{$\underbracket{\phantom{\strp{\overline{3}}\str{\overline{3}}\strp{3}}}$}
\strp{\overline{3}}\str{\overline{3}}\strp{3}
\strp{14\overline{1}}\strb{\overline{413}1}\str{\overline{4}}\strb{\overline{1}314}].
\]
The unique unpaired letter corresponding to $\str{3}$ is $\strp{3}$
found by considering the second corresponding inverse letter.

\item
If the sequence of letters ends with a letter $\str{f}^k$ in $w$,
we pair $\str{f}^0$ and $\str{\bar{f}}^k$ in $F$.
This does not create any linked pairs in $F$ since
$\str{f}^0$ and $\str{\bar{f}}^k$ participate in pairings in $F'$,
they are not linked by any other pairing in $F'$.
For example,
let $F'$ be the folding indicated by the overbrackets,
with the underbrackets denoting corresponding conjugate inverse letters,
\[
[\str{2}\strp{1\overline{413}}
\rlap{$\overbracket{\phantom{\str{3}\strp{314\overline{1}}\str{1}\strb{\overline{41}}\strb{\overline{3}}}}$}
\str{3}\strp{314\overline{1}}\str{1}\strb{\overline{41}}
\rlap{$\underbracket{\phantom{\strb{\overline{3}} \strb{1}\str{4}\strb{\overline{1}} \strb{3}}}$}
\strb{\overline{3}} \strb{1}\str{4}\strb{\overline{1}}
\rlap{$\overbracket{\phantom{\strb{314}\str{2}\strp{1\overline{413}} \str{\overline{3}}}}$}
\strb{314}\str{2}\strp{1\overline{413}}\str{\overline{3}}\strp{3}
\strp{14\overline{1}}\strb{\overline{413}1}\str{\overline{4}}\strb{\overline{1}314}].
\]
The pairing $(\str{3,\overline{3}})$ is then added to $F$.
\end{itemize}

\noindent The sum of the areas of the unpaired letters of $F$ is at most the sum
of the areas of the unpaired letters of $F'$ and thus $||w||\leq ||w'||$.
This proves the lemma.
\end{proof}

We show the invariance for positive foldability

\begin{restatable}
%[Positive fodability is Cable Invariant]
{lemma}{groupabletwist}
%\begin{lemma}[Positive fodability is Cable Invariant]
\label{lem:groupable-twist}
Being positively foldable does not depend on Dehn twists.
\end{restatable}

\begin{proof}
Let $w$ be a word and let $w'$ be the word after a Dehn twist about
the simple closed curve $c_{i,j}$ as in the proof of \lemref{norm-twist}.
Suppose $w$ is positively foldable with folding $F$. 
Then, the folding $F'$
described in \lemref{norm-twist} is a positive folding of $w'$.

On the other hand, suppose $w'$ has a positive folding $F'$.
Then, each element~$\str{\bar{f}}$ of $w'$ is paired with an element $\str{f}$.
If both $\str{\bar{f}}$ and $\str{f}$ are in $w$ then we pair them.
If neither $\str{\bar{f}}$ nor $\str{f}$ are in $w$ then we ignore them.
If one of $\str{\bar{f}}$ and $\str{f}$ are in $w$ and the other is added to $w'$ by the twist,
then, since letters are added to $w'$ in pairs, there is another element of $w$ paired
with an element of $w'\backslash w$. We can pair the two element in $w$
and ignore the two in $w'\backslash w$ because the pairings in $F'$ are not linked.
The resulting pairings give a positive folding of $w$.
\end{proof}

Together, \lemref{norm-switch} and \lemref{norm-twist} imply the cancellation
norm does not depend on the isotopy class of the cables. Moreover,
\lemref{norm-switch} and \lemref{groupable-twist} imply the positive foldability of the word does not depend on the isotopy class of the cables.
One last subtlety: we show that we can handle multiple cables
per face. 
\paragraph*{Handling multiple cables per face.}

When a face $f$ in a subcurve contains multiple cables and punctures, we treat the punctures as distinct when applying \lemref{norm-switch} and \lemref{norm-twist}.
Now, once all the cables follow the cotree shortest paths, consider all cables ending in an arbitrary face $f$.
If they are separated by other cables in the cyclic ordering, we use \lemref{norm-switch} and isotopy to move all the cables not ending in $f$ out of the way until they no longer intersect the face. 
This makes sure that all the cables ending in $f$ are gathered together in the cyclic ordering.
Because they follow the same unique cotree path to the exterior, we can merge them into a single cable, where the corresponding symbols in the word are merged into a single symbol with the combined weight.
Thus, the subwords induced by a face pairing correspond to the subcurves generated by the Blank cut.

\subsection{Compute Min-Area Homotopy from Self-Overlapping Decomp.}
\label{SS:correctness}

\setlength{\intextsep}{0pt}%
\begin{wrapfigure}{r}{0.25\textwidth}
        \begin{center}
	 \includegraphics[width=.25\textwidth]{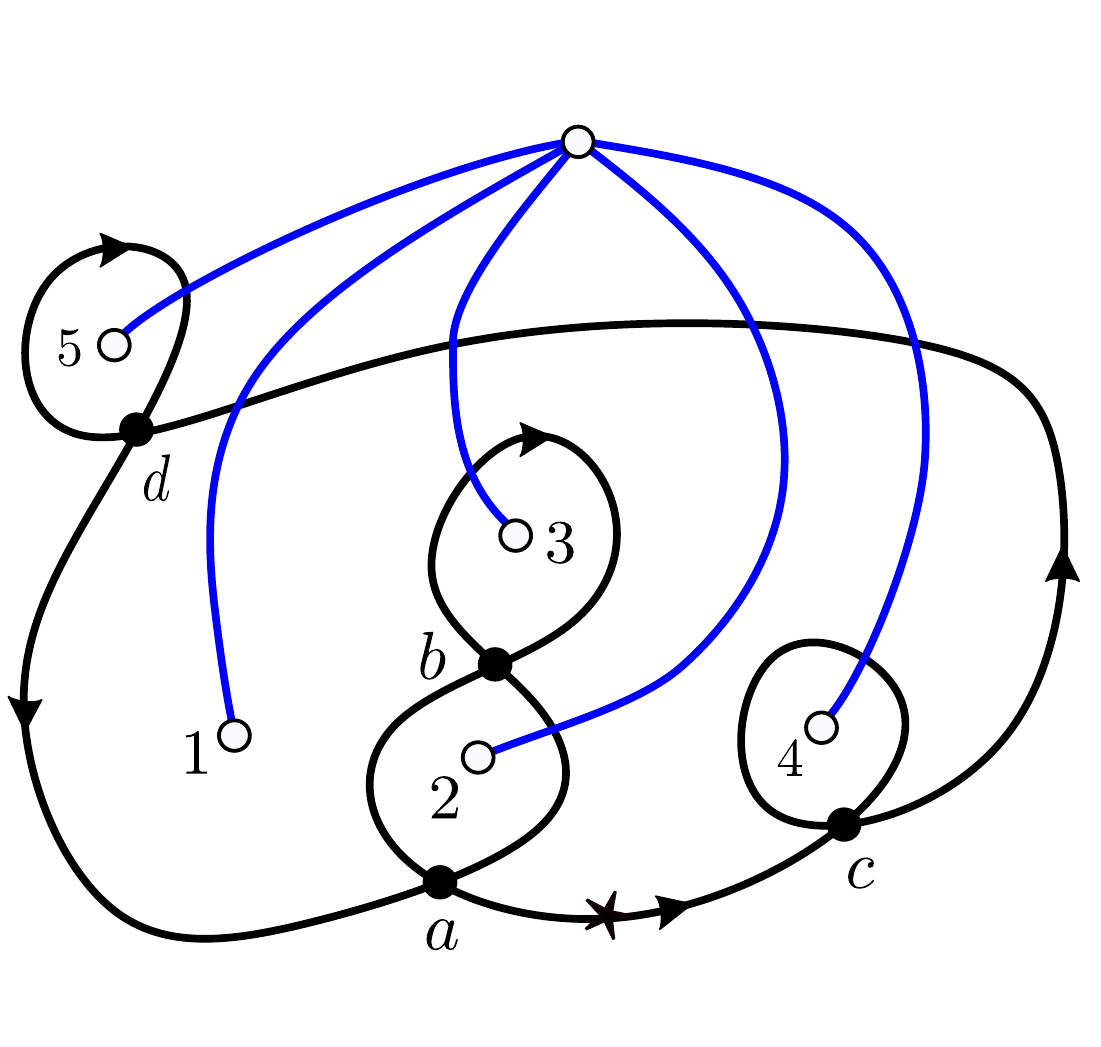}
	 \end{center}
	 \caption{A curve with combined word $[\str{c4c4231d\bar{5}da2b\bar{3}ba}]$.}
	 \label{fig:combined-word}
\end{wrapfigure}

A self-overlapping decomposition is a vertex decomposition where each subcurve is self-overlapping~\cite{fkw2017}.
By \thmref{structural}, there exists a self-overlapping decomposition and an associated homotopy whose area is equal to the minimum homotopy area of the original curve.

%
% \begin{lemma}
% \label{lem:sod-area}
%     \hsien{This is just rephrase of \thmref{structural}.}
%     Any self-overlapping decomposition $\Gamma$ of curve $\curve$ has an
%     associated homotopy $H$ of $\curve$ whose area is equal to the sum of
%     minimum homotopy area of each self-overlapping subcurves in $\Gamma$.
%     % In notation,
%     % \[
%     % \Area_\Gamma(\gamma) = \Area(H).
%     % \]
% \end{lemma}
%
In order to relate vertex decompositions and face decompositions,
we define a word that includes both the faces and vertices.
Given any curve $\gamma$ and cable system~$\Pi$,
traverse $\gamma$ and record both self-crossings and (signed) cable intersections; we call the resulting sequence of vertices and faces the \EMPH{combined word} $[\word](\Pi)$.
See Figure~\ref{fig:combined-word} for an~example.

We now show that every self-overlapping decomposition (with respect to the vertex word of $\gamma$) 
determines a folding (of the face word of $\gamma$) using the combined word.
%The proof is included in \appendref{foldings}.

 \begin{restatable}[S-O Decomp.\ to Folding]{theorem}{sodtofolding}
%\begin{theorem}[S-O Decomp.\ to Folding]
\label{thm:so-to-fold}
Given a self-overlapping decomposition $\Gamma$
%of a curve $\gamma$,
and a cable system $\Pi$ of $\gamma$,
there exists a folding $F$ of $\word(\Pi)$ whose area is
$\Area_\Gamma(\curve)$.
%the area of $\Gamma$ equals to the area of the folding $F$.
%\end{theorem}
\end{restatable}

\begin{proof}
    Begin with the combined word $[\word](\Pi)$.
    Decompose $[\word](\Pi)$ at the vertices given by the self-overlapping decomposition.
    Let $\Gamma=\{\gamma_1,\gamma_2,\ldots, \gamma_s\}$ be the self-overlapping subcurves and $[\word](\Pi)_i$ be the corresponding subwords of~$[\word](\Pi)$.
    If we remove the vertex symbols and turn each $[\word](\Pi)_i$ into a face word $[\gamma_i]'$, such word may not correspond to Blank words of the subcurves; indeed, when decomposing $\gamma$ into subcurves by $\Gamma$,
    the subcurve along with the relevant cables may contain multiple cables per face and cables might not be managed or follow shortest paths.
    See \figref{non-shortest} for an example.
    However, we can first tame the cable system by choosing a new managed cable system $\Pi^*$ where the cables follow shortest paths and has one cable per face (as in \secref{blank}).
    % even with multiple cable and non-shortest path cables, self-overlapping curves
    % have foldable words.
    %\note{Rewrite; claiming that cable systems can be tamed.}
    %
    % When a face in a subcurve contains more than one
    % cable, we perform homotopy to move one cable to the other, and identify the two cables by
    % summing the corresponding areas of the two faces of $\gamma$ together.
    \lemref{independence} ensures that the cancellation norm and positive foldability of the subcurve remain unchanged.
    Denote the new face word of $\gamma_i$ with respect to $\Pi^*$ as $[\gamma_i] = [\gamma_i](\Pi^*)$.

    % This can be done because there exists a isotopy class where the
    % two cables follow the same path to the exterior face.
    % To see this, let $\pi_i$ and $\pi_j$ be two cables from
    % the same face in a subcurve. Fix $\pi_i$ and have $\pi_j$ follow
    % the same path to the exterior face.
    % If there are cables between the two paths, twist them around the puncture
    % $p_j$ so that the cable ordering around $p_0$ is preserved.
    % Thus, we can consider the two letters are representing a single face.
    % By \lemref{groupable-twist} twisting  does not change the groupability of the word.
    % Thus, we can consider $\pi_i$ and $\pi_j$ as a single cable.
\setlength{\intextsep}{12pt}%
    \begin{figure}[htb]
        \captionsetup[subfigure]{justification=centering}
        \centering
        \begin{subfigure}[b]{0.21\textwidth}
                \includegraphics[width=\textwidth]{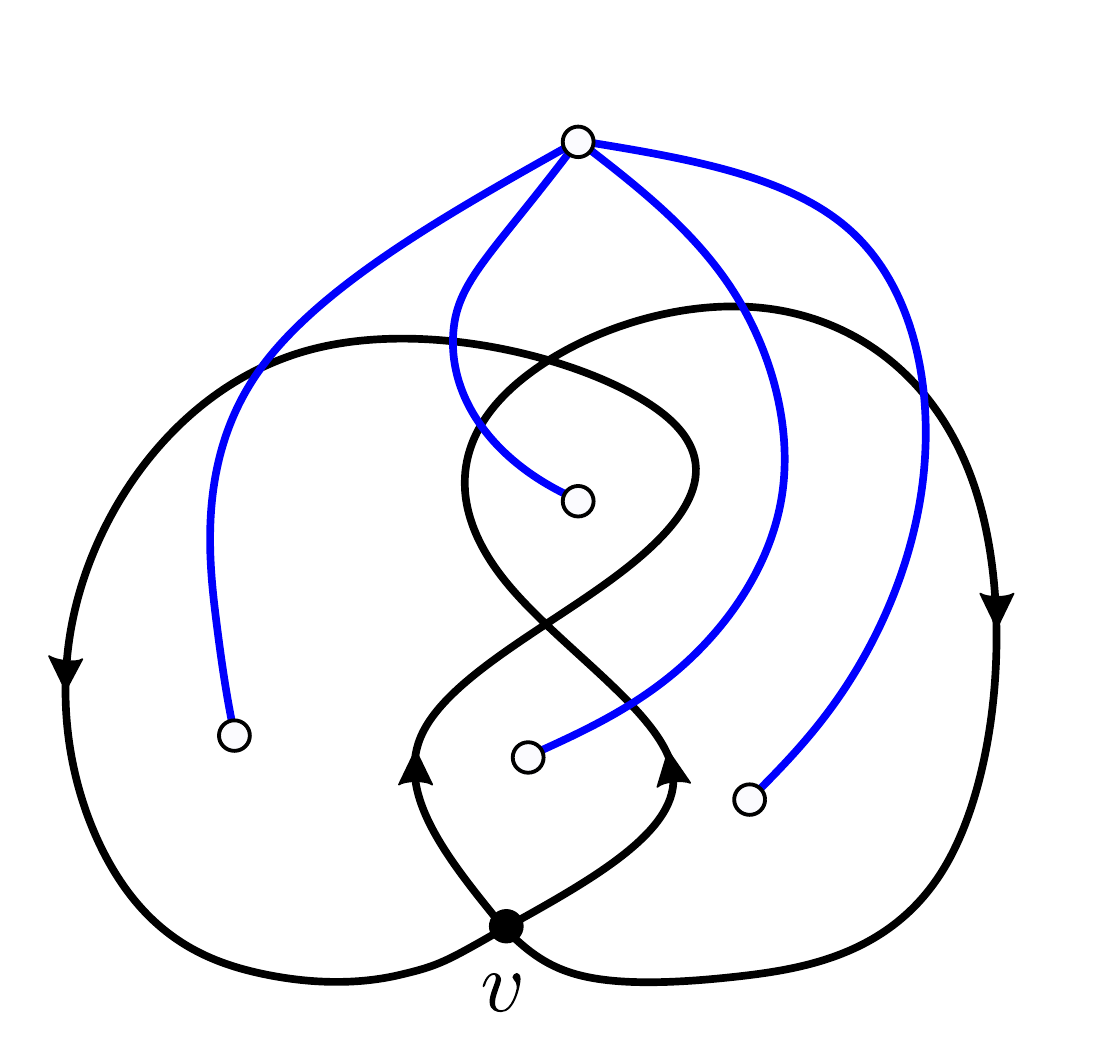}
                \subcaption{}\label{fig:pre-smoothing}
        \end{subfigure}
        \hspace{.5cm}
        \begin{subfigure}[b]{0.21\textwidth}
            \includegraphics[width=\textwidth]{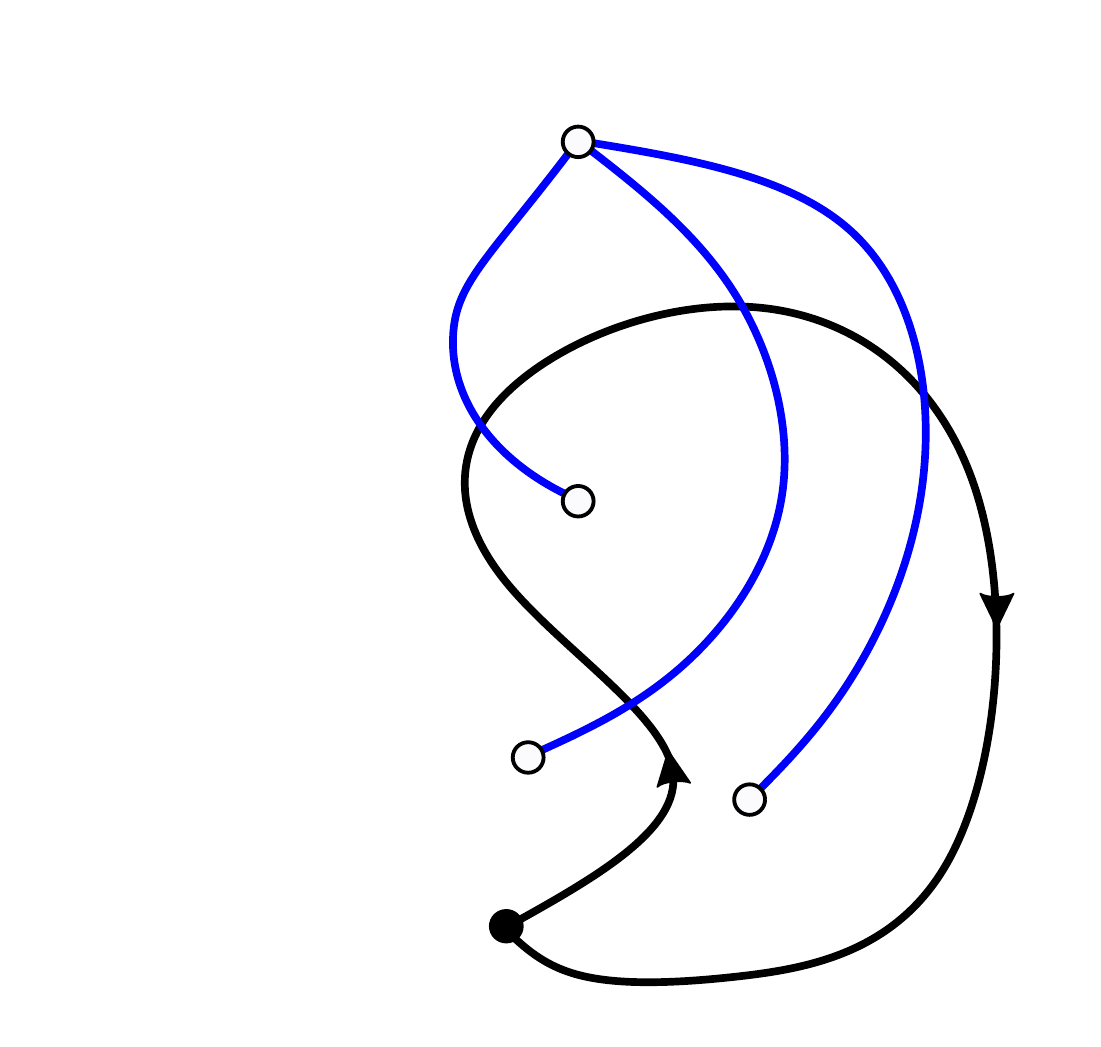}
            \subcaption{}\label{fig:post-smoothing}
        \label{fig:not-short}
        \end{subfigure}
        \caption{We decompose the curve in (\subref{fig:pre-smoothing}) at vertex $v$ into self-overlapping subcurves,
        the cable system  on the induced subcurve in (\subref{fig:post-smoothing}) has more
        than one marked point in a face and cables do not follow shortest
        paths.
        \label{fig:non-shortest}}
    \end{figure}

    Since each $\gamma_i$ is a self-overlapping subcurve in $\Gamma$,
    we can find a positive folding~$F_i$ of $[\gamma_i]$ by Theorem~\ref{thm:blank},
    and the minimum homotopy area of $\gamma_i$ is equal to the area of folding $F_i$.
    Now Lemma~\ref{lem:independence} implies that the subword $[\gamma_i]'$ from the original combined word also has a positive folding $F_i'$ whose area is equal to the minimum homotopy area of $\gamma_i$.
    By combining all foldings $F_i'$ of each face subword~$[\gamma_i]'$, we create a folding $F$ for $[\curve](\Pi)$ (no pairings between different~$F_i'$s can be linked).
    The area of folding $F$ is equal to the sum of areas of foldings~$F_i'$, which in turns is equal to $\sum_i \Area_H(\gamma_i)$, that is, the homotopy area of self-overlapping decomposition $\Area_\Gamma(\curve)$.
    This proves the theorem.
    \end{proof}

\begin{corollary}[Geometric Correctness]\label{cor:correctness}
The dynamic programming algorithm computes the minimum-area homotopy for any curve $\gamma$.
%in polynomial time.
\end{corollary}

\begin{proof}
By \thmref{structural}, there exists a self-overlapping decomposition
with minimum homotopy area.
By \thmref{so-to-fold}, some folding achieves a minimum area.
Using Lemma~\ref{lem:folding-area}, the minimum-area folding produces a minimum-area homotopy.
\end{proof}

%

% \begin{corollary}[Norm is Cable Invariant]
% \label{cor:independence}
% For any curve $\curve$ and a cable system $\Pi$, the minimum area of foldings (the cancellation norm) of $[\curve](\Pi)$ and the existence of a positive folding of $[\curve](\Pi)$ are independent to the choice of $\Pi$.
% \end{corollary}

\subsection{Min-Area Self-Overlapping Decomposition in Polynomial Time}
\label{SS:compute-sod}

Finally, we show how to construct a self-overlapping decomposition 
from a maximal folding with equal area.
Before we begin the proof, we include definitions that
will help describe the types of curves we encounter.
A curve $\curve$ is a \EMPH{$k$-stack} if it has rotation number $k$, all bounded faces have positive winding numbers, and 
$\Area_H(\gamma)=\Area_W(\gamma)$ (see \figref{2stack} for any example).
Any $k$-stack has a vertex decomposition into $k$ self-overlapping curves~\cite[Theorem~2.15]{evans18}.%
\footnote{$k$-stacks are called \EMPH{interior-boundaries} by Titus~\cite{titus} and the terminology was used in various previous work~\cite{marx74,fkw2017,evans18}.  The name \emph{$k$-stack} is well-justified as such curve is the boundary of a stack of $k$ disks~\cite{evans18}.}
A curve is a \emph{$-k$-stack} if its reversal is a $+k$-stack.
We called a curve $\gamma$ a \EMPH{stack} if $\gamma$ is a $\pm k$-stack.
A curve is \EMPH{good}~\cite{evans18}
if the depth of each face is equal to the absolute value of the winding number; any good curve must have $\Area_D(\gamma)=\Area_H(\gamma)=\Area_W(\gamma)$ by Equation~\ref{eq:sandwich}.  
Good curves are almost stacks, except that some faces might not have positive winding numbers.
A vertex of $\gamma$ is \EMPH{sign-changing} if the four incident faces have winding numbers $[1,0,-1,0]$ in cyclic order (see \figref{sign-change} for
an example).
If we smooth a sign-changing vertex of a good curve, the two induced subcurves remain good. 
We can always decompose a good curve into a collection of stacks at all the sign-changing vertices~\cite[Theorem~5.7]{evans18}.
%A \EMPH{block} of a curve is maximally connected subgraph that can not be separated at a sign changing vertex.

\begin{figure}[htb]
        \centering
        \begin{subfigure}[b]{0.3\textwidth}
      		  \includegraphics[width=\textwidth]{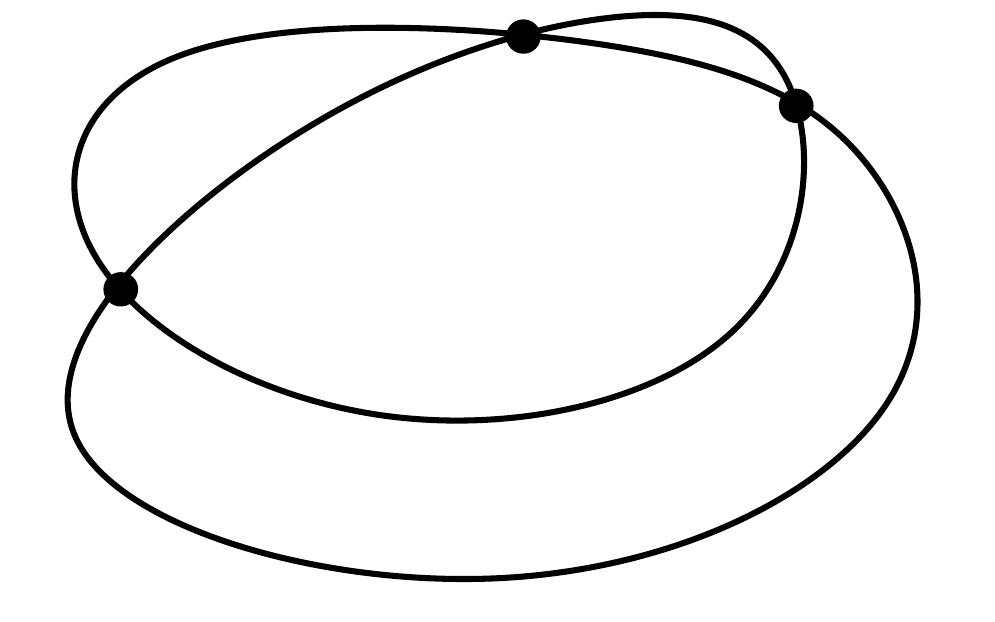}
		   \subcaption{A $2$-stack.}\label{fig:2stack}
        \end{subfigure}
          \hspace{.5cm}
         \begin{subfigure}[b]{0.3\textwidth}
     		   \includegraphics[width=\textwidth]{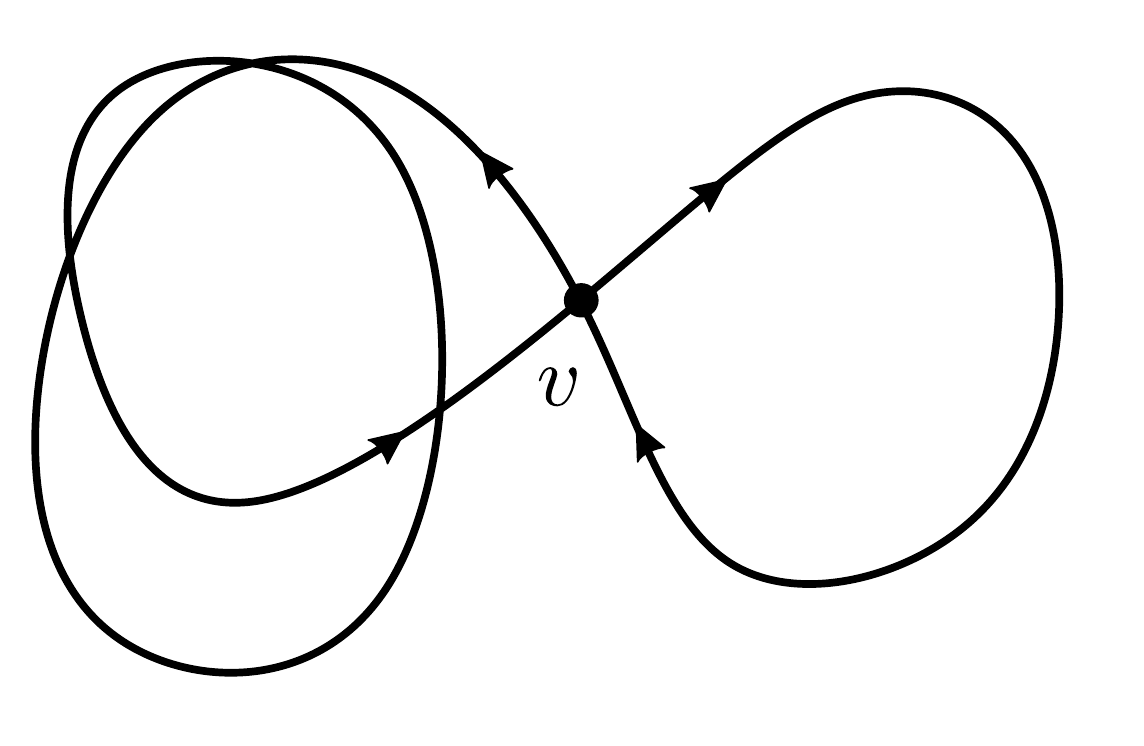}
    		    \subcaption{Sign-changing vertex.}\label{fig:sign-change}
        \end{subfigure}

		\caption{(\subref{fig:2stack}) A $k$-stack where $k=2$. Smoothing at any of the three vertices gives a self-overlapping decomposition.
		(\subref{fig:sign-change}) The vertex $v$ is a sign-changing vertex.
		\label{fig:stack}}
\end{figure}

With these definitions in hand, we now show how to construct a self-overlapping decomposition from a maximal folding with equal area.

\begin{restatable}[Folding to S.O.D.]{theorem}{foldingtosod}
\label{thm:fold-to-so}
Let $\gamma$ be a curve and $\Pi$ be a cable system.
Given a maximal folding $F$ of $\word(\Pi)$, there is a self-overlapping decomposition of $\gamma$ whose area is equal the area induced by the folding $F$.
\end{restatable}

\begin{proof}
%Recall that a self-overlapping decomposition is a collection of vertices of $\gamma$ that are smoothed so that each induced subcurve is self-overlapping.
%Given a folding we construct a self-overlapping decomposition with the desired area.
%
Let $F$ be any maximal folding of $\word(\Pi)$.
Without loss of generality we can assume that the cable system does not cut through the interior of any face by rerouting the cables similar to Section~\ref{SS:norm}.  
By Lemma~\ref{lem:isotopy} the Blank word remains unchanged.
%\hsien{Wait.  The word changes though.}

Let $\good$ be an arbitrary subcurve generated after Blank cutting along the pairings in $F$.
The curve $\good$ is must be good: otherwise, there is a face with winding number not equal to its depth, and thus $\good$ would cross the corresponding cable from left to right and from right to left.
Therefore, we can introduce an extra pair into the folding and $F$ remains unlinked; thus $F$ would not be maximal.
Decomposing at all the sign-changing vertices~\cite[Theorem~5.7]{evans18} turns $\good$ into a collection of stacks, each of which can be further decomposed into collection of self-overlapping curves~\cite[Theorem~2.15]{evans18}.

The area of the folding $F$ is equal to the number of unpaired faces in $\word(\Pi)$,
which is also equal to the sum of depth area of each good subcurve.
Since each subcurve $\good$ of folding is good, the depth area of $\good$ is equal to its winding area.  
Any additional decomposition of $\good$ into self-overlapping curves respects the additivity of winding areas.
Thus, the area of folding $F$ is equal to the area induced by our chosen self-overlapping decomposition.
\end{proof}

The above theorem implies a polynomial-time algorithm to compute a self-overlapping decomposition with minimum area.

\begin{corollary}[Polynomial Optimal Self-Overlapping Decomposition]\label{cor:poly-s-o}
Let $\gamma$ be a curve.
A self-overlapping decomposition of $\gamma$ with area equal to minimum homotopy area of $\gamma$ can be found in polynomial time.
\end{corollary}

\begin{proof}
    Apply the dynamic programming algorithm to compute the minimum-area folding $F$ for $[\curve](\Pi)$ with respect to some cable system $\Pi$.
    By \thmref{so-to-fold} the area of $F$ is equal to the minimum homotopy area of $\gamma$, and so does the corresponding self-overlapping decomposition given by \thmref{fold-to-so}.
\end{proof}

\paragraph{\bf{Acknowledgements}}
{Brittany Terese Fasy and Bradley McCoy are supported by
NSF grant DMS 1664858 and CCF 2046730.
Carola Wenk is supported by NSF grant CCF 2107434 .}
\newpage
\small

\bibliography{references}

\begin{thebibliography}{10}

\bibitem{blank}
Samuel~Joel Blank.
\newblock {\em Extending Immersions and Regular Homotopies in Codimension 1}.
\newblock PhD thesis, Brandeis University, May 1967.

\bibitem{cancellation-norm-def}
Michael Brandenbursky, \'{S}wiatosław Gal, Jarek K\c{e}dra, and Micha\l{l}
  Marcinkowski.
\newblock The cancelation norm and the geometry of bi-invariant word metrics.
\newblock {\em Glasg. Math. J.}, page 153–176, 2015.

\bibitem{bringmann_truly_2019}
Karl Bringmann, Fabrizio Grandoni, Barna Saha, and Virginia~Vassilevska
  Williams.
\newblock Truly subcubic algorithms for language edit distance and {RNA}
  folding via fast bounded-difference min-plus product.
\newblock {\em SIAM J. Comput.}, 48(2):481--512, 2019.

\bibitem{cw2013}
Erin Chambers and Yusu Wang.
\newblock Measuring similarity between curves on 2-manifolds via homotopy area.
\newblock {\em 29th ACM Symp. Comput. Geom.}, pages 425--434, 2013.

\bibitem{changErickson17}
Hsien-Chih Chang and Jeff Erickson.
\newblock Untangling planar curves.
\newblock {\em Discrete Comput. Geom.}, 58:889, 2017.

\bibitem{frisch10}
{Dennis Frisch}.
\newblock Extending immersions into the sphere.
\newblock 2010.
\newblock URL: \url{http://arxiv.org/abs/1012.4923}.

\bibitem{eppsteinMumford}
David Eppstein and Elena Mumford.
\newblock Self-overlapping curves revisited.
\newblock {\em 20th ACM-SIAM Symp. Discrete Algorithms}, pages 160--169, 2009.

\bibitem{erickson-note}
Jeff Erickson.
\newblock One-dimensional computational topology lecture notes.
\newblock Lecture 7, 2020.
\newblock URL:
  \url{https://mediaspace.illinois.edu/channel/CS+598+JGE+—\%C2\%A0Fall+2020/177766461/}.

\bibitem{evans18}
Parker Evans.
\newblock On {{Self}}-{{Overlapping Curves}}, {{Interior Boundaries}}, and
  {{Minimum Area Homotopies}}.
\newblock Bachelor's thesis, Tulane University, 2018.

\bibitem{evansFasyWenk}
Parker Evans, Brittany~Terese Fasy, and Carola Wenk.
\newblock Combinatorial properties of self-overlapping curves and interior
  boundaries.
\newblock {\em 36th ACM Symp. on Comput. Geom.}, 2020.

\bibitem{primer}
Benson Farb and Dan Margalit.
\newblock {\em A Primer on Mapping Class Groups}.
\newblock Princeton University Press, 2011.

\bibitem{fkw2017}
Brittany~Terese Fasy, Sel{\c c}uk Karako{\c c}, and Carola Wenk.
\newblock On minimum area homotopies of normal curves in the plane.
\newblock 2017.
\newblock URL: \url{http://arxiv.org/abs/1707.02251}.

\bibitem{gauss}
Carl~Friedrich Gauss.
\newblock Nachlass. {{I}}. {{Zur}} geometria situs.
\newblock {\em Werke}, vol. 8, 271\textendash 281, 1900.

\bibitem{karakoc17}
Sel{\c c}uk Karako{\c c}.
\newblock {\em On Minimum Homotopy Areas}.
\newblock PhD thesis, Tulane University, 2017.

\bibitem{so-graphics}
Yijing Li and Jernej Barbi{\v c}.
\newblock Immersion of self-intersecting solids and surfaces.
\newblock {\em ACM Trans. on Graph.}, 45:1--14, 2018.

\bibitem{martelli2016}
Bruno Martelli.
\newblock An introduction to geometric topology, 2016.
\newblock \href {http://arxiv.org/abs/arXiv:1610.02592}
  {\path{arXiv:arXiv:1610.02592}}.

\bibitem{marx74}
Morris~L. Marx.
\newblock Extensions of normal immersions of {{S}}{$^1$} into {{R}}.
\newblock {\em Trans. Amer. Math. Soc.}, 187:309--326, 1974.

\bibitem{mukherjee2014}
Uddipan Mukherjee.
\newblock Self-overlapping curves: {{Analysis}} and applications.
\newblock {\em Comput.-Aided Des.}, 40:227--232, 2014.

\bibitem{nie2014}
Zipei Nie.
\newblock On the minimum area of null homotopies of curves traced twice.
\newblock 2014.
\newblock URL: \url{http://arxiv.org/abs/1412.0101}.

\bibitem{nie-emails}
Zipei Nie.
\newblock Private correspondence.
\newblock 2016.

\bibitem{folding-1980}
Ruth Nussinov and Ann~B. Jacobson.
\newblock Fast algorithm for predicting the secondary structure of
  single-stranded {RNA}.
\newblock {\em Proc. Natl. Acad. Sci. USA}, 77(11):6309--6313, 1980.

\bibitem{poe-eic1-1968}
Valentin Po{\'e}naru.
\newblock Extension des immersions en codimension 1 (d'apr\`es {{Samuel
  Blank}}).
\newblock {\em S\'eminaire N. Bourbaki (1966\textendash 1968)}, 10:473--505,
  1968.

\bibitem{shor-van-wyk}
Peter Shor and Christopher Van~Wyk.
\newblock Detecting and decomposing self-overlapping curves.
\newblock {\em Comput. Geom.: Theory and Applications}, 2:31--50, 1992.

\bibitem{titus}
Charles~J. Titus.
\newblock The combinatorial topology of analytic functions on the boundary of a
  disk.
\newblock {\em Acta Math.}, pages 106(1--2):45--64, 1961.

\bibitem{tou-stsps-1989}
Godfried Toussaint.
\newblock On separating two simple polygons by a single translation.
\newblock {\em Discrete Comput. Geom.}, 4(3):265--278, June 1989.

\bibitem{whitney1937}
Hassler Whitney.
\newblock On regular closed curves in the plane.
\newblock {\em Compos. Math.}, 4:276–284, 1937.

\end{thebibliography}

\newpage
\appendix
\end{document}